\documentclass[11pt]{amsart}
\usepackage[T1]{fontenc}
\usepackage{amssymb,amsmath, amsthm, amsfonts, tikz-cd}

\usepackage{graphicx}
\usepackage{listings}
\usepackage[margin=.75in]{geometry}
\usepackage{lstautogobble}
\usepackage{enumerate}
\usepackage[shortlabels]{enumitem}
\usepackage{thmtools}
\usepackage{thm-restate}
\usepackage{amsthm}
\usepackage{verbatim}

\usepackage{mathtools}
\usepackage{physics}
\usepackage{color}
\usepackage{hyperref}
\usepackage[capitalise]{cleveref}
\usepackage[final]{showlabels}
\usepackage[bottom]{footmisc}
\crefformat{equation}{(#2#1#3)}

\usepackage{float}
\restylefloat{table}

\usepackage{mathrsfs}
\setlist{  
  listparindent=\parindent,
  parsep=0pt,
}

\theoremstyle{plain}
\newtheorem{thm}{Theorem}[section]
\newtheorem{prop}[thm]{Proposition}
\newtheorem{lemma}[thm]{Lemma}
\newtheorem{cor}[thm]{Corollary}

\theoremstyle{definition}
\newtheorem{mydef}[thm]{Definition}
\newtheorem{ex}[thm]{Example}

\newtheorem{remark}[thm]{Remark}

\newtheorem*{sum-invol}{Summary of result 1 - $\mathcal{H}_{n}$ is involution}
\newtheorem*{sum-GP_ham}{Summary of result 2 - GP Hamiltonian flows}

\numberwithin{equation}{section} 

\DeclarePairedDelimiter\ipp{\langle}{\rangle}

\DeclarePairedDelimiter{\paren}{\lparen}{\rparen}

\DeclarePairedDelimiter{\jp}{\langle}{\rangle}

\DeclareMathOperator{\supp}{supp}

\newcommand{\M}{{\mathcal{M}}}

\newcommand{\p}{{\partial}}

\renewcommand{\d}{\delta}

\newcommand{\R}{{\mathbb{R}}}
\newcommand{\C}{{\mathbb{C}}}
\newcommand{\N}{{\mathbb{N}}}

\newcommand{\Z}{{\mathbb{Z}}}

\renewcommand{\P}{{\mathcal{P}}}

\newcommand{\g}{{\mathfrak{g}}}

\newcommand{\f}{\mathfrak{f}}
\newcommand{\Fr}{{\mathfrak{F}}}
\newcommand{\Fbr}{{\bar{\mathfrak{F}}}}

\newcommand{\J}{{\mathbf{J}}}

\newcommand{\I}{{\mathcal{I}}}
\newcommand{\Sc}{{\mathcal{S}}}

\renewcommand{\M}{{\mathcal{M}}}

\newcommand{\tl}{\tilde}

\newcommand{\D}{\Delta}

\newcommand{\ph}{\phantom{=}}
\newcommand{\nn}{\nonumber}

\newcommand{\ol}{\overline}
\newcommand{\ul}{\underline}

\newcommand{\ux}{\underline{x}}
\newcommand{\uy}{\underline{y}}

\newcommand{\ua}{\underline{\alpha}}
\newcommand{\ue}{\underline{\eta}}
\newcommand{\ur}{\underline{r}}

\newcommand{\ep}{\epsilon}
\newcommand{\vep}{\varepsilon}
\newcommand{\al}{\alpha}
\newcommand{\et}{\eta}

\newcommand{\diag}{{\mathrm{diag}}}

\newcommand{\wh}{\widehat}

\let\oldtocsection=\tocsection
 
\let\oldtocsubsection=\tocsubsection
 
\let\oldtocsubsubsection=\tocsubsubsection
 
\renewcommand{\tocsection}[2]{\hspace{0em}\oldtocsection{#1}{#2}}
\renewcommand{\tocsubsection}[2]{\hspace{1em}\oldtocsubsection{#1}{#2}}
\renewcommand{\tocsubsubsection}[2]{\hspace{2em}\oldtocsubsubsection{#1}{#2}}

\begin{document}

\title[Mean-Field Convergence of Point Vortices without Regularity]{Mean-Field Convergence of Point Vortices without Regularity}

\author[M. Rosenzweig]{Matthew Rosenzweig}
\address{  
Massachusetts Institute of Technology\\ 
Department of Mathematics\\
Simons Building, Room 106\\
77 Massachusetts Avenue\\
Cambridge, MA 02139-4307}
\email{mrosenzw@mit.edu}
\begin{abstract}
We consider the classical point vortex model in the mean-field scaling regime, in which the velocity field experienced by a single point vortex is proportional to the average of the velocity fields generated by the remaining point vortices. We show that if at some time, the associated sequence of empirical measures converges in a suitable sense to a probability measure with density $\omega^0\in L^\infty(\R^2)$ and having finite energy, as the number of point vortices $N\rightarrow\infty$, then the sequence converges in the weak-* topology for measures to the unique solution $\omega$ of the 2D incompressible Euler equation with initial datum $\omega^0$, locally uniformly in time. In contrast to previous results \cite{Schochet1996, Serfaty2018mean}, our theorem requires no regularity assumptions on the limiting vorticity $\omega$, is at the level of conservation laws for the 2D Euler equation, and provides a quantitative rate of convergence. Our proof is based on a combination of the modulated-energy method of Serfaty \cite{Serfaty2017} and a novel mollification argument. We contend that our result is a mean-field convergence analogue of the famous theorem of Yudovich \cite{Yudovich1963} for global well-posedness of 2D Euler in $L^1(\R^2)\cap L^\infty(\R^2)$.
\end{abstract}

\maketitle

\section{Introduction}\label{sec:intro}
\subsection{The Point Vortex Model and the 2D Euler Equation}\label{ssec:intro_PVM_Eul}
The classical \emph{point vortex model} is the system of $N$ ordinary differential equations (ODEs)
\begin{equation}
\label{eq:PVM}
\begin{split}
\dot{x}_i(t) &= \sum_{{1\leq j\leq N}\atop {j\neq i}} a_j(\nabla^\perp\g)(x_i(t)-x_j(t))\\
x_i(0) &= x_i^0
\end{split}
\qquad \forall i\in \{1,\ldots,N\},
\end{equation}
where $N$ is the number of vortices, $a_1,\ldots,a_N\in \R\setminus\{0\}$ are the intensities of the vortices, $x_1^0,\ldots,x_N^0\in \R^2$ are the pairwise distinct initial positions, $\g(x) \coloneqq -\frac{1}{2\pi}\ln|x|$ is the two-dimensional (2D) Coulomb potential, and $\nabla^\perp \coloneqq (-\p_{x_2},\p_{x_1})$ is the perpendicular gradient. This model, which may be seen to be a finite-dimensional Hamiltonian system, goes back to work of Helmholtz \cite{Helmholtz1858} and Kirchoff \cite{Kirchoff1876}. In addition to being an interesting dynamical system in its own right (see \cite{MP1984ln, MP2012book} for a survey of results), the model \eqref{eq:PVM} is an idealization of  2D incompressible, inviscid fluid flow. Indeed, such a fluid is described in vorticity form by the \emph{2D incompressible Euler equation}
\begin{equation}
\label{eq:Eul}
\begin{split}
&\p_t\omega + u\cdot \nabla\omega = 0 \\
&u = (\nabla^\perp\g) \ast \omega \\
&\omega(0,x) = \omega^0(x)
\end{split}
\qquad (t,x)\in [0,\infty)\times\R^2.
\end{equation}
As the vorticity $\omega$ is simply advected by the velocity field of the fluid, one may informally view the flow as a continuum superposition of dynamical point vortices. In fact, the \emph{empirical measure} of the system \eqref{eq:PVM}, defined by
\begin{equation}
\label{eq:EM}
\omega_N(t,x) \coloneqq \sum_{i=1}^N a_i\d_{x_i(t)}(x),
\end{equation}
is a solution to the so-called \emph{weak vorticity formulation} of equation \eqref{eq:Eul}:
\begin{equation}
\label{eq:wvf}
\begin{split}
\int_{\R^2}\omega(t,x)f(t,x)dx &= \int_{\R^2}\omega^0(x)f(0,x)dx +\int_0^t\int_{\R^2}\omega(s,x)(\p_t f)(s,x)dxds\\
&\ph + \frac{1}{2}\int_0^t \int_{(\R^2)^2\setminus\D_2}\omega(t,x)\omega(t,y) (\nabla^\perp\g)(x-y)\cdot [(\nabla f)(s,x) - (\nabla f)(s,y)]dxdyds,
\end{split}
\end{equation}
for any spacetime test function $f$, where $\D_2\coloneqq \{(x,y)\in(\R^2)^2: x=y\}$. Due to the exclusion of the diagonal $\D_2$ in defining the singular integral in the second line, it is evident that the formulation \eqref{eq:wvf} is well-defined for measure-valued solutions of the form \eqref{eq:EM}.

In order to make sense of a solution to the system \eqref{eq:PVM}, the point vortices must not collide. In the repulsive case, the case considered in this article, where $a_1,\ldots,a_N$ are all of the same sign which we take to be $+1$ for simplicity, the system \eqref{eq:PVM} admits (see \cref{thm:PVM_GWP} below) a unique global solution $\ux_N(t)=(x_1(t),\ldots,x_N(t))$ in $C^\infty([0,\infty); (\R^2)^N\setminus \D_N)$, where
\begin{equation}
\label{eq:D_N_def}
\D_N \coloneqq \{(y_1,\ldots,y_N) \in (\R^2)^N : \exists \ 1\leq i\neq j\leq N \ \text{s.t.} \ y_i = y_j \}.
\end{equation}
Similarly, classical solutions to the Euler equation exist globally in time and are unique by work of Wolibner \cite{Wolibner1933}, and weak solutions (see \cref{def:Eul_ws}) exist globally and are unique for initial data in $L^1(\R^2)\cap L^\infty(\R^2)$ by a theorem of Yudovich \cite{Yudovich1963}, which we review in \cref{ssec:pre_WS} with \cref{thm:Yud}. Thus, the dynamics are globally well-posed at both the microscopic level of the point vortices and the macroscopic level of the Euler equation.

In this article, we are interested in the degree to which the ODE system \eqref{eq:PVM} is a particle approximation to the Euler equation, or conversely the degree to which the partial differential equation (PDE) \eqref{eq:Eul} is a continuum approximation to the point vortex model \eqref{eq:PVM}, when the number of vortices $N$ is very large. To mathematically make sense of this question, we consider the so-called \emph{mean-field} scaling regime where the intensities $a_i=1/N$, for every $1\leq i\leq N$, so that the velocity field experienced by a single point vortex is the average of the velocity fields generated by the other point vortices. Based on the observation that the empirical measure $\omega_N$ is a solution to the weak vorticity formulation \eqref{eq:wvf} of the Euler equation, we expect that $\omega_N$ converges to a solution of the Euler equation \eqref{eq:Eul}, as $N\rightarrow\infty$, in the weak-* topology for measures, point-wise in time:
\begin{equation}
\label{eq:MFC}
\forall t\geq 0, \qquad \omega_N(t) \xrightharpoonup[N\rightarrow \infty]{*} \omega(t).
\end{equation}

\subsection{Prior Results}
\label{ssec:intro_PR}
The subject of approximating the solution to the Euler equation \eqref{eq:Eul} using the point vortex system has a long history in mathematical fluid dynamics going under the banner of \emph{point-vortex methods}. These methods go back to at least work of Rosenhead \cite{Rosenhead1931} and Westwater \cite{Westwater1935}, which used the system \eqref{eq:PVM} to study vortex sheet roll-up. Later works \cite{GHL1990,GH1991} by Goodman, Hou, and Lowengrub focused on approximating the characteristics of classical solutions to the Euler equation \eqref{eq:Eul} using a grid-like discretization of the initial data and the system \eqref{eq:PVM}. We mention that these results very strongly rely on the choice of initial data and therefore are not applicable to treating random initial configurations. As it is not our intention to survey  the literature on this topic, we limit our following remarks to those works specifically treating \emph{mean-field convergence} \eqref{eq:MFC}. We refer the reader more broadly interested in point-vortex methods to the books of Marchioro and Pulvirenti \cite{MP2012book} and Newton \cite{Newton2001}, as well as the relatively recent survey of Jabin \cite{Jab2014}. 

To the best of our knowledge, the first work on establishing the mean-field convergence is that of Schochet \cite{Schochet1996}. He showed that for any sequence of initial data $(\ux_N^0)_{N\in\N}$, such that the $\ux_N^0$ have uniform-in-$N$ compact support, mass, and energy and such that the associated sequence of empirical measures $(\omega_N^0)_{N\in\N}$ converges in the weak-* topology for the space $\M(\R^2)$ of signed measures to some $\omega^0 \in \M(\R^2)$, there exists a subsequence $(\omega_{N_k})_{k\in\N}$ which weak-* converges to \emph{some} solution $\omega$ to the formulation \eqref{eq:wvf} with initial datum $\omega^0$, point-wise in time. His proof is based on his refinement \cite{Schochet1995} of Delort's famous result \cite{Delort1991} for the existence--but not uniqueness--of Radon-measure-valued solutions of a definite sign in $H^{-1}(\R^2)$. Schochet's argument exploits a logarithmic gain of integrability to control the number of close point vortices and compactness in time in order to prove convergence of the nonlinear term along a subsequence. We note that due to a lack of uniqueness for the class of measure-valued solutions to equation \eqref{eq:Eul} under consideration, one does not know from Schochet's result that if $\omega^0$ has a bounded density, then $\omega_{N_k}$ converges to the unique solution $\omega\in L^\infty([0,\infty); L^1(\R^2)\cap L^\infty(\R^2))$ with initial datum $\omega^0$.

A number of years later, Serfaty \cite{Serfaty2018mean} provided the first complete result establishing the mean-field convergence of the point vortex model to the Euler equation, which also has the advantage of being quantitative. We mention her result is a special case of a more general result treating mean-field limits of variants of the system \eqref{eq:PVM}, where the Coulomb potential is replaced by a Riesz potential (see \cref{def:RP} below). Serfaty's proof is based on the so-called \emph{modulated-energy method}, which she originally introduced \cite{Serfaty2017} in the context of Ginzburg-Landau vortices. We defer a discussion of this method until \cref{ssec:intro_road}. In contrast to the partial result of Schochet, which only assumes initial weak-* convergence to a measure $\omega^0$, Serfaty's proof requires that the limiting vorticity $\omega$ belong to $L^\infty([0,\infty);\P(\R^2)\cap L^\infty(\R^2))$, where $\P(\R^2)$ denotes the space of probability measures, and satisfy the bound
\begin{equation}
\label{eq:S_assmp}
\sup_{0\leq t\leq T} \|\nabla^2(\g\ast \omega(t))\|_{L^\infty(\R^2)} < \infty, \qquad \forall T>0.
\end{equation}
For prescribed initial datum, such a solution is global and unique. Note that the assumption \eqref{eq:S_assmp} implies that the velocity field $u$ is spatially Lipschitz continuous. Although we have the operator identity
\begin{equation}
\nabla^2(\g\ast) = \frac{\nabla^2}{\D},
\end{equation}
and therefore the left-hand side defines an order-zero Fourier multiplier, it is not bounded on $L^\infty(\R^2)$, as is generally the case for such operators. Thus, the condition \eqref{eq:S_assmp} cannot be ensured by finiteness of $\|\omega(t)\|_{L^\infty(\R^2)}$,\footnote{Even if $\omega$ is compactly supported, the velocity field $u$ need not be Lipschitz.} which is a conserved quantity for sufficiently nice weak solutions to \eqref{eq:Eul}.

Finally, we also mention the work of Duerinckx \cite{Duerinckx2016}, which was published in between Schochet's and Serfaty's respective results. His work treats the mean-field limit of the gradient-flow analogue of \eqref{eq:PVM}, where the $\nabla^\perp$ is replaced by $\nabla$, as well as mixed flows. While the argument breaks down for the point vortex model, several of the ideas behind the analysis in \cite{Duerinckx2016} are relevant to this work.

\subsection{Overview of Main Results}
As we saw in this last subsection, the only complete results proving the desired mean-field convergence of the point vortex system \eqref{eq:PVM} to the Euler equation \eqref{eq:Eul} require the existence of a limiting vorticity distribution $\omega$ satisfying assumptions which are not at the level of conserved quantities (e.g. the $L^p$ norms of $\omega$). Moreover, the Lipschitz norm of the velocity field $u$ in \eqref{eq:Eul}, and by implication the $C^\alpha$ norm of the vorticity $\omega$, in general grows in time.\footnote{See \cite{BC1994} for an example, and see \cite{KS2014} and references therein for more discussion on this point.} Given that the Euler equation is known to be globally well-posed for initial data in $L^1(\R^2)\cap L^\infty(\R^2)$, in particular requiring no regularity on the initial vorticity and at the level of conservation laws, thanks to the aforementioned theorem of Yudovich \cite{Yudovich1963}, a natural question is if one can show mean-field convergence under the same assumptions for the limiting initial vorticity distribution. Thus, we are led to our main results, which affirmatively answer this question.

\begin{thm}[Main result I]
\label{thm:main}
There exists an absolute constant $C>0$ such that the following holds. Let $\omega\in L^\infty([0,\infty);\P(\R^2)\cap L^\infty(\R^2))$ be a weak solution to the Euler equation \eqref{eq:Eul} with initial datum $\omega^0$,\footnote{See \cref{def:Eul_ws} for the notion of weak solution intended here.} such that
\begin{equation}
\label{eq:main_id_con}
\int_{\R^2}\ln\jp{x}\omega^0(x)dx + \int_{(\R^2)^2}\ln\jp{x-y}\omega^0(x)\omega^0(y)dxdy<\infty.
\end{equation}
Let $N\in\N$, and let $\ux_N\in C^\infty([0,\infty);(\R^2)^N\setminus\D_N)$ be a solution to the point vortex system \eqref{eq:PVM}.
Define the functional
\begin{equation}
\label{eq:FN_avg_def}
\Fr_N^{avg}: [0,\infty)\rightarrow\R, \qquad \Fr_N^{avg}(\ux_N(t), \omega(t)) \coloneqq \int_{(\R^2)^2\setminus\D_2} \g(x-y)d(\omega_N-\omega)(t,x)d(\omega_N-\omega)(t,y),
\end{equation}
where $\omega_N$ is the empirical measure defined in \eqref{eq:EM}. If for given $t>0$, $N\in\N$ is sufficiently large so that
\begin{equation}
\label{eq:N_cond}
\begin{split}
&\frac{Ct(\|\omega^0\|_{L^\infty(\R^2)}^{1/2} + \|\omega^0\|_{L^\infty(\R^2)}^{3/2})|\ln N|^2}{N} + |\Fr_N^{avg}(\ux_N(0),\omega(0))| < \exp\paren*{-e^{Ct(\|\omega_0||_{L^\infty(\R^2)}^{1/2}+\|\omega^0\|_{L^\infty(\R^2)}^{3/2}}},
\end{split}
\end{equation}
then $\Fr_N^{avg}(\ux_N(t),\omega(t))$ satisfies the inequality
\begin{equation}
\begin{split}
&\left|\Fr_N^{avg}(\ux_N(t),\omega(t))\right|\\
&\leq \paren*{\left|\Fr_N^{avg}(\ux_N(0),\omega(0))\right| + \frac{Ct(\|\omega^0\|_{L^\infty(\R^2)}^{1/2} + \|\omega^0\|_{L^\infty(\R^2)}^{3/2})|\ln N|^2}{N} }^{e^{-Ct(\|\omega^0\|_{L^\infty(\R^2)}^{1/2}+\|\omega^0\|_{L^\infty(\R^2)}^{3/2})}}.
\end{split}
\end{equation}
\end{thm}

As a corollary of \cref{thm:main}, we obtain that the empirical measures $\omega_N$ converge to $\omega$, as $N\rightarrow\infty$, in the weak-* topology for the Besov space $B_{2,\infty}^{-1}(\R^2)$ (see \cref{def:Bes} for the definition of this function space). Our estimate for this Besov convergence implies norm convergence in the Sobolev space $H^{s}(\R^2)$, for any $s<-1$, and weak-* convergence in the space of signed measures $\M(\R^2)$, locally uniformly in time.

\begin{cor}[Main Result II]
\label{cor:main}
Given $s<-1$, there exist absolute constants $C,C_s>0$ such that the following holds. If $\ux_N$ and $\omega$ satisfy the assumptions of \cref{thm:main} and $N$ is sufficiently large so as to satisfy the condition \eqref{eq:N_cond} with $t=T>0$, then
\begin{equation}
\label{eq:cor_main_est}
\begin{split}
&\|\omega_N-\omega\|_{L^\infty([0,T]; H^{s}(\R^2))} \\
&\leq C_s\paren*{|\Fr_N^{avg}(\ux_N^0,\omega^0)|+\frac{CT(\|\omega^0\|_{L^\infty(\R^2)}^{1/2}+\|\omega^0\|_{L^\infty(\R^2)}^{3/2})|\ln N|^2}{N}}^{e^{-CT(\|\omega^0\|_{L^\infty(\R^2)}^{1/2}+\|\omega^0\|_{L^\infty(\R^2)}^{3/2})}}\\
&\ph + \frac{C_s(|\ln N|^{1/2}+\|\omega^0\|_{L^\infty(\R^2)})}{N^{1/2}}.
\end{split}
\end{equation}
Consequently, if $\Fr_N^{avg}(\ux_N^0,\omega^0)\xrightarrow[N\rightarrow\infty]{}0$, then $\omega_N\xrightharpoonup[N\rightarrow\infty]{*} \omega$ in $\M(\R^2)$, locally uniformly in time.
\end{cor}

Given the requirement that limiting vorticity distribution $\omega$ only belongs to $L^\infty([0,\infty)\times\R^2)$ and due to the similarity between our estimate for $\Fr_N^{avg}(\ux_N(t),\omega(t))$ and the estimate \eqref{eq:Hold} for the H\"older continuity of the flow map associated to $\omega$, we contend that \cref{thm:main} and \cref{cor:main} constitute a mean-field-convergence analogue to the work of Yudovich \cite{Yudovich1963} on global well-posedness for the Euler equation \eqref{eq:Eul}.

Before commenting on the proofs of our two main results, we record some remarks on the assumptions in the statements of the theorem and corollary, respectively, and some of the implications of our results.

\begin{remark}
\label{rem:main_assmps}
The assumption on the initial vorticity $\omega^0$ in \eqref{eq:main_id_con} is two-fold. The finiteness of the first term, which is propagated by the flow (see \cref{lem:log_grow}), is a qualitative, technical condition, not unique to our work, ensuring that the stream  function $\g\ast\omega$ is well-defined. Similarly, the finiteness of the second term ensures by conservation of energy (see \cref{lem:e_con}) that $\omega(t)$ has finite energy, so that $\Fr_N^{avg}(\ux_N(t),\omega(t))$ is well-defined. Additionally, an examination of the proof of \cref{thm:main}, in particular the use of \cref{prop:key}, shows that we do not need to assume that $\omega^0\in L^\infty(\R^2)$. It suffices to have a weak solution in $L^\infty([0,T]; L^1(\R^2)\cap L^p(\R^2))$, for some $p>2$, provided that we have an \emph{a priori} bound on the log-Lipschitz norm of the velocity field $u(t)$ uniformly in $t\in [0,T]$. Of course, this bound is provided if $\omega^0\in L^\infty(\R^2)$ (see \cref{lem:pot_bnds}).
\end{remark}

\begin{remark}
It is possible to obtain from our proof control on $\Fr_N^{avg}(\ux_N(t),\omega(t))$ beyond the time at which this quantity is of size $~1$. However, since such timescales are beyond the mean-field regime, we have not attempted to do so in this article.
\end{remark}

\begin{remark}
\label{rem:marg_pc}
The convergence given by \cref{cor:main} implies \emph{propagation of chaos}: if $f_N^0\in \P_{sym}((\R^2)^N)$ is a symmetric $N$-body probability density function for the initial positions of the point vortices satisfying some technical conditions and if $f_N$ is the solution to the Liouville equation associated to \eqref{eq:PVM} with initial datum $f_N^0= (\omega^0)^{\otimes N}$, then for every $k\in\N$, the $k$-particle marginal $f_N^{(k)}$ of $f_N$ converges to $\omega^{\otimes k}$ locally uniformly in time. Indeed, this convergence follows from the Grunbaum lemma and arguing similarly as to in \cite[Lemma 8.4]{RS2016}.
\end{remark}

\subsection{Road Map of Proofs of Main Results}
\label{ssec:intro_road}
Our proof of \cref{thm:main} is inspired by the modulated-energy method as developed by Duerinckx and Serfaty in the aforementioned works \cite{Duerinckx2016} and \cite{Serfaty2018mean}, respectively, on mean-field limits of Hamiltonian systems with Coulomb and super-Coulombic (i.e. Riesz) potential interactions. The latter work is particularly inspirational to us. This method relies on a weak-strong stability principle for the Euler equation \eqref{eq:Eul}. Its advantages are that it is quantitative and that it avoids the need for an understanding of the microscopic dynamics in terms of point-vortex trajectories. Until now, its drawback has been that it requires some regularity for the limiting vorticity or an assumption on the velocity field (e.g. Lipschitz continuous) for which one does not have a global-in-time bound.

The idea behind approaches based on the modulated-energy method is to consider as modulated energy the quantity
\begin{equation}
\Fr_N^{avg}(\ux_N(t),\omega(t)) \coloneqq \int_{(\R^2)^2\setminus\D_2}\g(x-y)d(\omega_{N}-\omega)(t,x)d(\omega_N-\omega)(t,y),
\end{equation}
previously introduced in the statement of \cref{thm:main}, which may be regarded as a renormalization of the quantity
\begin{equation}
\|\omega_N(t)-\omega(t)\|_{\dot{H}^{-1}(\R^2)}^2,
\end{equation}
so as to remove the infinite self-interaction between the point vortices. One then proceeds by the energy method. Formally (see \cref{lem:en_ineq} for the rigorous computation), the time derivative of $\Fr_N^{avg}$ is given by
\begin{equation}
\label{eq:intro_ME_deriv}
\frac{d}{dt}\Fr_N^{avg}(\ux_N(t),\omega(t)) = \int_{(\R^2)^2\setminus\D_2} (\nabla\g)(x-y)\cdot\paren*{u(t,x)-u(t,y)}d(\omega_N-\omega)(t,x)d(\omega_N-\omega)(t,y)
\end{equation}
for a.e. $t\in [0,\infty)$. Thus, the key challenge is to control the symmetrized expression appearing in the right-hand side, point-wise in time. Supposing for a moment that $u$ is spatially Lipschitz, locally uniformly in time, one can bound with a bit of hard work (see \cite[Proposition 2.3]{Serfaty2018mean}) the absolute value of the right-hand side of \eqref{eq:intro_ME_deriv} to obtain that
\begin{equation}
\begin{split}
\left|\Fr_N^{avg}(\ux_N(t),\omega(t))\right| &\leq \left|\Fr_N^{avg}(\ux_N(0),\omega(0))\right| + C\paren*{\|\omega^0\|_{L^\infty(\R^2)}, \|u\|_{L^\infty([0,T];W^{1,\infty}(\R^2))}} N^{-1/2} \\
&\ph +C\paren*{\|\omega^0\|_{L^\infty(\R^2)}, \|u\|_{L^\infty([0,T];W^{1,\infty}(\R^2))}}\int_0^t |\Fr_N^{avg}(\ux_N(s),\omega(s))|ds,
\end{split}
\end{equation}
for every $t\in [0,T]$, where $C(\cdot,\cdot)$ denotes a constant depending on its two inputs. One then concludes the proof by an application of the Gronwall-Bellman inequality.

Removing the Lipschitz assumption on the velocity field $u$ requires our introduction of several new ideas in order to prove \cref{thm:main}. The overall workhorse of our proof is \cref{prop:key} stated below.
\begin{restatable}{prop}{key}
\label{prop:key}
Assume that $\mu\in\P(\R^2)\cap L^p(\R^2)$ for some $2<p\leq \infty$. Then for any bounded, log-Lipschitz vector field $v:\R^2\rightarrow\R^2$ and vector $\ux_N\in (\R^2)^N\setminus\D_N$, we have the estimate
\begin{equation}
\label{eq:prop_key}
\begin{split}
&\frac{1}{N^2}\left|\int_{(\R^2)^2\setminus\D_2} \paren*{v(x)-v(y)}\cdot(\nabla\g)(x-y)d(\sum_{i=1}^N\d_{x_i}-N\mu)(x)d(\sum_{i=1}^N\d_{x_i}-N\mu)(y)\right|\\
&\lesssim_p \|v\|_{LL(\R^2)}|\ln\ep_2| |\Fr_N^{avg}(\ux_N,\mu)| + \frac{\|v\|_{LL(\R^2)}|\ln\ep_1| |\ln\ep_2|}{N} + C_p\|v\|_{LL(\R^2)}\ep_3^{\frac{2(p-1)}{p}}|\ln\ep_3| \\
&\ph +\frac{\ep_1 \|v\|_{L^\infty(\R^2)}}{\ep_3^2}+ \frac{\ep_2|\ln\ep_2|\|v\|_{LL(\R^2)}}{\ep_3} + \ep_2|\ln\ep_2| C_p\|v\|_{LL(\R^2)}\|\mu\|_{L^p(\R^2)}^{\frac{p}{2(p-1)}} \\
&\ph + \|v\|_{L^\infty(\R^2)}\paren*{C_p\|\mu\|_{L^p(\R^2)}\ep_1^{\frac{p-2}{p}}+ C_\infty \|\mu\|_{L^\infty(\R^2)}\ep_1 |\ln\ep_1| 1_{\geq \infty}(p)}
\end{split}
\end{equation}
for all parameters $(\ep_1,\ep_2,\ep_3)\in (\R_+)^{3}$ satisfying $0<2\ep_1 < \ep_2 < \ep_3\ll 1$. Here, $\|\cdot\|_{LL(\R^2)}$ denotes the log-Lipschitz semi-norm defined in \cref{def:LL} and $C_p,C_\infty>0$ are constants.
\end{restatable}

At a high level, the proof of \cref{prop:key} is inspired by that of the aforementioned \cite[Proposition 2.3]{Serfaty2018mean} and is divided into five steps, the first of which is completely novel to our work and the fourth and fifth of which require significant new analysis compared to \cite{Duerinckx2016,Serfaty2018mean}:
\begin{enumerate}[(S1)]
\item\label{item:intro_moll}
Mollification,
\item\label{item:intro_renorm}
Renormalization,
\item
Analysis of Diagonal Terms,
\item\label{item:intro_recom}
Recombination,
\item\label{item:intro_conc}
Conclusion.
\end{enumerate}
The proof of \cref{prop:key} is quite long and technical, so as to maintain the accessibility of the introduction, we defer an in-depth overview of the proof to the beginning of \cref{sec:kprop}. Instead, we only comment that in our new \ref{item:intro_moll}, we replace $v$ by the mollified vector field $v_{\ep_2} \coloneqq v\ast \chi_{\ep_2}$, where $\chi_{\ep_2}(x)\coloneqq \ep_2^{-2}\chi(\ep_2^{-1}x)$ is a standard approximate identity. Using the log-Lipschitz regularity of $v$, we can get a quantitative bound for the error introduced in this replacement. We then proceed to \ref{item:intro_renorm}-\ref{item:intro_conc}, working with $v_{\ep_2}$. Throughout the analysis, we avoid any estimates that require us to put $\nabla v_{\ep_2}$ in $L^\infty(\R^2)$ (i.e. take the Lipschitz semi-norm of $v_{\ep_2})$; however, there is one point in \ref{item:intro_conc} (see \eqref{eq:v_ep_Lip}) where we cannot avoid doing so. It is as this point that the mollification is needed, as it allows us to estimate $\|\nabla v_{\ep_2}\|_{L^\infty(\R^2)}$ at the cost of a factor of $|\ln\ep_2|$. It is worth mentioning that the inspiration for this mollification step comes from previous work of the author on the mean-field limit of the Lieb-Liniger model \cite{Rosenzweig2019_LL},\footnote{As is perhaps well-known, the mean-field limit of the Lieb-Liniger model is the one-dimensional cubic nonlinear Schr\"odinger equation.} a one-dimensional model for bosons interacting via the $\delta$-potential, where an $N$-dependent mollification of the interaction potential was combined with the partial H\"older regularity of the $N$-body wave function. This mollification argument appears quite effective, and we expect it to have further application to problems of mean-field convergence.

With \cref{prop:key} in hand, we return to estimating the expression in the right-hand side of the identity \eqref{eq:intro_ME_deriv} for the derivative of the modulated energy. We apply \cref{prop:key} to the right-hand side point-wise in time, allowing $\ep_1(t),\ep_2(t),\ep_3(t)$ to depend on both $t$ and $N$ through $|\Fr_N^{avg}(\ux_N(t),\omega(t))|$ (see \eqref{eq:ep1_choice}, \eqref{eq:ep2_choice}, \eqref{eq:ep3_choice} for the precise choices). In particular, the rate of mollification is dictated by the size of the modulated energy. Thus, we find with a bit of work that for all $N$ sufficiently large depending on $(t,\Fr_N^{avg}(\ux_N(0),\omega(0)))$,
\begin{equation}
\begin{split}
\left|\Fr_N^{avg}(\ux_N(t),\omega(t))\right| &\leq \left|\Fr_N^{avg}(\ux_N(0),\omega(0))\right| + \frac{C(\|\omega^0\|_{L^\infty(\R^2)})t|\ln N|^2}{N} \\
&\ph + C(\|\omega^0||_{L^\infty(\R^2)})\int_0^t |\Fr_N^{avg}(\ux_N(s),\omega(s))| |\ln|\Fr_N^{avg}(\ux_N(s),\omega(s))||ds,
\end{split}
\end{equation}
where $C(\|\omega^0\|_{L^\infty(\R^2)})>0$ is a constant depending only $\|\omega^0\|_{L^\infty(\R^2)}$. Due to the logarithmic factor in the third term in the right-hand side of the preceding inequality, we cannot use the Gronwall-Bellman inequality to solve this integral inequality. Instead, we rely on the Osgood lemma (see \cref{lem:Os}), which allows us to complete the proof of \cref{thm:main}.

Obtaining \cref{cor:main} from \cref{thm:main} is comparatively straightforward using \cref{prop:bes_conv}. This latter proposition shows that the modulated energy $\Fr_N^{avg}(\ux_N(t),\omega(t))$ controls the Sobolev norm $\|\omega_N(t)-\omega(t)\|_{H^s(\R^2)}$, uniformly in time, up to an error which is $O((\ln N)^{1/2}/N^{1/2})$. The weak-* convergence $\omega_N(t)\xrightharpoonup[]{*}\omega(t)$, as $N\rightarrow\infty$, in $\M(\R^2)$ then follows from a standard approximation argument.

\subsection{Organization of the Article}
Having presented the main results of this paper and outlined a road map of their proofs, we now comment on the organization of the article.

\cref{sec:pre} is devoted to preliminary facts needed throughout the body of paper. \cref{ssec:pre_not} introduces the basic notation used in the article. \cref{ssec:pre_HA} collects technical facts from harmonic analysis, many of which are standard in the study of fluid equations, concerning potential and singular integral estimates used extensively in \cref{sec:kprop,sec:MR}. \cref{ssec:pre_Os} reviews Osgood moduli of continuity and the Osgood lemma, which we shall need to conclude the proof of \cref{thm:main} in \cref{ssec:MR_thm}. Finally, \cref{ssec:pre_WS} reviews basic properties of the point vortex model and Euler equation concerning notions of solutions and well-posedness. As \cref{sec:pre} is intended to aid in the reading of \Cref{sec:CE,sec:kprop,sec:MR}, the reader may wish to consult this section only as necessary.

\cref{sec:CE} is devoted to properties of the Coulomb potential $\g$ and the modulated energy $\Fr_N^{avg}(\cdot,\cdot)$ previously introduced in \eqref{eq:FN_avg_def}. There is substantial overlap in this section with results scattered throughout the articles \cite{Duerinckx2016,Serfaty2018mean}. We include the material again in this article as we have organized our presentation differently and in the interests of completeness of our exposition. \cref{ssec:CE_set} introduces the truncation procedure for the potential $\g$ and the associated smearing procedure for point masses. \cref{ssec:CE_EF} reviews basic properties of the modulated energy. \cref{ssec:CE_CL} gives the proof of \cref{lem:count}, important for the proof of \cref{prop:key}, which counts the number of ``close'' point vortices. Finally, \cref{ssec:CE_coer} gives the proof of \cref{prop:bes_conv}, which shows that the modulated energy controls the Sobolev norm $\|\cdot\|_{H^s(\R^2)}$, up to an additive error.

\cref{sec:kprop} gives the proof of \cref{prop:key}, the workhorse of this paper. We begin the section with an extended overview of the proof of the proposition, and we have divided the section into several subsections, corresponding to the main steps of the proof. \cref{ssec:kprop_SET} briefly reviews the stress-energy tensor associated to two functions. \cref{ssec:kprop_moll} carries out our novel mollification step, culminating in the proof of \cref{lem:kprop_error}. \cref{ssec:kprop_renorm,ssec:kprop_diag} perform the renormalization and analysis-of-diagonal-terms steps, respectively, in which we prove \Cref{lem:kprop_renorm,lem:kprop_diag}. These steps essentially carry over from \cite{Serfaty2018mean}, and therefore we have been parsimonious with the details. \cref{ssec:kprop_recomb} is the recombination step, which is the longest and most involved and where a large chunk of the analysis new to our work takes place. The main conclusion of this subsection is \cref{lem:kprop_recomb}. Lastly, \cref{ssec:kprop_conc} concludes the proof of \cref{prop:key}.

Finally, \cref{sec:MR} provides the proofs of our main results: \cref{thm:main} and \cref{cor:main}. In \cref{ssec:MR_EI}, we give the rigorous proof of the modulated energy derivative identity \eqref{eq:intro_ME_deriv} with \cref{lem:en_ineq}. In \cref{ssec:MR_thm} and \cref{ssec:MR_cor}, we prove \cref{thm:main} and \cref{cor:main}, respectively.

\subsection{Acknowledgments}
The author thanks Nata\v{s}a Pavlovi\'{c} for helpful comments on an earlier draft of the manuscript. The author also thanks Sylvia Serfaty for her feedback, which has improved the discussion in \cref{ssec:intro_PR} and corrected an earlier misattribution. The author acknowledges financial support from the Simons Foundation and from The University of Texas at Austin through a Provost Graduate Excellence Fellowship.

\section{Preliminaries}
\label{sec:pre}
In this section, we introduce the basic notation used throughout this article, review some facts from harmonic analysis and the theory of integral inequalities, and lastly review well-known results on the well-posedness of the point vortex model \eqref{eq:PVM} and the Euler equation \eqref{eq:Eul}. The reader may skip this section upon first reading and consult as necessary.

\subsection{Basic Notation}
\label{ssec:pre_not}

Given nonnegative quantities $A$ and $B$, we write $A\lesssim B$ if there exists a constant $C>0$, independent of $A$ and $B$, such that $A\leq CB$. If $A \lesssim B$ and $B\lesssim A$, we write $A\sim B$. To emphasize the dependence of the constant $C$ on some parameter $p$, we sometimes write $A\lesssim_p B$ or $A\sim_p B$.

We denote the natural numbers excluding zero by $\N$ and including zero by $\N_0$. Similarly, we denote the nonnegative real numbers by $\R_{\geq 0}$ and the positive real numbers by $\R_+$ or $\R_{>0}$.

Given $N\in\N$ and points $x_1,\ldots,x_N$ in some set $X$, we will write $\ux_N$ to denote the $N$-tuple $(x_1,\ldots,x_N)$. We define the generalized diagonal $\Delta_N$ of the Cartesian product $X^N$ to be the set
\begin{equation}
\Delta_N \coloneqq \{(x_1,\ldots,x_N) \in X^N : x_i=x_j \text{ for some $i\neq j$}\}.
\end{equation}
Given $x\in\R^n$ and $r>0$, we denote the ball and sphere centered at $x$ of radius $r$ by $B(x,r)$ and $\p B(x,r)$, respectively. We denote the uniform probability measure on the sphere $\p B(x,r)$ by $\sigma_{\p B(x,r)}$. Given a function $f$, we denote the support of $f$ by $\supp(f)$. On $\R^2$, we define the spatial gradient $\nabla=(\p_{x_1},\p_{x_2})$ and the perpendicular (spatial) gradient $\nabla^\perp = (-\p_{x_2},\p_{x_1})$, where $\p_{x_j}$ is the partial derivative in the $x_j$-direction. We use the notation $\jp{x}\coloneqq (1+|x|^2)^{1/2}$ to denote the Japanese bracket.

If $A=(A^{ij})_{i,j=1}^N$ and $B=(B^{ij})_{i,j=1}^N$ are two $N\times N$ matrices, with entries in $\C$, we denote their Frobenius inner product by
\begin{equation}
A : B \coloneqq \sum_{i,j=1}^N A^{ij}\ol{B^{ij}}.
\end{equation}

We denote the space of complex-valued Borel measures on $\R^n$ by $\M(\R^n)$. We denote the subspace of probability measures (i.e. elements $\mu\in\M(\R^n)$ with $\mu\geq 0$ and $\mu(\R^n)=1$) by $\P(\R^n)$. When $\mu$ is in fact absolutely continuous with respect to Lebesgue measure on $\R^n$, we shall abuse notation by writing $\mu$ for both the measure and its density function. 

We denote the Banach space of complex-valued continuous, bounded functions on $\R^n$ by $C(\R^n)$ equipped with the uniform norm $\|\cdot\|_{\infty}$. More generally, we denote the Banach space of $k$-times continuously differentiable functions with bounded derivatives up to order $k$ by $C^k(\R^n)$ equipped with the natural norm, and we define $C^\infty \coloneqq \bigcap_{k=1}^\infty C^k$. We denote the subspace of smooth functions with compact support by $C_c^\infty(\R^n)$, and use the subscript $0$ to indicate functions vanishing at infinity. We define the space of locally continuous functions by $C_{loc}(\R^n)$ and similarly for locally $C^k$ and locally smooth functions. We denote the Schwartz space of functions by $\Sc(\R^n)$ and the space of tempered distributions by $\Sc'(\R^n)$.

For $p\in [1,\infty]$ and $\Omega\subset\R^n$, we define $L^p(\Omega)$ to be the usual Banach space equipped with the norm
\begin{equation}
\|f\|_{L^p(\Omega)} \coloneqq \paren*{\int_\Omega |f(x)|^p dx}^{1/p}
\end{equation}
with the obvious modification if $p=\infty$. When $\Omega=\R^n$, we sometimes just write $\|f\|_{L^p}$. Lastly, we use the notation $\|\cdot\|_{\ell^p}$ in the conventional manner. In the case where $f: \Omega\rightarrow X$ takes values in some Banach space $(X,\|\cdot\|_{X})$ (e.g. $L^p(\Omega;L^q(\Omega))$, we shall write $\|f\|_{L^p(\Omega;X)}$. In the special case of $L^p([0,T]; L^q(\R^2))$, shall use the abbreviation $L_t^pL_x^q([0,T]\times\R^2)$, which is justified by the Fubini-Tonelli theorem.

Our conventions for the Fourier transform and inverse Fourier transform are respectively
\begin{align}
\mathcal{F}(f)(\xi) &\coloneqq \wh{f}(\xi) \coloneqq \int_{\R^n}f(x)e^{-i\xi\cdot x}dx \qquad \forall \xi\in\R^n,\\
\mathcal{F}^{-1}(f)(x) &\coloneqq f^{\vee}(x) \coloneqq \frac{1}{(2\pi)^n}\int_{\R^n}f(\xi)e^{i x\cdot\xi}d\xi \qquad \forall x\in\R^n.
\end{align}

For integer $k\in\N_0$ and $1\leq p\leq \infty$, we define the usual Sobolev spaces
\begin{equation}
\begin{split}
W^{k,p}(\R^n) \coloneqq \{\mu \in L^p(\R^n) : \nabla^k \mu \in L^p(\R^n;(\C^n)^{\otimes k}),\quad \|\mu\|_{W^{k,p}(\R^n)} &\coloneqq \sum_{k=0}^n \|\nabla^{k}\mu\|_{L^p(\R^n)}.
\end{split}
\end{equation}
For $s\in\R$, we define the inhomogeneous Sobolev space $H^s(\R^n)$ to be the space of $\mu\in\Sc'(\R^n)$ such that $\wh{\mu}$ is locally integrable and 
\begin{equation}
\label{eq:H^s_def}
\|\mu\|_{H^s(\R^n)} \coloneqq \paren*{\int_{\R^n} \jp{\xi}^{2s} |\wh{\mu}(\xi)|^2d\xi }^{1/2}<\infty,
\end{equation}
and we use the notation $\|\mu\|_{\dot{H}^s(\R^n)}$ to denote the semi-norm where $\jp{\xi}$ is replaced by $|\xi|$.\footnote{As is standard notation, a dot superscript indicates a homogeneous semi-norm for a function space in this paper.} As is well-known, $W^{k,2}(\R^n) = H^k(\R^n)$ for $k\in\N_0$.

\subsection{Harmonic Analysis}
\label{ssec:pre_HA}
In this subsection, we recall some basic facts from harmonic analysis concerning function spaces, Littlewood-Paley theory, and Riesz potential estimates. This material is standard in the field, and the reader can consult references such as \cite{Stein1970,Stein1993,grafakos2014c, grafakos2014m}.

We begin with the definition of the Riesz potential.

\begin{mydef}[Riesz potential]
\label{def:RP}
Let $n\in\N$. For $s>-n$, we define the Fourier multiplier $(-\Delta)^{s/2}$ by
\begin{equation}
((-\Delta)^{s/2}f)(x) \coloneqq (|\xi|^{s}\wh{f}(\xi))^\vee(x), \qquad x\in \R^n,
\end{equation}
for a Schwartz function $f\in \Sc(\R^n)$. Since, for $s\in (-n,0)$, the inverse Fourier transform of $|\xi|^s$ is the tempered distribution
\begin{equation}
\frac{2^s\Gamma(\frac{n+s}{2})}{\pi^{\frac{n}{2}}\Gamma(-\frac{s}{2})} |x|^{-s-n},
\end{equation}
it follows that
\begin{equation}
((-\Delta)^{s/2}f)(x) = \frac{2^s \Gamma(\frac{n+s}{2})}{\pi^{\frac{n}{2}}\Gamma(-\frac{s}{2})}\int_{\R^n}\frac{f(y)}{|x-y|^{s+n}}dy, \qquad x\in\R^n.
\end{equation}
For $s\in (0,n)$, we define the \emph{Riesz potential operator} of order $s$ by $\mathcal{I}_s \coloneqq (-\Delta)^{-s/2}$ on $\Sc(\R^n)$.
\end{mydef}

$\mathcal{I}_s$ extends to a well-defined operator on any $L^p$ space, the extension also denoted by $\mathcal{I}_s$ by an abuse notation, as a consequence of the \emph{Hardy-Littlewood-Sobolev lemma}.

\begin{prop}[Hardy-Littlewood-Sobolev]
\label{prop:HLS}
Let $n\in\N$, $s \in (0,n)$, and $1<p<q<\infty$ satisfy the relation
\begin{equation}
\frac{1}{p}-\frac{1}{q} = \frac{s}{n}.
\end{equation}
Then for all $f\in\Sc(\R^n)$,
\begin{align}
\|\mathcal{I}_s(f)\|_{L^q(\R^n)} &\lesssim_{n,s,p} \|f\|_{L^p(\R^n)},\\
\|\mathcal{I}_s(f)\|_{L^{\frac{n}{n-s},\infty}(\R^n)} &\lesssim_{n,s} \|f\|_{L^1(\R^n)},
\end{align}
where $L^{r,\infty}$ denotes the weak-$L^r$ space. Consequently, $\mathcal{I}_s$ has a unique extension to $L^p$, for all $1\leq p<\infty$.
\end{prop}

Although the Hardy-Littlewood-Sobolev lemma breaks down at the endpoint case $p=\infty$ (one has a $BMO$ substitute which is not useful for our purposes), the next lemma allows us to control the $L^\infty$ norm of $\mathcal{I}_s(f)$ in terms of the $L^1$ norm and $L^p$ norm, for some $p=p(s,n)$. Its usefulness stems from the fact that if $w$ solves the Poisson equation $-\D w = \varphi$ with zero boundary condition at infinity, then
\begin{equation}
|(\nabla w)(x)| \lesssim \I_1(|w|)(x) \qquad \forall x\in\R^2,
\end{equation}
as can be seen from using the Fourier transform.

\begin{lemma}[$L^{\infty}$ bound for Riesz potential]
\label{lem:Linf_RP}
For any $n\in\N$, $s\in (0,n)$, and $p\in (\frac{n}{s},\infty]$,
\begin{equation}
\|\mathcal{I}_s(f)\|_{L^\infty(\R^n)} \lesssim_{s,n,p} \|f\|_{L^{1}(\R^n)}^{1-\frac{n-s}{n(1-\frac{1}{p})}} \|f\|_{L^{p}(\R^n)}^{\frac{n-s}{n(1-\frac{1}{p})}}.
\end{equation}
\end{lemma}

We next define the Besov scale of function spaces, which first requires us to recall some basic facts from Littlewood-Paley theory. Let $\phi\in C_{c}^{\infty}(\R^n)$ be a radial, nonincreasing function, such that $0\leq \phi\leq 1$ and
\begin{equation}
\phi(x)=
\begin{cases}
1, & {|x|\leq 1}\\
0, & {|x|>2}
\end{cases}.
\end{equation}
Define the dyadic partitions of unity
\begin{align}
1&=\phi(x)+\sum_{j=1}^{\infty}[\phi(2^{-j}x)-\phi(2^{-j+1}x)] \eqqcolon \psi_{\leq 0}(x)+\sum_{j=1}^{\infty}\psi_{j}(x) \qquad \forall x\in\R^n\\
1&=\sum_{j\in\mathbb{Z}}[\phi(2^{-j}x)-\phi(2^{-j+1}x)] \eqqcolon \sum_{j\in\mathbb{Z}}{\psi}_{j}(x) \qquad \forall x\in\R^n\setminus\{0\}.
\end{align}
For any $j\in\Z$, we define the Littlewood-Paley projector $P_{j}, P_{\leq 0}$ by
\begin{align}
(P_jf)(x) &\coloneqq (\psi_{j}(D)f)(x) = \int_{\R^n}K_{j}(x-y)f(y)dy \qquad K_j \coloneqq \psi_j^{\vee},\\
(P_{\leq 0}f)(x) &\coloneqq (\psi_{\leq 0}(D)f)(x) = \int_{\R^n} K_{\leq 0}(x-y)f(y)dy \qquad K_{\leq 0} \coloneqq \psi_{\leq 0}^{\vee}.
\end{align}

\begin{mydef}[Besov space]
\label{def:Bes}
Let $s\in\R$ and $1\leq p,q\leq\infty$. We define the inhomogeneous Besov space $B_{p,q}^s(\R^n)$ to be the space of $\mu\in\Sc'(\R^n)$ such that
\begin{equation}
\|\mu\|_{B_{p,q}^s(\R^n)} \coloneqq \paren*{\|P_{\leq 0}\mu\|_{L^p(\R^n)} + \sum_{j=1}^\infty 2^{jqs}\|P_j\mu\|_{L^p(\R^n)}^q}^{1/q} < \infty.
\end{equation}
For $p,q,s$ as above, we also define the homogeneous Besov semi-norm
\begin{equation}
\|\mu\|_{\dot{B}_{p,q}^s(\R^n)} \coloneqq \paren*{\sum_{j\in\Z} 2^{jqs} \|P_j\mu\|_{L^p(\R^n)}^{q}}^{1/q}.
\end{equation}
\end{mydef}

\begin{remark}
Any two choices of Littlewood-Paley partitions of unity used to define $\|\cdot\|_{B_{p,q}^s(\R^n)}$ (resp. $\|\cdot\|_{\dot{B}_{p,q}^s(\R^n)}$) lead to equivalent norms (resp. semi-norms). 
\end{remark}

\begin{remark}
The space $B_{2,2}^s(\R^n)$ coincides with the Sobolev space $H^s(\R^n)$, as can be seen from Plancherel's theorem. For $s\in \R_+\setminus\N$, the space $B_{\infty,\infty}^s(\R^n)$ coincides with the H\"older space $C^{[s],s-[s]}(\R^n)$ of bounded functions $\mu:\R^n\rightarrow\C$ such that $\nabla^{k}\mu$ is bounded, for integers $0\leq k\leq [s]$ and
\begin{equation}
\|\nabla^{[s]}\mu\|_{\dot{C}^{s-[s]}(\R^n)} \coloneqq \sup_{{0<|x-y|\leq 1}} \frac{|(\nabla^{[s]}\mu)(x) - (\nabla^{[s]}\mu)(y)|}{|x-y|^{s-[s]}} <\infty. 
\end{equation}
For integer $s$, the space $B_{\infty,\infty}^s(\R^n)$, sometimes called the Zygmund space of order $s$, is strictly large than $C^s(\R^n)$.
\end{remark}

We next define the space of \emph{log-Lipschitz functions}, which are of central to importance to this work.

\begin{mydef}[Log-Lipschitz space]
\label{def:LL}
We define $LL(\R^n)$ to be the space of functions $\mu\in C(\R^n)$ such that
\begin{equation}
\|\mu\|_{LL(\R^n)} \coloneqq \sup_{0<|x-y|\leq e^{-1}} \frac{|\mu(x)-\mu(y)|}{|x-y||\ln|x-y||} <\infty.
\end{equation}
\end{mydef}

The next lemma shows that $B_{\infty,\infty}^1(\R^n)$ continuously embeds in $LL(\R^n)$.

\begin{lemma}
\label{lem:LL_embed}
It holds that
\begin{equation}
\|\mu\|_{LL(\R^n)} \lesssim_n \|\nabla \mu\|_{B_{\infty,\infty}^0(\R^n)}, \qquad \forall \mu\in B_{\infty,\infty}^1(\R^n).
\end{equation}
\end{lemma} 

Using \cref{lem:LL_embed}, Bernstein's lemma, and the isomorphism $B_{\infty,\infty}^s(\R^n)\cong C^{[s],s-[s]}(\R^n)$, for non-integer $s$, we obtain the following regularity estimates for Riesz potentials. In particular, we obtain that if the velocity field $u$ for the Euler equation \eqref{eq:Eul} belongs to $LL(\R^n)$, provided that the vorticity $\omega\in L^1(\R^n)\cap L^\infty(\R^n)$. 

\begin{lemma}
\label{lem:pot_bnds}
For $0<s<2$ and $n\in\N$, it holds that 
\begin{align}
\|\I_s(\mu)\|_{\dot{C}^{s}(\R^n)} &\lesssim_{s,n} \|\mu\|_{L^\infty(\R^n)} \qquad 0<s<1,\\
\|\I_s(\mu)\|_{LL(\R^n)} &\lesssim_n \|\mu\|_{L^\infty(\R^n)} \qquad s=1,\\
\|\I_s(\mu)\|_{\dot{C}^{1,s-1}(\R^n)} &\lesssim_{s,n} \|\mu\|_{L^\infty(\R^n)} \qquad 1<s<2
\end{align}
for all $u\in L^p(\R^2)\cap L^\infty(\R^n)$, for some finite $p$.
\end{lemma}

The next lemma contains some potential theory estimates for solutions to Poisson's equation in two dimensions. We recall from the introduction that $\g(x) = -\frac{1}{2\pi}\ln|x|$ is the 2D Coulomb potential.

\begin{lemma}
\label{lem:PE_bnds}
Suppose that $\mu\in L^1(\R^2)\cap L^p(\R^2)$, for some $1<p\leq\infty$, is a measurable function such that $\int_{\R^2}\ln\jp{x}|\mu(x)|dx<\infty$. Then the convolution $\g\ast\mu$ is a well-defined continuous function and we have the point-wise estimate
\begin{equation}
|(\g\ast\mu)(x)| \lesssim_p \jp{x}^{\frac{p-1}{p}} \ln(2\jp{x}) \|\mu\|_{L^p(\R^2)} + \int_{\R^2}\ln(2\jp{y}) |\mu(y)|dy.
\end{equation}
Moreover, if $1<p\leq 2$, then
\begin{equation}
\|\g\ast\mu\|_{\dot{B}_{\infty,\infty}^{\frac{2(p-1)}{p}}(\R^2)} \lesssim_p \|\mu\|_{L^p(\R^2)},
\end{equation}
and if $2<p\leq\infty$, then $\nabla(\g\ast\mu)\in B_{\infty,\infty}^{\frac{p-2}{p}}(\R^2)$ and
\begin{equation}
\|\nabla(\g\ast\mu)\|_{B_{\infty,\infty}^{\frac{p-2}{p}}(\R^2)} \lesssim_p \|\mu\|_{L^1(\R^2)}^{1/2}\|\mu\|_{L^p(\R^2)}^{1/2} + \|\mu\|_{L^p(\R^2)}.
\end{equation}
\end{lemma}
\begin{proof}
To prove that $\g\ast\mu$ is well-defined, for $x\neq 0$, we decompose
\begin{equation}
\int_{\R^2}\g(x-y)\mu(y)dy = \int_{|x-y|\leq 2\jp{x}} \g(x-y)\mu(y)dy + \int_{|x-y|>2\jp{x}} \g(x-y)\mu(y)dy.
\end{equation}
By H\"older's inequality, we have that
\begin{align}
\int_{|x-y|\leq 2\jp{x}}|\g(x-y)\mu(y)|dy &\leq \|\mu\|_{L^p(\R^2)} \paren*{\int_{|x-y|\leq 2\jp{x}} |\ln|x-y| |^{p'}dy}^{1/p'} \nn\\
&\lesssim_p \jp{x}^{\frac{p-1}{p}} \ln(2\jp{x}) \|\mu\|_{L^p(\R^2)},
\end{align}
where the ultimate inequality follows from changing to polar coordinates. If $|x-y|>2\jp{x}$, then it follows from the reverse triangle inequality that $2|y|>|x-y|$ and therefore
\begin{equation}
|\g(x-y)| \lesssim \ln 2 + \ln\jp{y}.
\end{equation}
Thus,
\begin{equation}
\int_{|x-y|>2\jp{x}} |\g(x-y)\mu(y)|dy \lesssim \int_{\R^2}\ln(2\jp{y})|\mu(y)|dy<\infty.
\end{equation}
That $\g\ast\mu$ is continuous follows readily from the dominated convergence theorem and arguing similarly as to above.

Since the distributional Fourier transform $\wh{\g}$ coincides with the function $|\xi|^{-2}$ outside the origin, it follows from Plancherel's theorem that for any $j\in\Z$,
\begin{equation}
\wh{P_j(\g\ast\mu)}(\xi) = \psi_j(\xi)|\xi|^{-2}\wh{\mu}(\xi), \qquad \forall \xi\in\R^2\setminus\{0\}.
\end{equation}
So it follows from Young's inequality and Bernstein's lemma that
\begin{equation}
\|P_j(\g\ast\mu)\|_{L^\infty(\R^2)} \lesssim 2^{-2j} \|P_j\mu\|_{L^\infty(\R^2)} \lesssim_p 2^{\frac{2j}{p}-2j} \|P_j\mu\|_{L^p(\R^2)}.
\end{equation}
Multiplying both sides of the preceding inequality by $2^{\frac{2(p-1)j}{p}}$ and taking the supremum over $j\in\Z$, we conclude that
\begin{equation}
\|\g\ast\mu\|_{\dot{B}_{\infty,\infty}^{\frac{2(p-1)}{p}}(\R^2)} \lesssim_p \|\mu\|_{L^p(\R^2)}.
\end{equation}

Now if $2<p<\infty$, it's straightforward to check from integrating against a test function, the Fubini-Tonelli theorem, and integration by parts that
\begin{equation}
\nabla(\g\ast\mu)= (\nabla\g\ast\mu) = \I_{1}(\mu),
\end{equation}
with equality in the sense of distributions. By \cref{lem:Linf_RP}, the right-hand side is well-defined in $L^\infty(\R^2)$. By Bernstein's lemma and \cref{lem:Linf_RP},
\begin{align}
\|P_{\leq 0}\I_1(\mu)\|_{L^\infty(\R^)} &\lesssim \|\I_1(\mu)\|_{L^\infty(\R^2)} \lesssim_p \|\mu\|_{L^1(\R^2)}^{1/2} \|\mu\|_{L^p(\R^2)}^{1/2}, \nn\\
\|P_j\I_1(\mu)\|_{L^\infty(\R^2)} &\lesssim 2^{-j}\|P_j\mu\|_{L^\infty(\R^2)}  \lesssim 2^{\frac{2j}{p} - j} \|\mu\|_{L^p(\R^2)}.
\end{align}
Multiplying both sides of the inequality in the second line by $2^{\frac{j(p-2)}{p}}$ and taking the supremum over $j\in\N$ completes the proof.
\end{proof}

We conclude this subsection with some quantitative estimates for the rate of convergence of mollification. These estimates are frequently used during the course of the proof of \cref{prop:key} in \cref{sec:kprop}.

\begin{lemma}
\label{lem:conv_bnds}
Let $\chi\in C_c^\infty(\R^n)$ such that $\chi\geq 0$, $\supp(\chi)\subset B(0,1)$, and $\int_{\R^n}\chi(x)dx=1$. For $0<\ep\ll 1$, define $\chi_\ep(x)\coloneqq \ep^{-2}\chi(x/\ep)$. and $\mu_\ep\coloneqq \mu\ast\chi_\ep$. Then for every $\mu\in LL(\R^n)$, we have the estimates
\begin{align}
\|\mu_\ep\|_{L^\infty(\R^n)} &\leq \|\mu\|_{L^\infty(\R^n)}, \label{eq:v_Linf} \\ 
\|\mu-\mu_\ep\|_{L^\infty(\R^n)} &\leq \|\mu\|_{LL(\R^n)}\ep|\ln\ep|, \label{eq:v_diff_Linf}\\
\|\nabla\mu_\ep\|_{L^\infty(\R^n)} &\lesssim \|\mu\|_{LL(\R^n)}|\ln\ep| \label{eq:v_grad_Linf}
\end{align}
\end{lemma}
\begin{proof}
Estimate \eqref{eq:v_Linf} is immediate from Young's inequality. For estimate \eqref{eq:v_diff_Linf}, we use that $\chi$ has unit mean to write
\begin{equation}
\mu(x) -\mu_\ep(x) = \int_{\R^n}\paren*{\mu(x)-\mu(x-y)}\chi_\ep(y)dy.
\end{equation}
By definition of the semi-norm $\|\cdot\|_{LL(\R^n)}$ in \cref{def:LL} and that $\supp(\chi_\ep)\subset B(0,\ep)$, we have that for $\ep<e^{-1}$,
\begin{equation}
\left|\int_{\R^n}\paren*{\mu(x)-\mu(x-y)}\chi_\ep(y)dy\right| \leq  \|\mu\|_{LL(\R^n)}\int_{\R^n}|y\ln |y|| |\chi_\ep(y)|dy.
\end{equation}
Using the dilation invariance of Lebesgue measure, we see that the right-hand side of the preceding inequality equals
\begin{align}
\|\mu\|_{LL(\R^n)}\ep\int_{\R^n}\frac{|y|}{\ep}\left|\ln \frac{|y|}{\ep} + \ln\ep\right||\chi_\ep(y)|dy &\leq \|\mu\|_{LL(\R^n)}\paren*{\ep|\ln \ep| +\ep\int_{\R^n}|y\ln|y|| |\chi(y)|dy} \nn\\
&\lesssim \|\mu\|_{LL(\R^n)}\ep|\ln\ep|.
\end{align}

For estimate \eqref{eq:v_Linf}, we first observe that since
\begin{equation}
\int_{\R^n} (\nabla\chi_\ep)(y)dy = \vec{0}\in\R^n
\end{equation}
by the fundamental theorem of calculus and the compact support of $\chi_\ep$, we may write
\begin{equation}
(\nabla\mu_\ep)(x) = \int_{\R^n}\paren*{\mu(x-y)-\mu(x)} (\nabla\chi_\ep)(y)dy.
\end{equation}
Hence, for $\ep<e^{-1}$,
\begin{align}
\left|(\nabla\mu_\ep)(x)\right| &\leq \|\mu\|_{LL(\R^n)}\int_{\R^n}|y\ln|y|| |(\nabla\chi_\ep)(y)|dy \nn\\
&=\frac{\|\mu\|_{LL(\R^n)}}{\ep^2}\int_{\R^n} \left|\frac{y}{\ep}\paren*{\ln\frac{|y|}{\ep} +\ln\ep}\right| |(\nabla\chi)(\frac{y}{\ep})|dy \nn\\
&\leq\|\mu\|_{LL(\R^n)}\int_{\R^n}|y|\paren*{\ln|y|+\ln\ep} |(\nabla\chi)(y)|dy,
\end{align}
where the ultimate line follows from the dilation invariance of Lebesgue measure. Since $x\in\R^n$ was arbitrary, the proof of \eqref{eq:v_grad_Linf}, and therefore the lemma, is complete.
\end{proof}

\subsection{The Osgood Lemma}
\label{ssec:pre_Os}
In this subsection, we recall some facts related to moduli of continuity and the Osgood lemma, which is a generalization of the Gronwall-Bellman inequality. The presentation in this subsection closely follows that of \cite[Section 3.1]{BCD2011}.

\begin{mydef}[Modulus of continuity]
\label{def:mod_cont}
Let $a\in (0,1]$. A \emph{modulus of continuity} is an increasing, nonzero continuous function $\rho:[0,a]\rightarrow [0,\infty)$ such that $\rho(0)=0$. We say that a modulus of continuity satisfies the \emph{Osgood condition} or is an \emph{Osgood modulus of continuity}, if
\begin{equation}
\label{eq:Os_con}
\int_0^a \frac{dr}{\rho(r)} = \infty.
\end{equation}
\end{mydef}

\begin{ex}
Evidently, the function $\rho: [0,1]\rightarrow [0,\infty), \ \rho(r)\coloneqq r$ is an Osgood modulus of continuity. Since the anti-derivative of the reciprocal of the function
\begin{equation}
\rho: [0,e^{-1}] \rightarrow [0,\infty), \qquad \rho(r) \coloneqq r\paren*{\ln\frac{1}{r}} 
\end{equation}
is, up to an additive constant, $-\ln\ln(\frac{1}{r})$, we see from the fundamental theorem of calculus that this is also an example of an Osgood modulus of continuity. This latter example will be important to the proof of \cref{thm:main} due to the Euler velocity field only being log-Lipschitz when the vorticity is bounded. Non-examples of Osgood moduli of continuity include $r\mapsto r^\alpha$, for $0<\alpha<1$, and $r\mapsto r(\ln \frac{1}{r})^\alpha$, for $\alpha>1$.
\end{ex}

\begin{lemma}[Osgood lemma]
\label{lem:Os}
Fix $a\in (0,1]$. Let $f:[t_0,T]\rightarrow [0,a]$ be a measurable function, $\gamma:[t_0,T]\rightarrow [0,\infty)$ a locally integrable function, and $\rho:[0,a]\rightarrow [0,\infty)$ an Osgood modulus of continuity. Suppose that there exists a constant $c>0$ such that
\begin{equation}
f(t)\leq c + \int_{t_0}^t\gamma(t')\rho(f(t'))dt' \qquad \text{a.e} \ t \in [t_0,T].
\end{equation}
Define the function
\begin{equation}
\mathfrak{M}:(0,a] \rightarrow [0,\infty), \qquad \mathfrak{M}(x) \coloneqq \int_{x}^a\frac{dr}{\rho(r)}dr.
\end{equation}
Then $\mathfrak{M}$ is bijective, and if $t$ is such that $\int_{t_0}^t\gamma(t')dt' \leq \mathfrak{M}(c)$, it holds that
\begin{equation}
f(t) \leq \mathfrak{M}^{-1}\paren*{\mathfrak{M}(c)-\int_{t_0}^t\gamma(t')dt'}.
\end{equation}

\end{lemma}
\begin{proof}
See \cite[Lemma 3.4, Corollary 3.5]{BCD2011}.
\end{proof}

\begin{remark}
\label{rem:Os_ex_log}
Note that if $\rho:[0,e^{-1}]\rightarrow [0,\infty)$ is defined by $\rho(r)\coloneqq r\ln(1/r)$, then
\begin{equation}
\mathfrak{M}(x) = \ln\ln(\frac{1}{x}) \qquad \text{and} \qquad \mathfrak{M}^{-1}(y) = e^{-e^{y}},
\end{equation}
as the reader may check. The importance of this example will become clear in \cref{ssec:MR_thm}.
\end{remark}

\subsection{The Point Vortex Model and Euler Equation}
\label{ssec:pre_WS}
In an effort to make this article self-contained, we briefly review properties of the point vortex model \eqref{eq:PVM} and Euler equation \eqref{eq:Eul}, beginning with the former. For more details, in particular proofs of the omitted results, we refer to the texts \cite{MB2002,BCD2011,MP2012book}.

As mentioned in the introduction, the system of equations \eqref{eq:PVM} is a finite-dimensional Hamiltonian system. Indeed, specializing to the repulsive case $a_1=\cdots=a_N=1/N$, the reader can check that if we define the $2\times 2$ matrix corresponding to $90^\circ$ rotation,
\begin{equation}
\label{eq:J_def}
\J \coloneqq \begin{bmatrix} 0 & -1 \\ 1 & 0 \end{bmatrix},
\end{equation}
and the $2N\times 2N$ block-diagonal matrix $\J_N\coloneqq \frac{1}{N}\diag_N(\J,\ldots,\J)$, then
\begin{equation}
\omega_N: (\R^2)^N\times (\R^2)^N \rightarrow \R, \qquad \omega_N(\ux_N,\uy_N) \coloneqq \J_N^{-1}\ux_N\cdot \uy_N,
\end{equation}
where $()\cdot()$ denotes the standard inner product on $(\R^2)^N$, defines a symplectic form. Defining the Hamiltonian functional
\begin{equation}
\label{eq:ham_def}
H_N(\ux_N) \coloneqq \frac{1}{2N^2}\sum_{1\leq i\neq j\leq N} \g(x_i-x_j) \qquad \forall \ux_N\in (\R^2)^N\setminus\D_N,
\end{equation}
it is a straightforward computation (see \cite[Section 3.1]{Rosenzweig2019_PV}) that the point vortex equations \eqref{eq:PVM} can be rewritten as 
\begin{equation}
\dot{\ux}_N(t) = (\nabla_\omega H_N)(\ux_N(t)),
\end{equation}
where $\nabla_{\omega}H_N$ denotes the symplectic gradient of $H_N$ with respect to the form $\omega$.

\begin{remark}\label{rem:ham_con}
Using the anti-symmetry of the symplectic form $\omega_N$, it follows that the Hamiltonian $H_N$ is trivially conserved by solutions to the system \eqref{eq:PVM}. Furthermore, since $H_N$ is invariant under translation and rotation, it follows from Noether's theorem--or direct computation--that the center of vorticity and moment of inertia, respectively given by
\begin{align}
M(\ux_N) &\coloneqq \frac{1}{N}\sum_{i=1}^N x_i, \label{eq:CV_deF}\\
I(\ux_N) &\coloneqq \frac{1}{N}\sum_{i=1}^N |x_i|^2, \label{eq:MI_def}
\end{align}
are also conserved by solutions.
\end{remark}

The next definition clarifies what we mean by a solution to the system \eqref{eq:PVM}.

\begin{mydef}[Solution to PVM]
\label{def:pvm_soln}
Given a time $T>0$ and initial configuration $\ux_{N}^{0}\in(\R^2)^N \setminus\Delta_N$, we say that a function $\ux_{N}\in C([0,T];(\R^2)^N)$ is a solution to the system \eqref{eq:PVM} if
\begin{equation}
\label{eq:sep_con}
\min_{1\leq i< j\leq N} |x_{i}(t)-x_{j}(t)| > 0 \qquad \forall t\in [0,T]
\end{equation}
and for every $i\in\{1,\ldots,N\}$,
\begin{equation}
\label{eq:PVM_mild}
x_{i}(t) = x_{i}^{0}+\sum_{1\leq i\neq j\leq N} \frac{1}{N} \int_{0}^{t} (\nabla^\perp\g)(x_{i}(s)-x_{j}(s))ds \qquad  \forall t\in [0,T].
\end{equation}
We say that $\ux_{N}\in C_{loc}([0,T);(\R^2)^N)$ is a solution to \eqref{eq:PVM} if for every $0\leq T'<T$, $\ux_{N}\in C([0,T'];(\R^2)^N)$ is a solution to \eqref{eq:PVM}. We say that a solution $\ux_N \in C_{loc}([0,T);(\R^2)^N)$ has \emph{maximal lifespan} if $\ux_N$ is not a solution on the interval $[0,T]$. If $T=\infty$, then we say that the solution is \emph{global}.
\end{mydef}

\begin{remark}
\label{rem:PVM_smooth}
Since for finite $T>0$, the interval $[0,T]$ is compact, it follows from Weierstrass's extreme value theorem that the conditions
\begin{align}
&\min_{1\leq i< j\leq N} |x_{i}(t)-x_{j}(t)| > 0 \enspace \forall t\in [0,T] \qquad \text{and} \qquad
\min_{t\in [0,T]} \min_{1\leq i< j\leq N} |x_{i}(t)-x_{j}(t)| > 0
\end{align}
are equivalent. Using this property together with the smoothness of $\nabla^\perp\g$ away from the origin and induction, it follows that if $\ux_{N}\in C([0,T];(\R^2)^N)$ is a solution to \eqref{eq:PVM}, then $\ux_{N}\in C^{\infty}([0,T];(\R^2)^N)$.
\end{remark}

\begin{remark}
\label{rem:PVM_uniq}
If a solution to equation \eqref{eq:PVM} exists, then it is necessarily unique, as the reader may check.
\end{remark}

While the Hamiltonian $H_N$ defined in \eqref{eq:ham_def} is not nonnegative-definite, conservation of the Hamiltonian together with conservation of the moment of $I_N$ control the minimal distance between the point vortices globally in time. Indeed, using that $H_N$ is conserved, we have that for any distinct pair $(i_0,j_0)\in\{1,\ldots,N\}^2$ such that $|x_{i_0}(t)-x_{j_0}(t)| <1$, 
\begin{align}
-\frac{1}{2\pi}\ln|x_{i_0}(t)-x_{j_0}(t)| &\leq \sum_{{1\leq i\neq j\leq N}\atop{|x_i(t)-x_j(t)| < 1}}\g(x_i(t)-x_j(t))  \nn\\
&= H_N(0) - \sum_{{1\leq i\neq j\leq N}\atop{|x_i(t)-x_j(t)|\geq 1}} \g(x_i(t)-x_j(t)) \nn\\
&=H_N(0) + \frac{1}{2\pi}\sum_{{1\leq i\neq j\leq N}\atop{|x_i(t)-x_j(t)|\geq 1}} \ln|x_i(t)-x_j(t)|,
\end{align}
for every $t\geq 0$. Using the triangle inequality and that $x\mapsto \ln|x|$ is increasing, we find that
\begin{align}
\ln|x_i(t)-x_j(t)| \leq \ln\paren*{|x_i(t)|+|x_j(t)|} \leq \ln\sqrt{\frac{|x_i(t)|^2 + |x_j(t)|^2}{2}} \leq \ln\sqrt{\frac{I_N(0)}{2}},
\end{align}
where the ultimate inequality follows from the definition \eqref{eq:MI_def} of $I_N(t)$ and that $I_N(t)$ is conserved. After some algebraic manipulation, we then conclude that for any $1\leq i\neq j\leq N$,
\begin{equation}
\label{eq:md_lb}
|x_i(t)-x_j(t)| \geq \min\left\{1,\exp\paren*{-2\pi\paren*{H_N(0) + \frac{{N\choose 2}}{2\pi}\ln\sqrt{\frac{I_N(0)}{2}}}}\right\} \qquad \forall t\geq 0.
\end{equation}
This observation leads to global well-posedness for the system \eqref{eq:PVM}, summarized in the next theorem.

\begin{thm}[Point vortex GWP]
\label{thm:PVM_GWP}
Let $N\in\N$, and let $\ux_N^0\in (\R^2)^N\setminus\D_N$. Then there exists a unique global solution $\ux_N\in C^\infty([0,\infty);(\R^2)^N\setminus\D_N)$ to the system \eqref{eq:PVM} with initial datum $\ux_N^0$.
\end{thm}
\begin{proof}
See \cite[Section 4.2]{MP2012book} or \cite[Sections 3.2 and 3.3]{Rosenzweig2019_PV}.
\end{proof}

We now turn to discussing the Euler equation \eqref{eq:Eul}, beginning with our notion of a weak solution.

\begin{mydef}[Euler weak solution]
\label{def:Eul_ws}
For $T>0$, we say that a spacetime measurable function $\omega: [0,T]\times\R^2\rightarrow\R$ is a \emph{weak solution} to the Cauchy problem \eqref{eq:Eul} if $\omega \in L^\infty([0,T]; L^1(\R^2)\cap L^p(\R^2))$, for some $p>2$, and if for every spacetime test function $\varphi \in W^{1,\infty}([0,T); C_c^1(\R^2))$, we have the Duhamel formula
\begin{equation}
\label{eq:Eul_Duh}
\begin{split}
\int_{\R^2}\omega(t,x)\varphi(t,x)dx &= \int_{\R^2}\omega^0(x)\varphi(0,x)dx + \int_0^t\int_{\R^2}\omega(s,x)(\p_s\varphi)(s,x)dxds \\
&\ph + \int_0^t\int_{\R^2} (\nabla\varphi)(s,x)\cdot u(s,x)\omega(s,x)dxds, \qquad \forall t\in [0,T].
\end{split}
\end{equation}
\end{mydef}

We now make record some remarks on the notion of weak solution contained in \cref{def:Eul_ws}.
\begin{remark}
\label{rem:traj_Lip}
It follows from the identity \eqref{eq:Eul_Duh}, our assumptions on $\omega$, and \cref{lem:conv_bnds} that the trajectory
\begin{equation}
\rho: [0,T]\rightarrow \R, \qquad \rho(t) \coloneqq \int_{\R^2}\omega(t,x)\varphi(t,x)dx)
\end{equation}
is Lipschitz continuous and that we have the bound
\begin{equation}
|\rho(t)-\rho(s)| \lesssim |t-s|\|\nabla\varphi\|_{L_{t,x}^\infty([0,T]\times\R^2)}\|\omega\|_{L_t^\infty L_x^1([0,T]\times\R^2)} \|\omega\|_{L_t^\infty(L_x^1\cap L_x^p)([0,T]\times\R^2)}, \qquad \forall t,s\in[0,T].
\end{equation}
\end{remark}

\begin{remark}
\label{rem:wd_at}
Although our notion of weak solution in \cref{def:Eul_ws} may seem to allow for the possibility of a nonempty, but measure zero, set of times $t$, such that $\omega(t)\notin L^1(\R^2)\cap L^p(\R^2)$, the Duhamel formula \eqref{eq:Eul_Duh} actually implies that
\begin{equation}
\sup_{0\leq t\leq T}\paren*{\|\omega(t)\|_{L^1(\R^2)} + \|\omega(t)\|_{L^p(\R^2)}} < \infty.
\end{equation}
Indeed, suppose there exists some $t_0\in [0,T]$ such that $\omega(t_0)\notin L^q(\R^2)$, for any $1\leq q\leq p$. Then by duality and density, given any $N>0$, there exists a test function $\varphi\in C_c^\infty(\R^2)$ with $\|\varphi\|_{L^{q'}(\R^2)} = 1$ such that
\begin{equation}
\int_{\R^2}\varphi(x)\omega(t_0,x)dx \geq N.
\end{equation}
Since $\omega\in L^\infty([0,T]; L^1(\R^2)\cap L^p(\R^2))$, we see from H\"older's inequality that
\begin{equation}
\left|\int_{\R^2}\varphi(x)\omega(t,x)dx\right| \leq \|\omega(t)\|_{L^q(\R^2)}, \qquad \text{a.e.} \ t\in [0,T].
\end{equation}
By \cref{rem:traj_Lip},
\begin{equation}
\lim_{t\rightarrow t_0} \int_{\R^2}\varphi(x)\omega(t,x)dx = \int_{\R^2}\varphi(x)\omega(t_0,x)dx.
\end{equation} 
Choosing $N>\|\omega\|_{L_t^\infty L_x^q([0,T]\times\R^2)}$, we obtain a contradiction. Thus, there is no ambiguity about measure zero sets in inequalities/identities that are point-wise in time in the sequel.
\end{remark}

The next lemma asserts that if at some time $t_0$, without loss of generality $t_0=0$, the weak solution $\omega$ to \eqref{eq:Eul}, satisfies a logarithmic growth bound, then it satisfies the same bound (with a possibly time-dependent constant) on its lifespan. This logarithmic growth condition ensures that the convolution $\g\ast \omega$ (i.e. the stream function associated to $\omega$) is well-defined. We leave the proof as an exercise to the reader.

\begin{lemma}
\label{lem:log_grow}
Let $T>0$, and let $\omega\in L^\infty([0,T]; L^1(\R^2)\cap L^p(\R^2))$, for $2<p\leq \infty$, be a weak solution to the Euler equation in the sense of \cref{def:Eul_ws}. Suppose also that $\omega(0)$ satisfies the logarithmic growth bound
\begin{equation}
\label{eq:log_grow}
\int_{\R^2}\ln\jp{x}\omega(0,x)dx<\infty.
\end{equation}
Then
\begin{equation}
\int_{\R^2}\ln\jp{x}\omega(t,x)dx \leq t \|u\|_{L_{t,x}^\infty([0,t]\times\R^2)} \|\omega\|_{L_t^\infty L_x^1([0,t]\times\R^2)}, \qquad \forall t\in[0,T],
\end{equation}
\end{lemma}

For a final property, we observe that weak solutions to the Euler equation in the sense of \cref{def:Eul_ws} conserve the Hamiltonian.

\begin{lemma}
\label{lem:e_con}
Let $T>0$, and let $\omega\in L^\infty([0,T]; L^1(\R^2)\cap L^p(\R^2))$, for $2<p\leq \infty$, be a weak solution to the Euler equation in the sense of \cref{def:Eul_ws}, such that $\omega(0)$ satisfies the growth condition \eqref{eq:log_grow}. Then
\begin{equation}
\int_{\R^2}\g(x-y)\omega(t,x)\omega(t,y)dxdy = \int_{\R^2}\g(x-y)\omega(0,x)\omega(0,y)dxdy, \qquad \forall 0\leq t\leq T.
\end{equation}
\end{lemma}
\begin{proof}
See \cite{CFLS2016}.
\end{proof}

If the velocity field $u$ associated to a weak solution $\omega$ is log-Lipschitz (for instance, $\omega\in L^\infty([0,T]; L^1(\R^2)\cap L^\infty(\R^2))$), then the solution $\omega$ is unique. This is a famous result of Yudovich \cite{Yudovich1963}. In fact, his result shows that given initial datum $\omega^0\in L^1(\R^2)\cap L^\infty(\R^2)$, there is a unique global weak solution $\omega \in L^\infty([0,T];L^1(\R^2)\cap L^\infty(\R^2))$ to the equation \eqref{eq:Eul}, as the next theorem asserts.

\begin{thm}[{\cite{Yudovich1963}}]
\label{thm:Yud}
Let $\omega^0\in L^1(\R^2)\cap L^\infty(\R^2)$. Then there exists a unique solution $\omega \in L^\infty([0,\infty); L^1(\R^2)\cap L^\infty(\R^2))$ to \eqref{eq:Eul} in the sense of \cref{def:Eul_ws}. Moreover, there exists a unique continuous map $\psi: [0,\infty)\times\R^2\rightarrow\R^2$ such that
\begin{equation}
\begin{split}
\psi(t,x) &= x + \int_0^t u(t',\psi(t',x))dt' \qquad \forall (t,x)\in [0,\infty)\times\R^2,
\end{split}
\end{equation}
where $u$ is the velocity field associated $\omega$, and
\begin{equation}
\omega(t,x) = \omega^0(\psi^{-1}(t,x)).
\end{equation}
Additionally, for each $t\geq 0$, $\psi(t,\cdot)$ is a measure-preserving homeomorphism and there exists an absolute constant $C>0$ such that if $|x-y|\leq \exp(1-\int_0^t \|u(t')\|_{LL(\R^2)}dt')$, then
\begin{align}
\left|\psi(t,x)-\psi(t,y)\right| &\leq C\exp\paren*{1-\exp\paren*{\int_0^t \|u(t')\|_{LL(\R^2)}dt'}}|x-y|^{\exp(-\int_0^t \|u(t')\|_{LL(\R^2)}dt')}, \label{eq:Hold}\\
\left|\psi^{-1}(t,x)-\psi^{-1}(t,y)\right| &\leq C\exp\paren*{1-\exp\paren*{\int_0^t \|u(t')\|_{LL(\R^2)}dt'}}|x-y|^{\exp(-\int_0^t \|u(t')\|_{LL(\R^2)}dt')}.
\end{align}
\end{thm}

\begin{remark}
\label{rem:Lp_con}
Since $\psi(t,\cdot)$ is measure-preserving for every $t\geq 0$, an immediate consequence of \cref{thm:Yud} is that the $L^p$ norms of $\omega$, for any $1\leq p\leq \infty$, are conserved.
\end{remark}

\section{The Modulated Energy}
\label{sec:CE}
The goal of this section is to introduce the modulated energy and its renormalization, which we use to measure the distance between the $N$-body empirical measure $\omega_N$ and the mean-field measure $\omega$.

\subsection{Setup}
\label{ssec:CE_set}
Recall from the introduction that $\g(x)\coloneqq -\frac{1}{2\pi}\ln|x|$ is the 2D Coulomb potential. Given $\eta>0$, we define the \emph{truncation to distance $\eta$} of $\g$ by
\begin{equation}
\label{eq:g_trunc}
\g_\eta:\R^2\rightarrow\R, \qquad \g_\eta(x) \coloneqq \begin{cases} \g(x), & |x|\geq \eta \\ \tl{\g}(\eta), & |x|<\eta \end{cases},
\end{equation}
where we have introduced the notation $\g(x) = \tl{\g}(|x|)$ to reflect that $\g$ is a radial function. Evidently, $\g_\eta$ is a continuous function on $\R^2$ and decreases like $\g$ as $|x|\rightarrow\infty$. The next lemma provides us with some identities for the distributional gradient and Laplacian of $\g_\eta$, of which we shall make heavy use in the sequel.

\begin{lemma}
\label{lem:g_id}
For any $\eta>0$, we have the distributional identities
\begin{align}
(\nabla\g_\eta)(x) &= -\frac{x}{2\pi|x|^2}1_{\geq \eta}(x), \label{eq:g_eta_grad_id}\\
(\D\g_\eta)(x) &= -\sigma_{\p B(0,\eta)}(x), \label{eq:g_eta_lapl_id}
\end{align}
where $\sigma_{\p B(0,\eta)}$ is the uniform probability measure on the sphere $\p B(0,\eta)$.
\end{lemma}
\begin{proof}
Fix $\eta>0$, and let $u\in C_c^\infty(\R^2;\R^2)$. Then by definition of the distributional gradient,
\begin{align}
\label{eq:g_lhs_ibp}
\ipp{\nabla\g_\eta,u} &= -\int_{\R^2}\g_\eta(x) (\nabla\cdot u)(x)dx = -\int_{|x|\geq \eta} \g(x)(\nabla\cdot u)(x)dx - \int_{|x|<\eta}\tl{\g}(\eta)(\nabla\cdot u)(x)dx.
\end{align}
Using the divergence theorem and that $\frac{x}{|x|}$ is the outward unit normal to a ball $B(0,\eta)$, we find that the second term in left-hand side of \eqref{eq:g_lhs_ibp} equals
\begin{equation}
\label{eq:g_grad_in}
-\tl{\g}(\eta)\int_{\p B(0,\eta)}u(x)\cdot\frac{x}{|x|}d\sigma_{\p B(0,\eta)|}(x).
\end{equation}
Similarly, using that the outward unit normal to the region $B(0,\eta)^c$ is $-\frac{x}{|x|}$, the first term in the left-hand side of \eqref{eq:g_lhs_ibp} equals
\begin{equation}
\label{eq:g_grad_out}
\begin{split}
&\int_{\p B(0,\eta)}\g(x)u(x)\cdot \frac{x}{|x|}d\sigma_{\p B(0,\eta)}(x) +\int_{|x|\geq\eta}(\nabla\g)(x)\cdot u(x)dx\\
& = \tl{\g}(\eta)\int_{\p B(0,\eta)}u(x)\cdot\frac{x}{|x|}d\sigma_{\p B(0,\eta)}(x)+\int_{|x|\geq\eta}(\nabla\g)(x)\cdot u(x)dx.
\end{split}
\end{equation}
Combining identities \eqref{eq:g_grad_in} and \eqref{eq:g_grad_out}, we obtain that
\begin{equation}
\ipp{\nabla\g_\eta,u} = \int_{|x|\geq\eta} (\nabla\g)(x)\cdot u(x)dx.
\end{equation}
Since by the chain rule,
\begin{equation}
(\nabla\g)(x) = -\frac{x}{2\pi|x|^2}, \qquad \forall x\in\R^2\setminus\{0\},
\end{equation}
we conclude that
\begin{equation}
(\nabla\g_\eta)(x) = -\frac{x}{2\pi|x|^2}1_{\geq\eta}(x),
\end{equation}
with equality in the sense of distributions. Similarly, for any $\varphi\in C_c^\infty(\R^2)$,
\begin{align}
\ipp{\D\g_\eta,\varphi} =-\int_{\R^2}(\nabla\g_\eta)(x)\cdot(\nabla\varphi)(x)dx = -\int_{|x|\geq\eta} (\nabla\g)(x)\cdot(\nabla\varphi)(x)dx,
\end{align}
where we use identity \eqref{eq:g_eta_grad_id} to obtain the ultimate equality. Integrating by parts, this last expression equals
\begin{equation}
\int_{\p B(0,\eta)}(\nabla\g)(x)\cdot\frac{x}{|x|}\varphi(x)d\sigma_{\p B(0,\eta)}(x).
\end{equation}
Since for every $x\in\p B(0,\eta)$,
\begin{equation}
(\nabla\g)(x) \cdot\frac{x}{|x|} = -\frac{x}{2\pi|x|^2}\cdot\frac{x}{|x|} = -\frac{1}{2\pi\eta},
\end{equation}
we conclude identity \eqref{eq:g_eta_lapl_id}.
\end{proof}

With \cref{lem:g_id}, we can define the \emph{smearing to scale $\eta$} of the Dirac mass $\delta_0$ by
\begin{equation}
\label{eq:delta_smear}
\d_0^{(\eta)} \coloneqq -\Delta \g_\eta = \sigma_{\p B(0,\eta)}.
\end{equation}
A useful identity satisfied by $\d_0^{(\eta)}$ is
\begin{equation}
\label{eq:g_conv_smear}
(\g\ast \d_0^{(\eta)})(x) = \g_\eta(x),
\end{equation}
which follows from the definition \eqref{eq:delta_smear}, the associativity and commutativity of convolution, and the fact that $\g$ is a fundamental solution of the operator $-\D$.

Next, given parameters $\infty>\eta,\alpha>0$, we define the function
\begin{equation}
\label{eq:f_et_al_def}
\f_{\eta,\alpha}(x) \coloneqq (\g_{\alpha}-\g_\eta)(x), \qquad \forall x\in\R^2.
\end{equation}
From identity \eqref{eq:g_conv_smear} and the bilinearity of convolution, we see that $\f_{\et,\al}$ satisfies the identity
\begin{equation}
\label{eq:f_conv_smear}
\f_{\eta,\alpha} = \g\ast (\d_0^{(\alpha)}-\d_0^{(\eta)}).
\end{equation}
Additionally, for $\alpha>\eta$, we find by direct computation that
\begin{equation}
\label{eq:f_et_al_id}
\f_{\eta,\alpha}(x) = \begin{cases}0, & {|x|\geq\alpha}\\ \\ -\frac{1}{2\pi}\ln(\frac{\alpha}{|x|}), & {\eta\leq|x|\leq\alpha} \\ \\ -\frac{1}{2\pi}\ln(\frac{\alpha}{\eta}), & {|x|<\eta} \end{cases}
\end{equation}
and
\begin{equation}
\label{eq:f_grad_et_al_id}
(\nabla\f_{\eta,\alpha})(x) = -\frac{1}{2\pi}\frac{x}{|x|^2}1_{\eta\leq\cdot\leq\alpha}(x),
\end{equation}
with equality in the sense of distributions. The next lemma provides useful estimates for the $L^p$ norms of $\f_{\eta,\alpha}, \nabla\f_{\eta_\alpha}$. We shall use these estimates extensively in \cref{sec:kprop}.

\begin{lemma}
\label{lem:f_et_al_bnds}
For any $1\leq p\leq\infty$ and $\infty>\alpha>\eta>0$, we have that
\begin{equation}
\label{eq:f_et_al_bnd}
\|\f_{\eta,\alpha}\|_{L^p(\R^2)} \leq \frac{\alpha^{2/p}}{(2\pi)^{(p-1)/p}}\paren*{\int_{\frac{\eta}{\alpha}}^1 |\ln r|^prdr}^{1/p}+ \frac{\pi^{1/p}\eta^{2/p}}{2\pi}\ln\frac{\alpha}{\eta}
\end{equation}
and
\begin{equation}
\label{eq:f_grad_et_al_bnd}
\|\nabla\f_{\eta,\alpha}\|_{L^p(\R^2)} = \begin{cases} \frac{(\alpha^{2-p}-\eta^{2-p})^{1/p}}{(2\pi)^{\frac{p-1}{p}}(2-p)^{1/p}},& {1\leq p<2}\\  \\ \sqrt{\frac{1}{2\pi}\ln\frac{\alpha}{\eta}}, & {p=2}\\ \\ \frac{(\eta^{2-p}-\alpha^{2-p})^{1/p}}{(2\pi)^{\frac{p-1}{p}}(p-2)^{1/p}} ,& {p>2} . \end{cases}
\end{equation}
\end{lemma}
\begin{proof}
We first show estimate \eqref{eq:f_et_al_bnd}. The case $p=\infty$ is immediate from identity \eqref{eq:f_et_al_id}, so we consider the case $1\leq p<\infty$. Using identity \eqref{eq:f_et_al_id} and polar coordinates, we find that
\begin{align}
\int_{\R^2} |\f_{\eta,\alpha}(x)|^pdx &= \int_{\eta\leq|x|\leq\alpha}|\f_{\eta,\alpha}(x)|^pdx + \int_{|x|<\eta} |\f_{\eta,\alpha}(x)|^pdx \nn\\
&=\frac{2\pi}{(2\pi)^p}\int_\eta^\alpha \left|\ln\frac{\alpha}{r}\right|^p rdr + \frac{\pi\eta^2}{(2\pi)^p} \left|\ln\frac{\alpha}{\eta}\right|^p.
\end{align}
Using dilation invariance of Lebesgue measure, the first term in the second line equals
\begin{equation}
\frac{\alpha^2}{(2\pi)^{p-1}}\int_{\frac{\eta}{\alpha}}^1 |\ln r|^p rdr.
\end{equation}
Thus, using the inclusion $\ell^{1/p}\subset\ell^1$, for $1\leq p<\infty$, we obtain that
\begin{equation}
\|\f_{\eta,\alpha}\|_{L^p(\R^2)} \leq  \frac{\alpha^{2/p}}{(2\pi)^{(p-1)/p}}\paren*{\int_{\frac{\eta}{\alpha}}^1 |\ln r|^{p}rdr}^{1/p} + \frac{\pi^{1/p}\eta^{2/p}}{2\pi}\ln\frac{\alpha}{\eta}.
\end{equation}

We next show estimate \eqref{eq:f_grad_et_al_bnd}, again only considering the case $1\leq p<\infty$. Using identity \eqref{eq:f_grad_et_al_id}, polar coordinates, and the fundamental theorem of calculus, we find that for $1\leq p<2$ or $2<p<\infty$,
\begin{align}
\label{eq:p_neq2}
\int_{\R^2}|\nabla\f_{\eta,\alpha}(x)|^pdx = \frac{2\pi}{(2\pi)^p}\int_{\eta}^{\alpha}\frac{r}{r^p}dr = \frac{(\alpha^{2-p} - \eta^{2-p})}{(2\pi)^{p-1}(2-p)}.
\end{align}
If $p=2$, then
\begin{equation}
\label{eq:p_eq2}
\int_{\R^2}|\nabla\f_{\eta,\alpha}(x)|^2dx = \frac{\ln(\alpha)-\ln(\eta)}{2\pi} = \frac{1}{2\pi}\ln\frac{\alpha}{\eta}.
\end{equation}
Taking the $p$-th roots of both sides in \eqref{eq:p_neq2} and square roots in \eqref{eq:p_eq2} completes the proof of the desired identity.
\end{proof}

Next, given a probability measure $\mu$, such that $\int_{\R^2}|\ln|x|| d\mu(x)$, a vector $\ux_N \in (\R^2)^N$, and vector $\ue_N\in (\R_+)^N$, we define the quantities
\begin{align}
H_N^{\mu,\ux_N} &\coloneqq \g \ast (\sum_{i=1}^N\d_{x_i}-N\mu), \label{eq:HN_def}\\
H_{N,\ue_N}^{\mu,\ux_N} &\coloneqq \g\ast (\sum_{i=1}^N\d_{x_i}^{(\eta_i)}-N\mu) \label{eq:HN_trun_def},
\end{align}
where $\d_{x_i}^{(\eta_i)} = \d_{0}^{(\eta_i)}(\cdot-x_i)$.

\subsection{Energy Functional}
\label{ssec:CE_EF}
For a vector $\ux_N\in (\R^2)^N$ and a measure $\mu\in \P(\R^2)\cap L^p(\R^2)$, for some $1<p\leq\infty$, which has the property
\begin{equation}
\int_{(\R^2)^2}\ln\jp{x-y}d\mu(x)d\mu(y) <\infty,
\end{equation}
we define the functional
\begin{equation}
\label{eq:def_EN}
\Fr_N(\ux_N,\mu) \coloneqq \int_{(\R^2)^2\setminus \D_2}\g(x-y)d(\sum_{i=1}^N\d_{x_i}-N\mu)(x)d(\sum_{i=1}^N\d_{x_i}-N\mu)(y)
\end{equation}
where $\D_2\coloneqq \{(x,y)\in (\R^2)^2 : x=y\}$. Note that $\Fr_N(\ux_N,\mu) = N^2\Fr_N^{avg}(\ux_N,\mu)$, where $\Fr_N^{avg}(\ux_N,\mu)$ was defined in \eqref{eq:FN_avg_def}. The reader can check from our assumptions on $\mu$ and \cref{lem:PE_bnds} that $\Fr_N(\ux_N,\mu)$ is well-defined. This functional serves as a $\dot{H}^1(\R^2)$ squared-distance that has been renormalized so as to remove the infinite-self interaction between elements of $\ux_N$. Our first lemma computes the Coulomb energy of the smeared point mass $\d_{0}^{(\eta)}$.

\begin{lemma}
\label{lem:g_smear_si}
For any $0<\eta<\infty$, we have that
\begin{equation}
\int_{(\R^2)^2}\g(x-y)d\d_0^{(\eta)}(x)d\d_0^{(\eta)}(y) = \tl{\g}(\eta).
\end{equation}
\end{lemma}
\begin{proof}
Recalling from \eqref{eq:delta_smear} that $\delta_0^{(n)}$ is the uniform probability measure on the sphere $\p B(0,\eta)$, we find by changing to spherical coordinates that
\begin{equation}
\int_{(\R^2)^2}\g(x-y)d\d_0^{(\eta)}(x)d\d_0^{(\eta)}(y) = -\frac{1}{(2\pi)(2\pi\eta)^2}\int_{0}^{2\pi}\int_{0}^{2\pi}\ln|\eta(\cos\theta-\cos\theta')| \eta^2 d\theta d\theta'.
\end{equation}
By rotational symmetry and using the product-to-sum property of the logarithm, the right-hand side simplifies to
\begin{equation}
\tl{\g}(\eta) - \frac{1}{2\pi}\int_0^{2\pi} \ln(\sqrt{(1-\cos\theta)^2 + \sin^2\theta})d\theta = \tl{\g}(\et)-\frac{1}{2\pi}\int_{-\pi}^{\pi} \ln(\sqrt{(1-\cos\theta)^2 + \sin^2\theta})d\theta.
\end{equation}
We claim that the second term is zero. Indeed, since the function $z\mapsto \ln|1-z|$ is harmonic in $\C\setminus\{1\}$, we find from the mean-value property of harmonic functions that for any $0<r<1$,
\begin{equation}
\label{eq:mvp}
\frac{1}{2\pi}\int_{-\pi}^{\pi}\ln|1-re^{i\theta}|d\theta = \ln 1 =0.
\end{equation}
Now
\begin{equation}
|\ln|1-re^{i\theta}|| \lesssim 1+ |\ln|1-e^{i\theta}||,
\end{equation}
and we claim that the function on the right-hand side is in $L^1([-\pi,\pi])$. Indeed, it suffices to check in a neighborhood $[-\varepsilon,\varepsilon]$, for $0<\varepsilon\ll 1$. We observe from using the Taylor series for $\cos \theta$ that
\begin{equation}
\ln|1-e^{i\theta}| = \ln|2(1-\cos \theta)| = \ln\left|2\sum_{k=1}^\infty \frac{(-1)^{k+1}\theta^{2k}}{(2k)!}\right| = 2\ln|\theta| + \ln\left|2\sum_{k=0}^\infty \frac{(-1)^{k}\theta^{2k}}{(2(k+1))!}\right|.
\end{equation}
Provided we choose $\varepsilon>0$ sufficiently small, the second term in the right-hand side is bounded for $\theta\in [-\varepsilon,\varepsilon]$. Since $\ln |\theta|$ is locally integrable on the real-line, we obtain the claim. Thus, by dominated convergence and letting $r\rightarrow 1^-$ in \eqref{eq:mvp}, we conclude the proof of the lemma.
\end{proof}

\begin{lemma}
\label{lem:fin_en}
Fix $N\in\N$. Let $\mu\in L^p(\R^2)$, for some $2<p\leq \infty$, such that $\int_{\R^2}\ln\jp{x}|\mu(x)|dx<\infty$, and let $\ux_N\in (\R^2)^N\setminus\D_N$. Then for any $\ue_N\in (\R_+)^N$, we have the identity
\begin{equation}
\label{eq:renorm_nts}
\begin{split}
&\int_{(\R^2)^2}\g(x-y)d(N\mu-\sum_{i=1}^N\d_{x_i}^{(\eta_i)})(x)d(N\mu-\sum_{i=1}^N\d_{x_i}^{(\eta_i)})(y) =\int_{\R^2} |(\nabla H_{N,\ul{\eta}_N}^{\mu,\ux_N})(x)|^2dx,
\end{split}
\end{equation}
In particular, the right-hand side is finite if and only if $\mu$ has finite Coulomb energy.
\end{lemma}
\begin{proof}
See the beginning of the proof of \cite[Proposition 3.3]{Serfaty2018mean}.
\end{proof}

The next proposition is essentially proven in \cite[Section 2.1]{PS2017} and \cite[Section 5]{Serfaty2018mean} in the greater generality of Riesz, not just Coulomb, interactions. We include a self-contained proof specialized to our setting.

\begin{prop}
\label{prop:CE}
Let $\mu \in \P(\R^2)\cap L^p(\R^2)$, for some $2<p\leq\infty$, such that $\int_{\R^2}\ln\jp{x} |\mu(x)|dx<\infty$, and let $\ux_N\in (\R^2)^N\setminus\D_N$. Then
\begin{equation}
\label{eq:EN_renorm_lim}
\Fr_N(\ux_N,\mu) = \lim_{|\ul{\eta}_N|\rightarrow 0} \paren*{\int_{\R^2}|(\nabla H_{N,\ue_N}^{\mu,\ux_N})(x)|^2dx - \sum_{i=1}^N\tl{\g}(\eta_i)}
\end{equation}
and there exists a constant $C_p>0$, such that
\begin{equation}
\label{eq:EN_renorm_bnd}
\begin{split}
\sum_{1\leq i\neq j\leq N} \paren*{\g(x_i-x_j)-\tl{\g}(\eta_i)}_{+} &\leq \Fr_N(\ux_N,\mu) - \paren*{\int_{\R^2}|(\nabla H_{N,\ul{\eta}_N}^{\mu,\ux_N})(x)|^2dx - \sum_{i=1}^N\tl{\g}(\eta_i)} \\
&\ph + C_pN \|\mu|_{L^p(\R^2)}\sum_{i=1}^N\eta_i^{2(p-1)/p},
\end{split}
\end{equation}
where $(\cdot)_{+}\coloneqq \max\{\cdot,0\}$.
\end{prop}
\begin{proof}
We start by proving \eqref{eq:EN_renorm_lim}. We first claim that
\begin{equation}
\begin{split}
&\int_{(\R^2)^2\setminus\D_2}\g(x-y)d(N\mu-\sum_{i=1}^N\d_{x_i})(x)d(N\mu-\sum_{i=1}^N\d_{x_i})(y) \\
&= \lim_{|\ul{\eta}_N|\rightarrow 0} \paren*{\int_{(\R^2)^2}\g(x-y)d(N\mu-\sum_{i=1}^N\d_{x_i}^{(\eta_i)})(x)d(N\mu-\sum_{i=1}^N\d_{x_i}^{(\eta_i)})(y) - \sum_{i=1}^N \tl{\g}(\eta_i)}.
\end{split}
\end{equation}
Indeed, the left-hand side of the preceding equality can be re-written as
\begin{equation}
\sum_{1\leq i\neq j\leq N} \g(x_i-x_j) + N^2\int_{(\R^2)^2}\g(x-y)d\mu(x)d\mu(y) -2N\sum_{i=1}^N (\g\ast\mu)(x_i).
\end{equation}
Since $x_i\neq x_j$ for $i\neq j$, it follows from dominated convergence that
\begin{equation}
\begin{split}
&\sum_{1\leq i\neq j\leq N}\g(x_i-x_j) -2N\sum_{i=1}^N(\g\ast\mu)(x_i) \\
&= \lim_{|\ul{\eta}_N|\rightarrow 0} \paren*{\sum_{1\leq i\neq j\leq N}\int_{(\R^2)^2}\g(x-y)d\d_{x_i}^{(\eta_i)}(x)d\d_{x_j}^{(\eta_j)}(y) - 2N\sum_{i=1}^N \int_{\R^2} (\g\ast\mu)(x)d\d_{x_i}^{(\eta_i)}(x)}.
\end{split}
\end{equation}
So by applying \cref{lem:g_smear_si} to the first term in the right-hand side, we conclude the proof of the claim.

Next, we show the bound \eqref{eq:EN_renorm_bnd}. Fix $\ul{\eta}_N\in (\R_{+})^N$, and let $\ua_N\in (\R_+)^N$, such that $\alpha_i\ll \eta_i$ for every $i\in\{1,\ldots,N\}$. It is straightforward to check that
\begin{equation}
\nabla H_{N,\ul{\eta}_N}^{\mu,\ux_N} = \nabla H_{N,\ua_N}^{\mu,\ux_N} + \sum_{i=1}^N (\nabla\f_{\alpha_i,\eta_i})(\cdot-x_i).
\end{equation}
Using this identity with a little algebra, we find that
\begin{equation}
\label{eq:main_id}
\begin{split}
\int_{\R^2}|(\nabla H_{N,\ul{\eta}_N}^{\mu,\ux_N})(x)|^2dx &= \int_{\R^2} |(\nabla H_{N,\ua_N}^{\mu,\ux_N})(x)|^2dx + 2\sum_{i=1}^N\int_{\R^2}(\nabla H_{N,\ua_N}^{\mu,\ux_N})(x)\cdot (\nabla\f_{\alpha_i,\eta_i})(x-x_i)dx \\
&\ph+ \int_{\R^2} |\sum_{i=1}^N (\nabla\f_{\alpha_i,\eta_i})(x-x_i)|^2dx.
\end{split}
\end{equation}
We go to work on the second and third terms in the right-hand side of the preceding equality. We expand the square to obtain
\begin{equation}
\int_{\R^2}\left|\sum_{i=1}^N (\nabla\f_{\alpha_i,\eta_i})(x-x_i)\right|^2dx = \sum_{i,j=1}^N \int_{\R^2}(\nabla\f_{\alpha_i,\eta_i})(x-x_i)\cdot(\nabla\f_{\alpha_j,\eta_j})(x-x_j)dx.
\end{equation}
Recalling the identities
\begin{equation}
(\nabla\f_{\alpha_j,\eta_j})(x) = -\frac{x}{2\pi|x|^{2}}1_{\alpha_j\leq\cdot\leq \eta_j}(x), \qquad (-\D\f_{\alpha_j,\eta_j})(x) = \paren*{\d_{0}^{(\eta_j)} - \d_0^{(\alpha_j)}}(x)
\end{equation}
and integrating by parts, we find that
\begin{equation}
\label{eq:renorm_comb1}
\sum_{i,j=1}^N \int_{\R^2}(\nabla\f_{\alpha_i,\eta_i})(x-x_i)\cdot(\nabla\f_{\alpha_j,\eta_j})(x-x_j)dx = \sum_{i,j=1}^N \int_{\R^2} \f_{\alpha_i,\eta_i}(x-x_i)d(\d_{x_j}^{(\eta_j)}-\d_{x_j}^{(\alpha_j)})(x).
\end{equation}
Similarly, using the identity $-\D H_{N,\ua_N}^{\mu,\ux_N} = \sum_{i=1}^N\d_{x_i}^{(\alpha_i)} - N\mu$ and integrating by parts, we find that
\begin{equation}
\label{eq:renorm_comb2}
2\sum_{i=1}^N\int_{\R^2}(\nabla H_{N,\ua_N}^{\mu,\ux_N})(x)\cdot (\nabla\f_{\alpha_i,\eta_i})(x-x_i)dx  = 2\sum_{i=1}^N\int_{\R^2}\f_{\alpha_i,\eta_i}(x-x_i)d(\sum_{i=1}^N\d_{x_i}^{(\alpha_i)}-N\mu)(x).
\end{equation}
Combining identities \eqref{eq:renorm_comb1} and \eqref{eq:renorm_comb2}, we obtain that
\begin{align}
&\sum_{i,j=1}^N \int_{\R^2}(\nabla\f_{\alpha_i,\eta_i})(x-x_i)\cdot(\nabla\f_{\alpha_j,\eta_j})(x-x_j)dx +\int_{\R^2}\left|\sum_{i=1}^N (\nabla\f_{\alpha_i,\eta_i})(x-x_i)\right|^2dx \nn\\
&=\sum_{i,j=1}^N \int_{\R^2} \f_{\alpha_i,\eta_i}(x-x_i)d(\d_{x_j}^{(\eta_j)}-\d_{x_j}^{(\alpha_j)})(x) + 2\sum_{i=1}^N\int_{\R^2}\f_{\alpha_i,\eta_i}(x-x_i)d(\sum_{j=1}^N\d_{x_j}^{(\alpha_j)}-N\mu)(x) \nn\\
&=\sum_{1\leq i\neq j\leq N}\int_{\R^2}\f_{\alpha_i,\eta_i}(x-x_i)d(\d_{x_j}^{(\eta_j)}+\d_{x_j}^{(\alpha_j)})(x) + \sum_{i=1}^N\int_{\R^2}\f_{\alpha_i,\eta_i}(x-x_i)d(\d_{x_i}^{(\eta_i)}+\d_{x_i}^{(\alpha_i)})(x) \nn\\
&\ph-2N\sum_{i=1}^N\int_{\R^2}\f_{\alpha_i,\eta_i}(x-x_i)\mu(x)dx \nn\\
&\eqqcolon \mathrm{Term}_1 + \mathrm{Term}_2 + \mathrm{Term}_3.
\end{align}
We proceed to estimate each of the $\mathrm{Term}_j$ separately.
\begin{itemize}[leftmargin=*]
\item
For $\mathrm{Term}_3$, we have by H\"older's inequality and \cref{lem:f_et_al_bnds} that
\begin{align}
\label{eq:en_T3_fin}
\left|\mathrm{Term}_3\right| &\lesssim N\sum_{i=1}^N \|\f_{\alpha_i,\eta_i}\|_{L^{p'}(\R^2)} \|\mu\|_{L^p(\R^2)} \nn\\
&\lesssim_p N\|\mu\|_{L^p(\R^2)}\sum_{i=1}^N\paren*{\eta_i^{2(p-1)/p}\paren*{\int_{\frac{\alpha_i}{\eta_i}}^1 |\ln r|^{p/(p-1)}rdr}^{(p-1)/p} + \alpha_i^{2(p-1)/p}\ln\frac{\eta_i}{\alpha_i}}.
\end{align}
\item
For $\mathrm{Term}_2$, we observe from a change of variable and the definition \eqref{eq:f_et_al_def} of $\f_{\alpha_i,\eta_i}$ that
\begin{align}
\mathrm{Term}_2 &= \sum_{i=1}^N\int_{\R^2}\f_{\alpha_i,\eta_i}(x)d(\d_0^{(\eta_i)}+\d_0^{(\alpha_i)})(x) \nn\\
&= \sum_{i=1}^N\int_{\R^2}\paren*{\g_{\eta_i}(x)d\d_0^{(\eta_i)}(x)-\g_{\alpha_i}(x)d\d_0^{(\alpha_i)}(x) +\g_{\eta_i}(x)d\d_0^{(\alpha_i)}(x) - \g_{\alpha_i}(x)d\d_0^{(\eta_i)}(x)} \nn\\
&\eqqcolon\mathrm{Term}_{2,1} + \cdots+\mathrm{Term}_{2,4}.
\end{align}
Since for any $r>0$, $\d_0^{(r)}$ is the canonical probability measure on the sphere $\p B(0,r)$, $\g_{r}\equiv \tl{\g}(r)$ on $\ol{B}(0,r))$, and $\al_i<\et_i$, we find that
\begin{equation}
\mathrm{Term}_{2,3} + \mathrm{Term}_{2,4} =0
\end{equation}
and
\begin{equation}
\mathrm{Term}_{2,1} + \mathrm{Term}_{2,2} = \sum_{i=1}^N \paren*{\tl{\g}(\eta_i)-\tl{\g}(\alpha_i)}.
\end{equation}
Hence,
\begin{equation}
\label{eq:en_T2_fin}
\mathrm{Term}_2 = \sum_{i=1}^N \paren*{\tl{\g}(\eta_i)-\tl{\g}(\alpha_i)}.
\end{equation}
\item
Finally, for $\mathrm{Term}_1$, we use that $\f_{\alpha_i,\eta_i}\leq 0$, for every $i\in\{1,\ldots,N\}$, a property which is evident from the identity \eqref{eq:f_et_al_id}, to estimate
\begin{align}
\mathrm{Term}_1 \leq \sum_{1\leq i\neq j\leq N}\int_{\R^2}\f_{\alpha_i,\eta_i}(x-x_i)d\d_{x_j}^{(\alpha_j)}(x) = \sum_{1\leq i\neq j\leq N}\int_{\R^2}\paren*{\g_{\eta_i}(x-x_i)-\g_{\alpha_i}(x-x_i)}d\d_{x_j}^{(\alpha_j)}(x),
\end{align}
where the penultimate expression follows from unpacking the definition of each $\f_{\alpha_i,\eta_i}$. Making a change of variable, we observe that for $1\leq i\neq j\leq N$,
\begin{align}
\int_{\R^2}\paren*{\g_{\eta_i}(x-x_i)-\g_{\alpha_i}(x-x_i)}d\d_{x_j}^{(\alpha_j)}(x) &= \int_{\R^2} \paren*{\g_{\eta_i}(x_j-x_i+y) - \g_{\alpha_i}(x_j-x_i + y)}d\d_{0}^{(\alpha_j)}(y) \nn\\
&=\int_{\R^2} \paren*{\tl{\g}_{\eta_i}(|x_j-x_i+y|) - \tl{\g}_{\alpha_i}(|x_j-x_i+y|)}d\d_{0}^{(\alpha_j)}(y).
\end{align}
Since $\tl{\g}_{\eta_i}, \tl{\g}_{\alpha_i}$ are nonincreasing, it follows from the triangle inequality that for any $y\in \p B(0,\alpha_j)$,
\begin{equation}
\paren*{\tl{\g}_{\eta_i}-\tl{\g}_{\alpha_i}}(|x_j-x_i+y|) \leq \begin{cases}0, & |x_j-x_i+y| \geq \eta_i\\ \tl{\g}(\eta_i) - \tl{\g}_{\alpha_i}(|x_j-x_i|+\alpha_j), & |x_j-x_i+y| < \eta_i \end{cases}.
\end{equation}
Since we also have that
\begin{equation}
\tl{\g}(\eta_i) - \tl{\g}_{\alpha_i}(|x_j-x_i+y|) \leq \begin{cases} \tl{\g}(\eta_i) - \tl{\g}(|x_j-x_i+y|)\leq 0, & {\alpha_i\leq |x_j-x_i+y| <\eta_i} \\ \tl{\g}(\eta_i)-\tl{\g}(\alpha_i)\leq 0, & {|x_j-x_i+y| < \alpha_i} \end{cases}
\end{equation}
by assumption that $\alpha_i<\eta_i$, it follows that
\begin{equation}
\label{eq:en_T1_fin}
\mathrm{Term}_1 \leq -\sum_{1\leq i\neq j\leq N} \paren*{\tl{\g}(|x_j-x_i|+\alpha_j)-\tl{\g}(\eta_i)}_{+}.
\end{equation}
\end{itemize}

Returning to the equation \eqref{eq:main_id} and combining identities \eqref{eq:en_T1_fin}, \eqref{eq:en_T2_fin}, and \eqref{eq:en_T3_fin} for $\mathrm{Term}_1$, $\mathrm{Term}_2$, and $\mathrm{Term}_3$, respectively, we find that there exists a constant $C_{p}>0$ such that
\begin{equation}
\begin{split}
\int_{\R^2}|(\nabla H_{N,\ul{\eta}_N}^{\mu,\ux_N})(x)|^2dx &\leq \int_{\R^2} |(\nabla H_{N,\ua_N}^{\mu,\ux_N})(x)|^2dx \\
&\ph + C_pN\|\mu\|_{L^p(\R^2)}\sum_{i=1}^N\paren*{\eta_i^{2(p-1)/p}\paren*{\int_{\frac{\alpha_i}{\eta_i}}^1 |\ln r|^{p/(p-1)}rdr}^{(p-1)/p} + \alpha_i^{2(p-1)/p}\ln\frac{\eta_i}{\alpha_i}}\\
&\ph +\sum_{i=1}^N \paren*{\tl{\g}(\eta_i)-\tl{\g}(\alpha_i)} -  \sum_{1\leq i\neq j\leq N} \paren*{\tl{\g}(|x_j-x_i|+\alpha_j)-\tl{\g}(\eta_i)}_{+}.
\end{split}
\end{equation}
We can rearrange this inequality to
\begin{equation}
\begin{split}
&\sum_{1\leq i\neq j\leq N} \paren*{\tl{\g}(|x_j-x_i|+\alpha_j)-\tl{\g}(\eta_i)}_{+} \\
&\leq \paren*{\int_{\R^2} |(\nabla H_{N,\ua_N}^{\mu,\ux_N})(x)|^2dx - \sum_{i=1}^N \tl{\g}(\alpha_i)} - \paren*{\int_{\R^2}|(\nabla H_{N,\ul{\eta}_N}^{\mu,\ux_N})(x)|^2dx - \sum_{i=1}^N \tl{\g}(\eta_i)} \\
&\ph + C_pN\|\mu\|_{L^p(\R^2)}\sum_{i=1}^N\paren*{\eta_i^{2(p-1)/p}\paren*{\int_{\frac{\alpha_i}{\eta_i}}^1 |\ln r|^{p/(p-1)}rdr}^{(p-1)/p} + \alpha_i^{2(p-1)/p}\ln\frac{\eta_i}{\alpha_i}}.
\end{split}
\end{equation}
Sending $|\ua_N|\rightarrow 0$ and using  identity \cref{eq:EN_renorm_lim} and that $\tl{\g}\in C_{loc}(\R_+)$, we see that the proof is complete.
\end{proof}

The following corollary is a generalization of \cite[Corollary 3.4]{Serfaty2018mean}, in particular a relaxation of the $\mu\in L^\infty(\R^2)$ assumption. 

\begin{cor}
\label{cor:grad_H}
Fix $N\in\N$. Let $\mu\in L^p(\R^2)$, for some $2<p\leq \infty$, such that $\int_{\R^2}\ln\jp{x}|\mu(x)|dx<\infty$, and let $\ux_N\in (\R^2)^N\setminus \D_N$. If for any $0<\ep_1 \ll 1$, we define
\begin{equation}
r_{i,\ep_1} \coloneqq \min\{\frac{1}{4}\min_{{1\leq j\leq N}\atop {j\neq i}} |x_i-x_j|, \ep_1\} \quad \text{and} \quad  \ur_{N,\ep_1}\coloneqq (r_{1,\ep_1},\ldots,r_{N,\ep_1}),
\end{equation}
then there exists a constant $C_p>0$ such that
\begin{equation}
\label{eq:g_r_sum_bnd}
\sum_{i=1}^N \tl{\g}(r_{i,\ep_1}) \leq \Fr_N(\ux_N,\mu) + N\paren*{2\tl{\g}(\ep_1) -\tl{\g}(4) + C_p\|\mu\|_{L^p(\R^2)}N\ep_1^{\frac{2(p-1)}{p}}}
\end{equation}
and
\begin{equation}
\label{eq:grad_H_r_bnd}
\int_{\R^2}|(\nabla H_{N,\ul{r}_{N,\ep_1}}^{\mu,\ux_N})(x)|^2dx \leq \Fr_N(\ux_N,\mu) + N\tl{\g}(\ep_1) + C_pN^2\|\mu\|_{L^p(\R^2)}\ep_1^{\frac{2(p-1)}{p}}.
\end{equation}
\end{cor}
\begin{proof}
We first show inequality \eqref{eq:g_r_sum_bnd}. Fix $\ep_1>0$ and choose $\eta_i = \ep_1$ in \eqref{eq:EN_renorm_bnd}, for each $i\in\{1,\ldots,N\}$. Hence,
\begin{equation}
\begin{split}
\sum_{1\leq i\neq j\leq N} \paren*{\tl{\g}(|x_i-x_j|)-\tl{\g}(\ep_1)}_{+} &\leq \Fr_N(\ux_N,\mu) - \int_{\R^2} |(\nabla H_{N,\ue_N}^{\mu,\ux_N})(x)|^2dx  +N\tl{\g}(\ep_1) \\
&\ph + C_p\|\mu\|_{L^p(\R^2)}N^2\ep_1^{\frac{2(p-1)}{p}}. \label{eq:LHS_bnd_blow}
\end{split}
\end{equation}
For each $i\in\{1,\ldots,N\}$, we consider two cases:
\begin{enumerate}[(C1)]
\item\label{item:C1}
$4r_{i,\ep_1}<\ep_1$,
\item\label{item:C2}
$4r_{i,\ep_1}\geq \ep_1$.
\end{enumerate}
In \ref{item:C1}, it follows from the definition of $r_{i,\ep_1}$ that there exists $j_i\in\{1,\ldots,N\}\setminus \{i\}$ such that 
\begin{equation}
|x_i-x_{j_i}| = 4 r_{i,\ep_1},
\end{equation}
which implies that
\begin{equation}
\sum_{{1\leq j\leq N}\atop{j\neq i}} \paren*{\tl{\g}(|x_i-x_j|)-\tl{\g}(\ep_1)}_{+} \geq \paren*{\tl{\g}(|x_i-x_{j_i}|)-\tl{\g}(\ep_1)}_{+} = \tl{\g}(4r_{i,\ep_1})-\tl{\g}(\ep_1).
\end{equation}
In \ref{item:C2}, it is immediate from the nonincreasing property of $\tl{\g}$ that
\begin{equation}
\tl{\g}(4r_{i,\ep_1})-\tl{\g}(\ep_1) \leq 0 = \paren*{\tl{\g}(|x_i-x_j|)-\tl{\g}(\ep_1)}_{+}, \qquad \forall j\in\{1,\ldots,N\}\setminus\{i\}.
\end{equation}
Combining both cases, we obtain from \eqref{eq:LHS_bnd_blow} the inequality
\begin{equation}
\begin{split}
\sum_{i=1}^N\paren*{\tl{\g}(4r_{i,\ep_1})-\tl{\g}(\ep_1)} &\leq \sum_{1\leq i\neq j\leq N} \paren*{\tl{\g}(|x_i-x_j|)-\tl{\g}(\ep_1)}_{+} \\
&\leq \Fr_N(\ux_N,\mu) - \int_{\R^2} |(\nabla H_{N,\ue_N}^{\mu,\ux_N})(x)|^2dx  +N\tl{\g}(\ep_1)  + C_p\|\mu\|_{L^p(\R^2)}N^2\ep_1^{\frac{2(p-1)}{p}}.
\end{split}
\end{equation}
Using that $\tl{\g}(4r_{i,\ep_1}) = \tl{\g}(4) + \tl{\g}(r_{i,\ep_1})$ and rearranging the last inequality, we obtain \eqref{eq:g_r_sum_bnd}.

We next show inequality \eqref{eq:grad_H_r_bnd}. We choose $\eta_i=r_{i,\ep_1}$ in inequality \eqref{eq:EN_renorm_bnd}, for each $i\in\{1,\ldots,N\}$, to obtain that
\begin{equation}
\begin{split}
0\leq \sum_{1\leq i\neq j\leq N}\paren*{\g(x_i-x_j)-\tl{\g}(r_{i,\ep_1})}_+ &\leq \Fr_N(\ux_N,\mu) -\int_{\R^2}|(\nabla H_{N,\ur_{N,\ep_1}}^{\mu,\ux_N})(x)|^2dx + \sum_{i=1}^N \tl{\g}(r_{i,\ep_1}) \\
&\ph+ C_pN\|\mu\|_{L^p(\R^2)}\sum_{i=1}^N r_{i,\ep_1}^{\frac{2(p-1)}{p}}.
\end{split}
\end{equation}
Inequality \eqref{eq:grad_H_r_bnd} then follows immediately from a rearrangement of the preceding inequality and using that $r_{i,\ep_1}\leq \ep_1$ by definition.
\end{proof}

\subsection{Counting Lemma}
\label{ssec:CE_CL}
In this subsection, we prove a preliminary lemma which uses the modulated energy $\Fr_N(\ux_N,\mu)$ to count the number of distinct pairs $(i,j)$, such that the distance between $x_i$ and $x_j$ is below a prescribed threshold. A similar result was implicit in the proof of \cite[Proposition 2.3]{Serfaty2018mean}; however, we need our refined version in the proofs of \Cref{lem:kprop_error,lem:kprop_recomb}, which are part of the proof of \cref{prop:key}, below.

\begin{lemma}[Counting lemma]
\label{lem:count}
Fix $N\in\N$. Then there exists a constant $C_p>0$, such that for any $\ux_N\in (\R^2)^N\setminus\D_N$ and $\mu\in\P(\R^2)\cap L^p(\R^2)$, for some $2<p\leq\infty$, with finite Coulomb energy and such that $\int_{\R^2}\ln\jp{x}|\mu(x)|dx<\infty$, we have the cardinality bound
\begin{equation}
\begin{split}
\left|\{(i,j)\in\N^2: i\neq j \enspace \text{and} \enspace |x_i-x_j|\leq  \ep_3\}\right| \lesssim \Fr_N(\ux_N,\mu) + N\tl{\g}(\ep_3) + C_pN^2 \ep_3^{\frac{2(p-1)}{p}},
\end{split}
\end{equation}
for any $0<\ep_3\ll 1$. 
\end{lemma}
\begin{proof}
Fix $\ep_3>0$. For every $i\in\{1,\ldots,N\}$, we chose $\eta_i = 2\ep_3$ in estimate \eqref{eq:EN_renorm_bnd} of \cref{prop:CE} to obtain the cardinality estimate
\begin{align}
&\frac{\ln(2)}{2\pi}|\{(i,j)\in \{1,\ldots,N\}^2 : i\neq j \enspace \text{and} \enspace |x_i-x_j|\leq \ep_3\}|\nn\\
&=\sum_{{1\leq i\neq j\leq N}\atop{|x_i-x_j|\leq\ep_3}} -\tl{\g}(2)\nn\\
&\leq \sum_{{1\leq i\neq j\leq N}\atop{|x_i-x_j|\leq\ep_3}}(\underbrace{\g(x_i-x_j)-\tl{\g}(2\ep_3)}_{\geq 0})_{+} \nn\\
&\leq\Fr_N(\ux_N,\mu) - \int_{\R^2}|(\nabla H_{N,\ul{2\ep_3}_N}^{\mu,\ux_N})(x)|^2dx + \sum_{i=1}^N\tl{\g}(2\ep_3)+ C_pN \|\mu\|_{L^p(\R^2)}\sum_{i=1}^N(2\ep_3)^{2(p-1)/p} \nn\\
&\leq \Fr_N(\ux_N,\mu) + N\tl{\g}(2\ep_3) + C_p' N^2 \|\mu\|_{L^p(\R^2)} {\ep_3}^{\frac{2(p-1)}{p}},
\end{align}
where $C_p'\geq C_p$.
\end{proof}

Using the radial and decreasing properties of the potential $\g$, we obtain from \cref{lem:count} the following corollary.

\begin{cor}
\label{cor:count}
Fix $N\in\N$. Let $2< p\leq \infty$. Then there exists a constant $C_p>0$, such that for any $\ux_N\in (\R^2)^N\setminus\D_N$ and $\mu\in\P(\R^2)\cap L^p(\R^2)$, satisfying $\int_{\R^2}\ln\jp{x}|\mu(x)|dx<\infty$ and having finite Coulomb energy, it holds that
\begin{equation}
\begin{split}
\sum_{{1\leq i\neq j\leq N}\atop{|x_i-x_j|\leq\ep_3}} \g(x_i-x_j) &\lesssim \tl{\g}(2\ep_3)\paren*{\Fr_N(\ux_N,\mu) + N\tl{\g}(2\ep_3) + C_p N^2 \|\mu\|_{L^p(\R^2)} {\ep_3}^{\frac{2(p-1)}{p}}}
\end{split}
\end{equation}
for any $0<\ep_3\ll 1$.
\end{cor}
\begin{proof}
Since we also have that
\begin{align}
&\Bigg(\sum_{{1\leq i\neq j\leq N}\atop{|x_i-x_j|\leq\ep_3}} \g(x_i-x_j)\Bigg) - \tl{\g}(2\ep_3)\left|\{(i',j')\in \{1,\ldots,N\}^2 : i'\neq j' \enspace \text{and} \enspace |x_{i'}-x_{j'}|\leq \ep_3\}\right| \nn\\
&\leq \sum_{{1\leq i\neq j\leq N}\atop{|x_i-x_j|\leq\ep_3}} \paren*{\g(x_i-x_j)-\tl{\g}(2\ep_3)}_{+},
\end{align}
the desired conclusion follows immediately from \cref{lem:count}, \cref{prop:CE}, and rearrangement.
\end{proof}

\subsection{Coerciveness of the Energy}
\label{ssec:CE_coer}
In this final subsection, we show that the functional $\Fr_N(\ux_N,\mu)$ controls convergence in the weak-* topology for the Besov space $B_{2,\infty}^{-1}(\R^2)$,\footnote{The Banach space $(B_{2,\infty}^{-1}(\R^2),\|\cdot\|_{B_{2,\infty}^{-1}(\R^2)})$ is isomorphic to the dual of the Banach space $(B_{2,1}^{1}(\R^2),\|\cdot\|_{B_{2,1}^1(\R^2)})$.} as $N\rightarrow\infty$, which, for this scale of function spaces, is the endpoint space containing the Dirac mass. We begin with a technical lemma for the Fourier transform of $\d_{x_0}^{(\eta)}$, which we recall is the uniform probability measure on the sphere $\p B(x_0,\eta)$.

\begin{lemma}
\label{lem:d_smear_FT}
For any $\eta>0$ and $x_0\in\R^2$, we have the distributional identity
\begin{equation}
\wh{\d_{x_0}^{(\eta)}}(\xi) = \frac{e^{- ix_0\cdot\xi}}{2\pi}\int_0^{2\pi}e^{ i \eta|\xi|\sin(\theta)}d\theta.
\end{equation}
\end{lemma}
\begin{proof}
Observe that
\begin{equation}
\d_{x_0}^{(\eta)}(x) = \frac{1}{\eta}\d_{0}^{(1)}(\frac{x-x_0}{\eta}),
\end{equation}
with translation and dilation to be understood in the distributional sense. We have the identity
\begin{equation}
\wh{\d_{0}^{(1)}}(\xi) = \frac{1}{2\pi}\int_0^{2\pi}e^{i|\xi|\sin(\theta)}d\theta,
\end{equation}
which follows from testing against a Schwartz function, the Fubini-Tonelli theorem, and changing to polar coordinates. The desired conclusion then follows from the dilation and translation/modulation invariance of the Fourier transform.
\end{proof}

The next proposition establishes the coerciveness of the modulated energy $\Fr_N$, in particular that it controls convergence in the weak-* topology on $\M(\R^2)$, and is an endpoint improvement of \cite[Proposition 3.5]{Serfaty2018mean}.

\begin{prop}
\label{prop:bes_conv}
Let $N\in\N$ and $\ux_N\in (\R^2)^N\setminus\D_N$. Then for any $\mu\in \P(\R^2)\cap L^{p}(\R^2)$, for some $2< p\leq \infty$, and $\varphi\in B_{2,1}^{1}(\R^2)$, we have the estimate
\begin{equation}
\label{eq:B2inf_conv}
\begin{split}
\left|\int_{\R^2}\varphi(x)d(\sum_{i=1}^N \d_{x_i}-N\mu)(x)\right| &\lesssim  N\paren*{\frac{\ep_1\|\varphi\|_{B_{2,1}^1(\R^2)}}{\ep_2} + \sum_{k\geq |\log_2\ep_2|} 2^k\|P_k\varphi\|_{L^2(\R^2)}} \\
&\ph +\|\nabla\varphi\|_{L^2(\R^2)}\paren*{\Fr_N(\ux_N,\mu) + N|\ln \ep_1| + C_p\|\mu\|_{L^p(\R^2)} N^2 \ep_1^{\frac{2(p-1)}{p}}}^{1/2},
\end{split}
\end{equation}
for any parameters $0<\ep_1<\ep_2\ll 1$. Consequently, for any $s<-1$, 
\begin{align}
\|\mu-\frac{1}{N}\sum_{i=1}^N\d_{x_i}\|_{H^s(\R^2)} \lesssim_{s,p} |\Fr_N^{avg}(\ux_N,\mu)|^{1/2} + N^{-1/2}|\ln N|^{1/2} + \paren*{1+\|\mu\|_{L^p(\R^2)}}N^{-1/2}, \label{eq:Hs_conv}
\end{align}
and if $\Fr_N^{avg}(\ux_N,\mu)\rightarrow 0$, as $N\rightarrow\infty$, then
\begin{equation}
\frac{1}{N}\sum_{i=1}^N \d_{x_i}\xrightharpoonup[N\rightarrow\infty]{*} \mu \ \text{in $\M(\R^2)$} \label{eq:M_conv}.
\end{equation}
\end{prop}
\begin{proof}
Let $\ue_N\in (\R_+)^N$ be a parameter vector, the precise value of which will be specified momentarily. We write
\begin{equation}
\int_{\R^2}\varphi(x)d(\sum_{j=1}^N \d_{x_j}-N\mu)(x) = \sum_{j=1}^N\int_{\R^2}\varphi(x)d(\d_{x_j}-\d_{x_j}^{(\eta_j)})(x) + \int_{\R^2}\varphi(x)d(\sum_{j=1}^N\d_{x_j}^{(\eta_j)}-N\mu)(x).
\end{equation}
By Plancherel's theorem and \cref{lem:d_smear_FT}, we see that
\begin{equation}
\label{eq:diff_Planc}
\sum_{j=1}^N\int_{\R^2}\varphi(x)d(\d_{x_j}-\d_{x_j}^{(\eta_j)})(x) = \sum_{j=1}^N\int_{\R^2}\wh{\varphi}(\xi)\frac{e^{ ix_j\cdot\xi}}{2\pi}\paren*{\int_{0}^{2\pi}\paren*{1-e^{- i\eta_j|\xi|\sin(\theta)}}d\theta}d\xi.
\end{equation}
By the fundamental theorem of calculus,
\begin{equation}
\label{eq:ftc_app}
\int_{0}^{2\pi}\paren*{1-e^{- i\eta_j|\xi|\sin(\theta)}}d\theta = -\int_0^{2\pi}\paren*{\int_{0}^1 (-i\eta_j|\xi|\sin(\theta))e^{- i s\eta|\xi|\sin(\theta)} ds}d\theta.
\end{equation}
For each $j\in\{1,\ldots,N\}$, we introduce a parameter $0<\al_j\ll 1$ and we decompose
\begin{equation}
\varphi=P_{\leq |\log_2\al_j|}\varphi + P_{>|\log_2\al_j|}\varphi,
\end{equation}
so that
\begin{equation}
\begin{split}
\int_{\R^2}\varphi(x)d(\d_{x_j}-\d_{x_j}^{(\eta_j)})(x) &= \int_{\R^2} (P_{\leq |\log_2\al_j|}\varphi)(x)d(\d_{x_j}-\d_{x_j}^{(\eta_j)})(x) + \int_{\R^2} (P_{>|\log_2\al_j|}\varphi)(x)d(\d_{x_j}-\d_{x_j}^{(\eta_j)})(x).
\end{split}
\end{equation}
For the low-frequency term, we use identity \eqref{eq:ftc_app} and estimate the right-hand side directly to obtain from Plancherel's theorem that
\begin{equation}
\begin{split}
\left|\int_{\R^2} (P_{\leq |\log_2\al_j|}\varphi)(x)d(\d_{x_j}-\d_{x_j}^{(\eta_j)})(x) \right| &\lesssim \et_j\int_{|\xi|\leq 2/\al_j}|\xi| |\wh{\varphi}(\xi)|d\xi.
\end{split}
\end{equation}
By Cauchy-Schwarz, the support of the Littlewood-Paley projectors, and Plancherel's theorem,
\begin{align}
\int_{\R^2} |\xi| |\wh{(P_{\leq 0}\varphi)}(\xi)|d\xi &\lesssim \|P_{\leq 0}\varphi\|_{L^2(\R^2)}, \\
\int_{\R^2} |\xi| |\wh{(P_{k}\varphi)}(\xi)|d\xi &\lesssim 2^{2k}\|P_k\varphi\|_{L^2(\R^2)},
\end{align}
which implies that
\begin{align}
\sum_{j=1}^N \et_j\int_{|\xi|\leq 2/\al_j}|\xi| |\wh{\varphi}(\xi)|d\xi &\lesssim \sum_{j=1}^N \et_j\paren*{\|P_{\leq 0}\varphi\|_{L^2(\R^2)} + \sum_{k=1}^{|\log_2\al_j|} 2^{2k}\|P_k\varphi\|_{L^2(\R^2)}} \nn\\
&\lesssim \sum_{j=1}^N \frac{\eta_j \|\varphi\|_{B_{2,1}^1(\R^2)}}{\al_j}. \label{eq:B21_bnd_lo}
\end{align}
For the high-frequency term, we use identity \eqref{eq:diff_Planc} and Cauchy-Schwarz to crudely estimate
\begin{equation}
\label{eq:B21_bnd_hi}
\begin{split}
\sum_{j=1}^N \left|\int_{\R^2} (P_{>|\log_2\al_j|}\varphi)(x)d(\d_{x_j}-\d_{x_j}^{(\eta_j)})(x)\right| &\lesssim \sum_{j=1}^N \sum_{k\geq|\log_2\al_j|} 2^k \|P_k\varphi\|_{L^2(\R^2)}.
\end{split}
\end{equation}
Note that since $\|\varphi\|_{B_{2,1}^1(\R^2)}<\infty$, this term tends to zero as $|\ua_N|\rightarrow 0$.

Next, we write
\begin{equation}
\int_{\R^2}\varphi(x)d(\sum_{j=1}^N\d_{x_j}^{(\eta_j)}-N\mu)(x) = -\int_{\R^2}\varphi(x)(\D H_{N,\ue_N}^{\mu,\ux_N})(x)dx.
\end{equation}
where the reader will recall the notation $H_{N,\ue_N}^{\mu,\ux_N}$ from \eqref{eq:HN_trun_def}. Integrating by parts once and then using Cauchy-Schwarz, we obtain that
\begin{equation}
\left|-\int_{\R^2}\varphi(x)(\D H_{N,\ue_N}^{\mu,\ux_N})(x)dx\right|  \leq \|\nabla\varphi\|_{L^2(\R^2)} \|\nabla H_{N,\ue_N}^{\mu,\ux_N}\|_{L^2(\R^2)}.
\end{equation}
Choosing $\ue_N=\ur_{N,\ep_1}$, we use estimate \eqref{eq:grad_H_r_bnd} of \cref{cor:grad_H} to obtain that the right-hand side of the preceding inequality is $\lesssim$
\begin{align}
&\|\nabla\varphi\|_{L^2(\R^2)}\paren*{\Fr_N(\ux_N,\mu) + N|\ln \ep_1| + C_p\|\mu\|_{L^p(\R^2)} N^2 \ep_1^{\frac{2(p-1)}{p}}}^{1/2}, \label{eq:H1_bnd}
\end{align}
where the ultimate inequality follows from the fact that $\|\ur_{N,\ep_1}\|_{\ell^\infty} \leq \ep_1$ by definition.

Next, choosing  $\al_j = \ep_2$ every $j\in\{1,\ldots,N\}$, and combining estimates \eqref{eq:B21_bnd_lo}, \eqref{eq:B21_bnd_hi}, and \eqref{eq:H1_bnd}, we conclude that
\begin{equation}
\begin{split}
\left|\int_{\R^2}\varphi(x)d(\sum_{i=1}^N\d_{x_i}-N\mu)(x)\right| &\lesssim N\paren*{\frac{\ep_1\|\varphi\|_{B_{2,1}^1(\R^2)}}{\ep_2} + \sum_{k\geq |\log_2\ep_2|} 2^k\|P_k\varphi\|_{L^2(\R^2)}} \\
&\ph +\|\nabla\varphi\|_{L^2(\R^2)}\paren*{\Fr_N(\ux_N,\mu) + N|\ln \ep_1| + C_p\|\mu\|_{L^p(\R^2)} N^2 \ep_1^{\frac{2(p-1)}{p}}}^{1/2},
\end{split}
\end{equation}
which completes the proof of the estimate \eqref{eq:B2inf_conv}.

For the Sobolev norm bound, observe that $H^{-s}(\R^2)$ is a dense subspace of $B_{2,1}^1(\R^2)$, for $s<-1$. Since $(H^{s}(\R^2))^* \cong H^{-s}(\R^2)$, it follows that
\begin{align}
\|\mu-\frac{1}{N}\sum_{i=1}^N\d_{x_i}\|_{H^{s}(\R^2)} &= \sup_{\|\varphi\|_{H^{-s}(\R^2)}=1} \left|\int_{\R^2}\varphi(x)d(\mu-\frac{1}{N}\sum_{i=1}^N\d_{x_i})(x)\right| \nn\\
&\lesssim \sup_{\|\varphi\|_{H^{-s}(\R^2)}=1} \Bigg(\frac{\ep_1\|\varphi\|_{B_{2,1}^1(\R^2)}}{\ep_2} + \frac{\|\varphi\|_{H^{-s}(\R^2)}}{\ep_2^{s+1}} \nn\\
&\ph + \|\nabla\varphi\|_{L^2(\R^2)}\paren*{\Fr_N^{avg}(\ux_N,\mu) + \frac{|\ln\ep_1|}{N} + C_p\|\mu\|_{L^p(\R^2)}\ep_1^{\frac{2(p-1)}{p}}}^{1/2}\Bigg) \nn\\
&\lesssim \paren*{\Fr_N^{avg}(\ux_N,\mu) + \frac{|\ln\ep_1|}{N} + C_p\|\mu\|_{L^p(\R^2)}\ep_1^{\frac{2(p-1)}{p}}}^{1/2} + \frac{\ep_1}{\ep_2}+ \frac{1}{\ep_2^{s+1}},
\end{align}
where the penultimate inequality follows from \eqref{eq:B2inf_conv} and the definition of the $H^{-s}$ norm and the ultimate inequality from the estimate $\|\varphi\|_{B_{2,1}^1(\R^2)}\lesssim \|\varphi\|_{H^{-s}(\R^2)}$. Next, we choose $\ep_2 =\ep_1^{-1/s}$ and 
\begin{equation}
\ep_1 = \min\{N^{-\frac{p}{2(p-1)}}, N^{-\frac{s}{2(s+1)}}\}
\end{equation}
in order to obtain the inequality
\begin{equation}
\|\mu-\frac{1}{N}\sum_{i=1}^N\d_{x_i}\|_{H^s(\R^2)} \lesssim_{s,p} |\Fr_N^{avg}(\ux_N,\mu)|^{1/2} + N^{-1/2}|\ln N|^{1/2} + \paren*{1+\|\mu\|_{L^p(\R^2)}}N^{-1/2}.
\end{equation}

For the weak-* convergence, we recall that $(\M(\R^2),\|\cdot\|_{TV})$, where $\|\cdot\|_{TV}$ is the total variation norm, is the Banach space dual to $(C_0(\R^2),\|\cdot\|_{\infty})$, the space of continuous function equipped with the uniform norm. Let $\varphi\in C_0(\R^2)$. Since $H^{-s}(\R^2)$, for $s<-1$, is a dense subspace of $C_0(\R^2)$, given any $\vep>0$, there exists $\varphi_\vep\in H^{-s}(\R^2)$ such that
\begin{equation}
\|\varphi-\varphi_\vep\|_{\infty}\leq \frac{\vep}{4}.
\end{equation}
Then by the triangle inequality,
\begin{equation}
\left|\int_{\R^2}\varphi(x)d(\mu-\frac{1}{N}\sum_{i=1}^N\d_{x_i})(x)\right| \leq \left|\int_{\R^2}\varphi_\vep(x)d(\mu-\frac{1}{N}\sum_{i=1}^N\d_{x_i})(x)\right| + \frac{\vep}{2}.
\end{equation}
Using estimate \eqref{eq:Hs_conv}, we see that the right-hand side is bounded by
\begin{equation}
\frac{\vep}{2} + C_{s,p}\|\varphi_{\vep}\|_{H^{-s}(\R^2)}\paren*{|\Fr_N^{avg}(\ux_N,\mu)|^{1/2} + N^{-1/2}|\ln N|^{1/2} + \paren*{1+\|\mu\|_{L^p(\R^2)}}N^{-1/2}},
\end{equation}
where $C_{s,p}>0$ is a constant depending only on $s$ and $p$. It then follows immediately that by choosing $N=N(\vep,s,p)\in\N$ sufficiently large, this expression is $<\vep$, which completes the proof.
\end{proof}

\section{Key Proposition}
\label{sec:kprop}
In this section, we prove \cref{prop:key}, which is the key ingredient for the proof of our main theorem. We recall the statement of the proposition below.
\key*

The method of proof is inspired by that of \cite[Proposition 2.3]{Serfaty2018mean}, but requires a more sophisticated analysis, as we no longer have at our disposal the assumption that $\nabla v\in L^\infty(\R^2;(\R^2)^{\otimes 2})$. To make the overall argument modular and easier for the reader to digest, we have divided the proof into several lemmas corresponding to the main steps of the argument. We recall from \cref{ssec:intro_road} that the main steps are
\begin{enumerate}[(S1)]
\item\label{item:moll}
Mollification,
\item\label{item:renorm}
Renormalization,
\item\label{item:diag}
Analysis of Diagonal Terms,
\item\label{item:recomb}
Recombination,
\item\label{item:conc}
Conclusion.
\end{enumerate}

In the new \ref{item:moll}, we introduce a parameter $0<\ep_2\ll 1$ and replace the vector field $v$ with the mollified vector field $v_{\ep_2} \coloneqq v\ast\chi_{\ep_2}$, for a standard approximate identity $\chi_{\ep_2}(x)\coloneqq {\ep_2}^{-2}\chi({\ep_2}^{-1}x)$. Using the log-Lipschitz regularity of $v$, we have an $L^\infty$ bound for $v-v_{\ep_2}$ (see \cref{lem:conv_bnds}), and we find with \cref{lem:kprop_error} that the error
\begin{equation}
\left|\int_{(\R^2)^2\setminus\D_2}\paren*{v-v_{\ep_2}}(x) \cdot (\nabla\g)(x-y) d(\sum_{i=1}^N\d_{x_i}-N\mu)(x)d(\sum_{i=1}^N\d_{x_i}-N\mu)(y)\right|
\end{equation}
is acceptable. For the remaining steps in the proof of \cref{prop:key}, we work with $v_{\ep_2}$ and now need to bound the quantity
\begin{equation}
\label{eq:intro_ntb}
\int_{(\R^2)^2\setminus\D_2}\paren*{v_{\ep_2}(x)-v_{\ep_2}(y)}\cdot(\nabla\g)(x-y) (\sum_{i=1}^N\d_{x_i}-N\mu)(x)d(\sum_{j=1}^N\d_{x_j}-N\mu)(y).
\end{equation}

In \ref{item:renorm}, we follow the renormalization procedure used in \cite{Serfaty2018mean} and prove \cref{lem:kprop_renorm}. We smear the Dirac masses $\d_{x_i},\d_{x_j}$, so that they are replaced by the probability measure $\d_{x_i}^{(\eta_i)},\d_{x_j}^{(\eta_j)}$ supported on the spheres $\p B(x_i,\eta_i), \p B(x_j,\et_j)$,\footnote{We note that the truncation parameter $\ue_N$ is allowed to vary in $i$. This was an important new contribution from \cite{Serfaty2018mean}, which we shall also make use of in our work.} in the quantity \eqref{eq:intro_ntb} and we then add back in the diagonal $\D_2$. The resulting quantity is divergent, but can be renormalized by subtracting off the self-interaction of the smeared point masses. An integration by parts (see \cref{lem:SE_div}) allows us to conclude that
\begin{equation}
\label{eq:intro_renorm}
\begin{split}
\eqref{eq:intro_ntb} &=\lim_{|\ue_N|\rightarrow 0} \Bigg(\int_{\R^2}(\nabla v_{\ep_2})(x): \comm{H_{N,\ue_N}^{\mu,\ux_N}}{H_{N,\ue_N}^{\mu,\ux_N}}_{SE}(x)dx \\
&\ph\hspace{25mm}- \sum_{i=1}^N\int_{(\R^2)^2} \paren*{v_{\ep_2}(x)-v_{\ep_2}(y)}\cdot(\nabla\g)(x-y)\d_{x_i}^{(\eta_i)}(x)\d_{x_i}^{(\eta_i)}(y) \Bigg),
\end{split}
\end{equation}
where $\comm{\cdot}{\cdot}_{SE}$ denotes the stress-energy tensor defined in \eqref{eq:se_ten_def} below.

In \ref{item:diag}, we consider the divergent diagonal terms in \eqref{eq:intro_renorm} (i.e. the second term in the right-hand side), which correspond to self-interaction of the smeared point masses $\d_{x_i}^{(\eta_i)}$, as the smearing parameter vector $\ue_N$ varies. Here, we can use the analysis from \cite{Serfaty2018mean}, as the proof does not rely on any special assumption for the vector field $v_{\ep_2}$. The conclusion of this step is summarized in \cref{lem:kprop_diag}.

In \ref{item:recomb}, we prove \cref{lem:kprop_recomb}, which together with \cref{lem:kprop_diag} from \ref{item:diag}, shows that the difference between the renormalized expressions \eqref{eq:intro_renorm} obtained for different values of $\ue_N$ equals an acceptable error which depends on $N$ and $\ep_2$. The proof of \cref{lem:kprop_recomb} is quite involved, so we will not go into too much detail here. We only say that it relies crucially on properties of the modulated energy established in \cref{sec:CE}.

Finally, in \ref{item:conc}, we perform a bookkeeping of our work from \ref{item:moll}-\ref{item:recomb} to see that the difference between the right-hand sides of \eqref{eq:intro_renorm} obtained for different values $\ue_N = \ua_N^{(2)}\geq \ua_N^{(1)} \in (\R_+)^N$ is small. By letting $|\ua_N^{(1)}|\rightarrow 0$ and choosing, for each $i\in\{1,\ldots,N\}$,
\begin{equation}
\ua_i^{(2)} = r_{i,\ep_1}\coloneqq \min\{\min_{{1\leq j\leq N}\atop{j\neq i}}|x_i-x_j|, \ep_1\},
\end{equation}
where have introduced an additional parameter $0<\ep_1< \ep_2/2$, we find that expression \eqref{eq:intro_ntb} equals
\begin{equation}
\label{eq:intro_conc_nte}
\begin{split}
&\int_{\R^2}\nabla v_{\ep_2}(x) : \comm{H_{N,\ur_{N,\ep_1}}^{\mu,\ux_N}}{H_{N,\ur_{N,\ep_1}}^{\mu,\ux_N}}_{SE}(x)dx -\sum_{i=1}^N\int_{(\R^2)^2} \paren*{v_{\ep_2}(x)-v_{\ep_2}(y)}\cdot(\nabla\g)(x-y)\d_{x_i}^{(r_{i,\ep_1})}(x)\d_{x_i}^{(r_{i,\ep_1})}(y) \\
&\ph + \mathrm{Error}_{N,\ep_2},
\end{split}
\end{equation}
where $\mathrm{Error}_{N,\ur_{N,\ep_1}}$ is an acceptable error. We estimate the first two terms in \eqref{eq:intro_conc_nte} using the log-Lipschitz regularity of $v_{\ep_2}$ (see \cref{lem:conv_bnds}) and properties of the modulated energy (see \cref{lem:g_smear_si} and \cref{cor:grad_H}). In particular, we rely on the estimate
\begin{equation}
\|\nabla v_{\ep_2}\|_{L^\infty(\R^2)} \lesssim \|v\|_{LL(\R^2)}|\ln\ep_2|,
\end{equation}
which necessitates the initial mollification. So far, the reader may be wondering where the additional parameter $\ep_3$ arise in the proof. It is a parameter needed to measure whether two points $x_i$ and $x_j$ are ``close'' or ``far'' (e.g., see \eqref{eq:moll_cl_far}). The floating parameters $\ep_1,\ep_2,\ep_3$, which are new to our work, will crucially be needed to balance terms in order to conclude the proof of \cref{thm:main}.

\subsection{Stress-Energy Tensor}\label{ssec:kprop_SET}
We begin by quickly recalling the definition of the stress-energy tensor. For functions $\varphi,\psi\in C_{loc}^1(\R^2)$, we define their \emph{stress-energy tensor} $\{\comm{\varphi}{\psi}_{SE}^{ij}\}_{i,j=1}^2$ to be the $2\times 2$ matrix with entries
\begin{equation}
\label{eq:se_ten_def}
\comm{\varphi}{\psi}_{SE}^{ij} \coloneqq \paren*{\p_i\varphi\p_j\psi + \p_j\varphi\p_i\psi} - \d_{ij}\nabla\varphi\cdot\nabla\psi,
\end{equation}
where $\d_{ij}$ is the Kronecker delta function. By a density argument, it follows that the stress-energy tensor is well-defined in $L^2(\R^2; (\R^2)^{\otimes 2})$ for any functions $\varphi,\psi\in \dot{H}^1(\R^2)$. By direct computation, we have the divergence identity
\begin{equation}
\p_i\comm{\varphi}{\psi}_{SE}^{ij} = \D \varphi \p_j\psi + \D\psi\p_j\varphi, \qquad \forall j\in\{1,2\},
\end{equation}
for any $\varphi,\psi \in \dot{H}^1(\R^2)\cap \dot{H}^2(\R^2)$, where we have used the convention of Einstein summation in index $j$. The following lemma from \cite{Serfaty2018mean} is used extensively in the sequel. Note the requirement that the vector field be Lipschitz, not log-Lipschitz, which will necessitate a mollification in order to apply the lemma.

\begin{lemma}[{\cite[Lemma 4.3]{Serfaty2018mean}}]
\label{lem:SE_div}
Let $v\in W^{1,\infty}(\R^2;\R^2)$. For any measures $\mu,\nu\in \M(\R^2)$, such that
\begin{equation}
\int_{\R^2}(|(\nabla\g\ast |\mu|)(x)|^2 + |(\nabla\g\ast|\nu|)(x)|^2)dx<\infty,
\end{equation}
we have the identity
\begin{equation}
\int_{(\R^2)^2} \paren*{v(x)-v(y)}\cdot(\nabla\g)(x-y)d\mu(x)d\nu(y) = \int_{\R^2}(\nabla v)(x) : \comm{\g\ast\mu}{\g\ast\nu}_{SE}(x)dx.
\end{equation}
\end{lemma}

\subsection{Step 1: Mollification}
\label{ssec:kprop_moll}
We commence with step 1 of the proof of \cref{prop:key}. Let $\chi\in C_c^\infty(\R^2)$ be a radial, nonincreasing bump function satisfying
\begin{equation}
\label{eq:chi}
\int_{\R^2}\chi(x) dx=1, \quad 0\leq \chi\leq 1, \quad \chi(x) = \begin{cases} 1, & {|x|\leq \frac{1}{4}} \\ 0, & {|x|>1} \end{cases}.
\end{equation}
For $\ep>0$, set
\begin{equation}
\label{eq:chi_v_ep}
\chi_\ep(x)\coloneqq \ep^{-2}\chi(x/\ep) \qquad \text{and} \qquad  v_\ep(x) \coloneqq (\chi_\ep\ast v)(x),
\end{equation}
where the convolution $\chi_{\ep}\ast v$ is performed component-wise. Then $v_\vep$ is $C^\infty(\R^2;\R^2)$ and
\begin{equation}
\|\nabla^k v_\ep\|_{L^\infty(\R^2)} \lesssim_k \ep^{-k}, \qquad \forall k\in\N_0.
\end{equation}
Using the log-Lipschitz regularity of $v$, we can get a bound, which is in terms of the modulated energy $\Fr_N(\ux_N,\mu)$, for the error stemming from replacing $v$ with $v_{\ep}$ in the left-hand side of equation \eqref{eq:prop_key}. The importance of this type of bound will become clear in \cref{ssec:MR_thm} when we allow the mollification parameter to depend on $|\Fr_N(\ux_N,\mu)|$.

\begin{lemma}
\label{lem:kprop_error}
There exists a constant $C_p>0$ such that for every $0<\ep_2,\ep_3\ll 1$, we have the estimate
\begin{equation}
\label{eq:LHS_split}
\begin{split}	
&\left|\int_{(\R^2)^2\setminus\D_2}\paren*{v-v_{\ep_2}}(x) \cdot (\nabla\g)(x-y) d(\sum_{i=1}^N\d_{x_i}-N\mu)(x)d(\sum_{i=1}^N\d_{x_i}-N\mu)(y)\right| \\
&\lesssim \|v\|_{LL(\R^2)}|\ln\ep_3|\paren*{\Fr_N(\ux_N,\mu) + N|\ln\ep_3|+C_pN^2\|\mu\|_{L^p(\R^2)}{\ep_3}^{\frac{2(p-1)}{p}}} \\
&\ph+N^2\|v\|_{LL(\R^2)}\ep_2|\ln\ep_2|\paren*{\frac{1}{\ep_3}+\|\mu\|_{L^p(\R^2)}^{\frac{p}{2(p-1)}}}.
\end{split}
\end{equation}
\end{lemma}
\begin{proof}
We split the left-hand side of \eqref{eq:LHS_split} into a sum of four terms and estimate each separately.
\begin{itemize}[leftmargin=*]
\item
Observe that
\begin{align}
&\left|\int_{(\R^2)^2\setminus\D_2} (v-v_{\ep_2})(x)\cdot(\nabla\g)(x-y)d(\sum_{i=1}^N\d_{x_i})(x)d(\sum_{i=1}^N\d_{x_i})(x)\right| \nn\\
&= \left|\sum_{1\leq i\neq j\leq N} (v-v_{\ep_2})(x_i)\cdot(\nabla\g)(x_i-x_j)\right| \nn\\
&= \left|\frac{1}{2}\sum_{1\leq i\neq j\leq N} \paren*{(v-v_{\ep_2})(x_i)-(v-v_{\ep_2})(x_j)} \cdot(\nabla\g)(x_i-x_j)\right|, \label{eq:ti_cl_far}
\end{align}
where the ultimate equality follows from swapping $i$ and $j$ and using that $(\nabla\g)(x-y)=-(\nabla\g)(y-x)$. By the triangle inequality,
\begin{align}
\label{eq:moll_cl_far}
\eqref{eq:ti_cl_far} &\lesssim \left|\sum_{{1\leq i\neq j\leq N}\atop{|x_i-x_j|\leq \ep_3}}\paren*{(v-v_{\ep_2})(x_i)-(v-v_{\ep_2})(x_j)} \cdot(\nabla\g)(x_i-x_j)\right| \nn\\
&\ph + \left|\sum_{{1\leq i\neq j\leq N}\atop{|x_i-x_j|> \ep_3}}\paren*{(v-v_{\ep_2})(x_i)-(v-v_{\ep_2})(x_j)} \cdot(\nabla\g)(x_i-x_j)\right| \nn\\
&\eqqcolon \mathrm{Term}_1 + \mathrm{Term}_2.
\end{align}
By another application of the triangle inequality, followed by using that $\|v_{\ep_2}\|_{LL(\R^2)}\leq \|v\|_{LL(\R^2)}$, together with the elementary bound
\begin{equation}
|(\nabla\g)(x_i-x_j)| \lesssim \frac{1}{|x_i-x_j|},
\end{equation}
we find that
\begin{align}
\mathrm{Term}_1 &\lesssim \sum_{{1\leq i\neq j\leq N}\atop{|x_i-x_j|\leq\ep_3}} \|v\|_{LL(\R^2)}\g(x_i-x_j)\nn\\
&\lesssim \|v\|_{LL(\R^2)}\tl{\g}(2\ep_3)\paren*{\Fr_N(\ux_N,\mu) + N\tl{\g}(2\ep_3) + C_p N^2 \|\mu\|_{L^p(\R^2)} {\ep_3}^{\frac{2(p-1)}{p}}},
\end{align}
where the ultimate inequality follows from applying \cref{cor:count}. For $\mathrm{Term}_2$, we observe the bound
\begin{equation}
|(\nabla\g)(x_i-x_j)| \lesssim \frac{1}{\ep_3},
\end{equation}
which is immediate from the definition of $\nabla\g$ and the separation of $x_i$ and $x_i$, and the bound
\begin{equation}
|(v-v_{\ep_2})(x_i)-(v-v_{\ep_2})(x_j)| \lesssim \|v\|_{LL(\R^2)}\ep_2|\ln\ep_2|,
\end{equation}
which follows from estimate \eqref{eq:v_diff_Linf} of \cref{lem:conv_bnds} and the triangle inequality, in order to obtain that
\begin{equation}
\mathrm{Term}_2 \lesssim \sum_{{1\leq i\neq j\leq N}\atop {|x_i-x_j|>\ep_3}} \frac{\|v\|_{LL(\R^2)}\ep_2 |\ln\ep_2|}{\ep_3} \leq \frac{N^2\|v\|_{LL(\R^2)}\ep_2|\ln\ep_2|}{\ep_3}.
\end{equation}
\item
Next, observe that by the Fubini-Tonelli theorem followed by H\"older's inequality,
\begin{align}
\left|\int_{(\R^2)^2\setminus\D_2} (v-v_{\ep_2})(x)\cdot(\nabla\g)(x-y)d(N\mu)(x)d(N\mu)(y)\right| &= N^2\left|\int_{\R^2} (v-v_{\ep_2})(x)\cdot(\nabla\g\ast\mu)(x)d\mu(x)\right| \nn\\
&\leq N^2\|v-v_{\ep_2}\|_{L^\infty(\R^2)}\|\nabla\g\ast\mu\|_{L^\infty(\R^2)}.
\end{align}
By \cref{lem:conv_bnds} applied to the first factor and \cref{lem:Linf_RP} applied to the second factor in the ultimate line, we find that
\begin{equation}
\label{eq:L_inf_mu_d}
N^2\|v-v_{\ep_2}\|_{L^\infty(\R^2)}\|\nabla\g\ast\mu\|_{L^\infty(\R^2)} \lesssim N^2\ep_2|\ln\ep_2| \|v\|_{LL(\R^2)} \|\mu\|_{L^p(\R^2)}^{\frac{p}{2(p-1)}}.
\end{equation}
\item
Next, observe that by the Fubini-Tonelli theorem
\begin{align}
\left|\int_{(\R^2)^2\setminus\D_2}(v-v_{\ep_2})(x)\cdot\nabla\g(x-y)d(\sum_{i=1}^N\d_{x_i})(x)d(N\mu)(y)\right| &= N\left|\sum_{i=1}^N (v-v_{\ep_2})(x_i) (\nabla\g\ast\mu)(x_i)\right| \nn\\
&\lesssim N^2\ep_2|\ln\ep_2|\|v\|_{LL(\R^2)} \|\mu\|_{L^{p}(\R^2)}^{\frac{p}{2(p-1)}},
\end{align}
again by using \Cref{lem:Linf_RP,lem:conv_bnds}.
\item
Finally, by proceeding similarly as to the previous case,
\begin{align}
\left|\int_{(\R^2)^2\setminus\D_2}(v-v_{\ep_2})(x)\cdot\nabla\g(x-y)d(\sum_{i=1}^N\d_{x_i})(y)d(N\mu)(x)\right| &= N\left|\sum_{i=1}^N(\nabla\g\ast\paren*{(v-v_{\ep_2})\mu})(x_i)\right| \nn\\
&\lesssim N^2\ep_2|\ln\ep_2|\|v\|_{LL(\R^2)} \|\mu\|_{L^{p}(\R^2)}^{\frac{p}{2(p-1)}}.
\end{align}
\end{itemize}
\end{proof}

\subsection{Step 2: Renormalization}
\label{ssec:kprop_renorm}
We now proceed to step 2 of the proof, working with the mollified vector field $v_{\ep_2}$. Crucially, this step relies on the qualitative assumption that $v_{\ep_2}$ is Lipschitz.

\begin{lemma}\label{lem:kprop_renorm}
It holds that
\begin{equation}
\label{eq:kprop_S2_id}
\begin{split}
&\int_{(\R^2)^2\setminus\D_2} (v_{\ep_2}(x)-v_{\ep_2}(y))\cdot (\nabla\g)(x-y)d(\sum_{i=1}^N\d_{x_i}-N\mu)(x)d(\sum_{i=1}^N\d_{x_i}-N\mu)(y)\\
&=\lim_{|\ue_N|\rightarrow 0} \paren*{\int_{\R^2} (\nabla v_{\ep_2})(x): \comm{H_{N,\ue_N}^{\mu,\ux_N}}{H_{N,\ue_N}^{\mu,\ux_N}}_{SE}(x)dx - \sum_{i=1}^N\int_{(\R^2)^2} (v_{\ep_2}(x)-v_{\ep_2}(y))\cdot(\nabla\g)(x-y)\d_{x_i}^{(\eta_i)}(x)\d_{x_i}^{(\eta_i)}(y)}.
\end{split}
\end{equation}
\end{lemma}
\begin{proof}
The reader can check from an application of dominated convergence and recalling the definition of $\d_{0}^{(\eta)}$ in \eqref{eq:delta_smear} that
\begin{equation}
\label{eq:s2_rhs_sub}
\begin{split}
&\int_{(\R^2)^2\setminus\D_2} \paren*{v_{\ep_2}(x)-v_{\ep_2}(y)}\cdot(\nabla\g)(x-y)d(\sum_{i=1}^N\d_{x_i}-N\mu)(x)d(\sum_{i=1}^N\d_{x_i}-N\mu)(y) \\
&=\lim_{|\ue_N|\rightarrow 0}\left(\int_{(\R^2)^2}\paren*{v_{\ep_2}(x)-v_{\ep_2}(y)}\cdot(\nabla\g)(x-y)d(\sum_{i=1}^N\d_{x_i}^{(\eta_i)}-N\mu)(x)d(\sum_{i=1}^N\d_{x_i}^{(\eta_i)}-N\mu)(y) \right.\\
&\hspace{50mm} \left.- \sum_{i=1}^N\int_{(\R^2)^2} \paren*{v_{\ep_2}(x)-v_{\ep_2}(y)}\cdot(\nabla\g)(x-y)\d_{x_i}^{(\eta_i)}(x)\d_{x_i}^{(\eta_i)}(y)\right).
\end{split}
\end{equation}
Since $v_{\ep_2}\in W^{1,\infty}(\R^2)$ and $H_{N,\ue_N}^{\mu,\ux_N}\in\dot{H}^1(\R^2)$ by \cref{lem:fin_en}, we may apply \cref{lem:SE_div} to obtain
\begin{equation}
\begin{split}
&\int_{(\R^2)^2}\paren*{v_{\ep_2}(x)-v_{\ep_2}(y)}\cdot(\nabla\g)(x-y)d(\sum_{i=1}^N\d_{x_i}^{(\eta_i)}-N\mu)(x)d(\sum_{i=1}^N\d_{x_i}^{(\eta_i)}-N\mu)(y) \\
&= \int_{\R^2}(\nabla v_{\ep_2})(x): \comm{H_{N,\ue_N}^{\mu,\ux_N}}{H_{N,\ue_N}^{\mu,\ux_N}}_{SE}(x)dx.
\end{split}
\end{equation}
Substituting this identity into the right-hand side of \eqref{eq:s2_rhs_sub}, we see that the proof of the lemma is complete.
\end{proof}

\subsection{Step 3: Diagonal Terms}
\label{ssec:kprop_diag}
We proceed to step 3 in the proof, in which we analyze how the diagonal contribution (i.e. the second term in the right-hand side) in identity \eqref{eq:kprop_S2_id} varies as the truncation vector $\ue_N$ varies. The following lemma, which caries over from \cite{Serfaty2018mean}, provides an identity which does just this.

\begin{lemma}
\label{lem:kprop_diag}
Let $\ua_N,\ue_N\in (\R_+)^N$, such that $\ua_i\geq \ue_i$ for every $i\in\{1,\ldots,N\}$. Then for each $i\in\{1,\ldots,N\}$, we have the identity
\begin{equation}
\label{eq:diag_zero_id}
\int_{(\R^2)^2}\paren*{v_{\ep_2}(x)-v_{\ep_2}(y)}\cdot(\nabla\g)(x-y)d\d_{x_i}^{(\alpha_i)}(x)d(\d_{x_i}^{(\eta_i)}-\d_{x_i}^{(\alpha_i)})(y) = 0,
\end{equation}
and consequently,
\begin{equation}
\begin{split}
&\int_{(\R^2)^2}\paren*{v_{\ep_2}(x)-v_{\ep_2}(y)}\cdot(\nabla\g)(x-y)\paren*{d\d_{x_i}^{(\eta_i)}(x)d\d_{x_i}^{(\eta_i)}(y)-d\d_{x_i}^{(\alpha_i)}(x)d\d_{x_i}^{(\alpha_i)}(y)} \\
&=\int_{\R^2}(\nabla v_{\ep_2})(x):\comm{\f_{\alpha_i,\eta_i}(\cdot-x_i)}{\f_{\alpha_i,\eta_i}(\cdot-x_i)}_{SE}(x)dx,
\end{split}
\end{equation}
where we recall the definition of $\f_{\al_i,\et_i}$ from \eqref{eq:f_et_al_def}.
\end{lemma}
\begin{proof}
See \cite[Section 4: Steps 2 and 3]{Serfaty2018mean}.
\end{proof}

\subsection{Step 4: Recombination}
\label{ssec:kprop_recomb}
Combining \cref{lem:kprop_renorm} from step 2 and \cref{lem:kprop_diag} from step 3, we have shown that for any $\ua_N\in (\R_+)^N$, such that $\al_i>\et_i$ for every $i\in\{1,\ldots,N\}$, it holds that
\begin{equation}
\label{eq:S4_pref}
\begin{split}
&\int_{(\R^2)^2\setminus\D_2} (v_{\ep_2}(x)-v_{\ep_2}(y))\cdot (\nabla\g)(x-y)d(\sum_{i=1}^N\d_{x_i}-N\mu)(x)d(\sum_{i=1}^N\d_{x_i}-N\mu)(y)\\
&\ph-\paren*{\int_{\R^2} (\nabla v_{\ep_2})(x): \comm{H_{N,\ua_N}^{\mu,\ux_N}}{H_{N,\ua_N}^{\mu,\ux_N}}_{SE}(x)dx - \sum_{i=1}^N\int_{(\R^2)^2} (v_{\ep_2}(x)-v_{\ep_2}(y))\cdot(\nabla\g)(x-y)\d_{x_i}^{(\al_i)}(x)\d_{x_i}^{(\al_i)}(y)} \\
&=\lim_{|\ue_N|\rightarrow 0} \Bigg(\int_{\R^2} (\nabla v_{\ep_2})(x): \paren*{\comm{H_{N,\ue_N}^{\mu,\ux_N}}{H_{N,\ue_N}^{\mu,\ux_N}}_{SE}-\comm{H_{N,\ua_N}^{\mu,\ux_N}}{H_{N,\ua_N}^{\mu,\ux_N}}_{SE}}(x)dx \\
&\ph \hspace{25mm} - \sum_{i=1}^N\int_{\R^2}(\nabla v_{\ep_2})(x):\comm{\f_{\alpha_i,\eta_i}}{\f_{\alpha_i,\eta_i}}_{SE}(x-x_i)dx\Bigg)
\end{split}
\end{equation}
In step 4, we analyze how the first term in the third line varies as $|\ua_N|, |\ue_N| \rightarrow 0$.

\begin{lemma}
\label{lem:kprop_recomb}
Define the parameters
\begin{equation}
\label{eq:r_def}
r_{i,\ep_1} \coloneqq \min\{\frac{1}{4}\min_{{1\leq j\leq N}\atop{i\neq j}}|x_i-x_j|, \ep_1\}, \qquad \ur_{N,\ep_1} \coloneqq (r_{1,\ep_1},\ldots,r_{N,\ep_1}).
\end{equation}
If $\ua_N,\eta_N\in(\R_+)^N$ are such that $\eta_i<\alpha_i\leq r_{i,\ep_1}$, for every $i\in\{1,\ldots,N\}$, then
\begin{equation}
\label{eq:recomb}
\begin{split}
&\int_{\R^2} (\nabla v_{\ep_2})(x) : \paren*{\comm{H_{N,\ul{\eta}_N}^{\mu,\ux_N}}{H_{N,\ul{\eta}_N}^{\mu,\ux_N}}_{SE} - \comm{H_{N,\ua_N}^{\mu,\ux_N}}{H_{N,\ua_N}^{\mu,\ux_N}}_{SE}}(x)dx \\
&= \sum_{i=1}^N \int_{\R^2} (\nabla v_{\ep_2})(x): \comm{\f_{\al_i,\et_i}}{\f_{\al_i,\et_i}}_{SE}(x-x_i)dx  + \mathrm{Error}_{\ep_2,\ua_N,\ue_N},
\end{split}
\end{equation}
where there exist constants $C_p,C_\infty>0$ such that
\begin{equation}
\label{eq:recomb_err_bnd}
\begin{split}
\left|\mathrm{Error}_{\ep_2,\ua_N,\ue_N}\right| &\lesssim \|v\|_{LL(\R^2)}|\ln\ep_3| \paren*{\Fr_N(\ux_N,\mu) + N|\ln\ep_3| + C_p N^2 \|\mu\|_{L^p(\R^2)} {\ep_3}^{\frac{2(p-1)}{p}}} \\
&\ph + N\sum_{i=1}^N\al_i \paren*{\frac{\|v\|_{L^\infty(\R^2)}}{{\ep_3}^2} +\frac{\|v\|_{{LL}(\R^2)} |\ln \al_i|}{\ep_3}}\\
&\ph +N\|v\|_{{LL}(\R^2)} \|\mu\|_{L^p(\R^2)}^{\frac{p}{2(p-1)}}\sum_{i=1}^N \al_i|\ln\al_i| \\
&\ph + N\|v\|_{L^\infty(\R^2)}\sum_{i=1}^N \paren*{C_p|\al_i|^{1-\frac{2}{p}}\|\mu\|_{L^p(\R^2)} + C_\infty|\al_i||\ln\al_i|\|\mu\|_{L^\infty(\R^2)} 1_{\geq \infty}(p)}.
\end{split}
\end{equation}
\end{lemma}
\begin{proof}
Since for any $\beta>0$, $\g\ast\delta_{0}^{(\beta)} = \g_\beta$, it follows that
\begin{align}
H_{N,\ua_N}^{\mu,\ux_N} &= (\g\ast N\mu) - \sum_{i=1}^N \g_{\alpha_i}(\cdot-x_i) \nn\\
&= (\g\ast N\mu) - \sum_{i=1}^N \g_{\eta_i}(\cdot-x_i) + \sum_{i=1}^N \f_{\alpha_i,\eta_i}(\cdot-x_i) \nn\\
&=H_{N,\ue_N}^{\mu,\ux_N} +\sum_{i=1}^N \f_{\alpha_i,\eta_i}(\cdot-x_i).
\end{align}
Observe that $\g_{\eta_i}(\cdot-x_i) = \g_{\alpha_i}(\cdot-x_i)$ outside the ball $B(x_i,\alpha_i)$ (recall that $\al_i>\et_i$ by assumption). Also observe that the closed balls $\ol{B}(x_i,\alpha_i)$ are pairwise disjoint. To see this property, observe that
\begin{equation}
\al_i \leq r_{i,\ep_1} \leq \frac{|x_i-x_j|}{4} \quad \text{and} \quad \al_j \leq r_{j,\ep_1} \leq \frac{|x_i-x_j|}{4}, \qquad \forall i\neq j,
\end{equation}
by definition of $r_{i,\ep_1},r_{j,\ep_1}$ and requirement that $\al_i\leq r_{i,\ep_1}$ and $\al_j\leq r_{j,\ep_1}$. Thus, we find that
\begin{equation}
H_{N,\ua_N}^{\mu,\ux_N}(x) = H_{N,\ue_N}^{\mu,\ux_N}(x), \qquad \forall x\in \R^2\setminus\bigcup_{i=1}^N B(x_i,\alpha_i).
\end{equation}
So by the bilinearity of the stress-energy tensor, we have the point-wise a.e. identity
\begin{equation}
\label{eq:se_diff_id}
\begin{split}
\comm{H_{N,\ua_N}^{\mu,\ux_N}}{H_{N,\ua_N}^{\mu,\ux_N}}_{SE} - \comm{H_{N,\ue_N}^{\mu,\ux_N}}{H_{N,\ue_N}^{\mu,\ux_N}}_{SE} &= \sum_{i=1}^N \comm{\f_{\alpha_i,\eta_i}}{\f_{\alpha_i,\eta_i}}_{SE}1_{B(x_i,\alpha_i)} \\
&\ph + 2\sum_{i=1}^N\comm{\f_{\alpha_i,\eta_i}}{H_{N,\ue_N}^{\mu,\ux_N}}_{SE}1_{B(x_i,\alpha_i)},
\end{split}
\end{equation}
where we use that the terms $\comm{\f_{\alpha_i,\eta_i}}{\f_{\alpha_j,\eta_j}}_{SE}$, for $i\neq j$, vanish due to $\ol{B}(x_i,\alpha_i)\cap \ol{B}(x_j,\alpha_j)=\emptyset$ and similarly, $\comm{\f_{\alpha_j,\eta_j}}{H_{N,\ue_N}^{\mu,\ux_N}}_{SE}1_{B(x_i,\alpha_i)}\equiv 0$. Thus, using identity \eqref{eq:se_diff_id}, we see that
\begin{align}
&\int_{\R^2}(\nabla v_{\ep_2})(x) : \paren*{\comm{H_{N,\ua_N}^{\mu,\ux_N}}{H_{N,\ua_N}^{\mu,\ux_N}}_{SE} - \comm{H_{N,\ue_N}^{\mu,\ux_N}}{H_{N,\ue_N}^{\mu,\ux_N}}_{SE}}(x)dx \nn\\
&=\sum_{i=1}^N \int_{B(x_i,\alpha_i)} (\nabla v_{\ep_2})(x): \paren*{\comm{\f_{\alpha_i,\eta_i}}{\f_{\alpha_i,\eta_i}}_{SE}+2\comm{\f_{\alpha_i,\eta_i}}{H_{N,\ue_N}^{\mu,\ux_N}}_{SE}}(x)dx \nn\\
&=\sum_{i=1}^N\int_{\R^2}(\nabla v_{\ep_2})(x) :\paren*{\comm{\f_{\alpha_i,\eta_i}}{\f_{\alpha_i,\eta_i}}_{SE}+2\comm{\f_{\alpha_i,\eta_i}}{H_{N,\ue_N}^{\mu,\ux_N}}_{SE}}(x)dx,  \label{eq:recomb_nte}
\end{align}
where the ultimate line follows from the fact that $\nabla\f_{\alpha_i,\eta_i}$ vanishes outside $\ol{B}(x_i,\alpha_i)$. Comparing this last expression to the right-hand side of identity \eqref{eq:recomb}, we see that we only need to estimate the modulus of
\begin{equation}
\label{eq:recomb_split_pre}
\begin{split}
&2\sum_{i=1}^N \int_{\R^2}(\nabla v_{\ep_2})(x) : \comm{\f_{\alpha_i,\eta_i}}{H_{N,\ue_N}^{\mu,\ux_N}}_{SE}(x)dx\\
&=2\sum_{i=1}^N \int_{(\R^2)^2}\paren*{v_{\ep_2}(x)-v_{\ep_2}(y)}\cdot (\nabla\g)(x-y)d(\sum_{j=1}^N\d_{x_j}^{(\alpha_j)}-N\mu)(x)d(\d_{x_i}^{(\eta_i)}-\d_{x_i}^{(\alpha_i)})(y),
\end{split}
\end{equation}
where the right-hand side follows from an application of \cref{lem:SE_div}. Note that by identity \eqref{eq:diag_zero_id},
\begin{equation}
\label{eq:recomb_split}
\eqref{eq:recomb_split_pre} = 2\sum_{i=1}^N \int_{(\R^2)^2}\paren*{v_{\ep_2}(x)-v_{\ep_2}(y)}\cdot (\nabla\g)(x-y)d(\sum_{{1\leq j\leq N} \atop {j\neq i}}\d_{x_j}^{(\alpha_j)}-N\mu)(x)d(\d_{x_i}^{(\eta_i)}-\d_{x_i}^{(\alpha_i)})(y).
\end{equation}
Thus, we can split the right-hand side of \eqref{eq:recomb_split} into a sum of two terms defined below:
\begin{align}
\mathrm{Term}_1 &\coloneqq 2\sum_{i=1}^N \int_{(\R^2)^2}v_{\ep_2}(x)\cdot(\nabla\g)(x-y)d(\sum_{{1\leq j\leq N}\atop {j\neq i}}\d_{x_j}^{(\alpha_j)}-N\mu)(x)d(\d_{x_i}^{(\eta_i)}-\d_{x_i}^{(\alpha_i)})(y),  \label{eq:recomb_T1_def}\\
\mathrm{Term}_2 &\coloneqq -2\sum_{i=1}^N\int_{(\R^2)^2} v_{\ep_2}(y)\cdot(\nabla\g)(x-y)d(\sum_{{1\leq j\leq N}\atop {j\neq i}}\d_{x_j}^{(\alpha_j)}-N\mu)(x)d(\d_{x_i}^{(\eta_i)}-\d_{x_i}^{(\alpha_i)})(y). \label{eq:recomb_T2_def}
\end{align}

\begin{description}[leftmargin=*]
\item[Estimate for $\mathrm{Term}_1$]
We observe that by using the Fubini-Tonelli theorem to first integrate out the $y$-variable and then applying identity \eqref{eq:f_conv_smear},
\begin{equation}
\mathrm{Term}_1 = 2\sum_{i=1}^N\int_{\R^2}v_{\ep_2}(x)\cdot (\nabla\f_{\alpha_i,\eta_i})(x-x_i)d(\sum_{{1\leq j\leq N}\atop {j\neq i}}\d_{x_j}^{(\alpha_j)}-N\mu)(x).
\end{equation}
Since $\supp(\nabla\f_{\alpha_i,\eta_i}(\cdot-x_i))\subset \ol{B}(x_i,\alpha_i)$ and $\supp(\d_{x_j}^{(\alpha_j)})\subset \ol{B}(x_j,\alpha_j)$, which is disjoint from $\ol{B}(x_i,\alpha_i)$, for $j\neq i$, we observe the cancellation
\begin{equation}
\mathrm{Term}_1 = -2N\sum_{i=1}^N\int_{\R^2}v_{\ep_2}(x)\cdot(\nabla\f_{\alpha_i,\eta_i})(x)d\mu(x).
\end{equation}
By triangle and H\"older's inequalities
\begin{equation}
\label{eq:T1_Hold}
\left|-2N\sum_{i=1}^N\int_{\R^2}v_{\ep_2}(x)\cdot(\nabla\f_{\alpha_i,\eta_i})(x)d\mu(x)\right| \lesssim N\sum_{i=1}^N \|v_{\ep_2}\|_{L^{p_1}(\R^2)} \|\nabla\f_{\alpha_i,\eta_i}\|_{L^{p_2}(\R^2)} \|\mu\|_{L^{p_3}(\R^2)},
\end{equation}
where $1\leq p_1,p_2,p_3\leq \infty$ are H\"older conjugate. Requiring $1\leq p_2< 2$, $p_1=\infty$, and $p_3=p$ and applying \cref{lem:f_et_al_bnds} to $\|\nabla\f_{\alpha_i,\eta_i}\|_{L^{p_2}(\R^2)}$, we find that
\begin{equation}
\eqref{eq:T1_Hold} \lesssim_{p} N\sum_{i=1}^N\alpha_i^{\frac{p-2}{p}} \|v_{\ep_2}\|_{L^\infty(\R^2)} \|\mu\|_{L^{p}(\R^2)}.
\end{equation}
Using that $\|v_{\ep_2}\|_{L^\infty(\R^2)}\leq \|v\|_{L^\infty(\R^2)}$ by \cref{lem:conv_bnds}, we conclude that
\begin{equation}
\label{eq:recomb_T1_fin}
|\mathrm{Term}_1| \lesssim_p N\|v\|_{L^\infty(\R^2)}\|\mu\|_{L^p(\R^2)}\sum_{i=1}^N\al_i^{\frac{p-2}{p}},
\end{equation}
which completes the analysis for $\mathrm{Term}_1$.

\item[Estimate for $\mathrm{Term}_2$]
We proceed to analyze $\mathrm{Term}_2$ from \eqref{eq:recomb_T2_def}. Using the Fubini-Tonelli theorem to first integrate out the $x$-variable together with identity \eqref{eq:g_conv_smear}, we can write
\begin{equation}
\begin{split}
\mathrm{Term}_2 &= 2\sum_{1\leq i\neq j\leq N}\int_{\R^2} v_{\ep_2}(y)\cdot(\nabla\g_{\alpha_j})(y-x_j)d(\d_{x_i}^{(\eta_i)}-\d_{x_i}^{(\alpha_i)})(y) \\
&\ph - 2\sum_{1\leq i\neq j\leq N}\int_{\R^2}v_{\ep_2}(y)\cdot(\nabla\g\ast\mu)(y)d(\d_{x_i}^{(\eta_i)}-\d_{x_i}^{(\alpha_i)})(y).
\end{split}
\end{equation}
Adding and subtracting $v_{\ep_2}(x_i)$ and $v_{\ep_2}(x_j)$, respectively, we can write
\begin{equation}
\begin{split}
&2\sum_{1\leq i\neq j\leq N}\int_{\R^2} v_{\ep_2}(y)\cdot(\nabla\g_{\alpha_j})(y-x_j)d(\d_{x_i}^{(\eta_i)}-\d_{x_i}^{(\alpha_i)})(y) \\
&= \underbrace{2\sum_{1\leq i\neq j\leq N}\int_{\R^2} \paren*{v_{\ep_2}(y)-v_{\ep_2}(x_j)}\cdot(\nabla\g_{\alpha_j})(y-x_j)d(\d_{x_i}^{(\eta_i)}-\d_{x_i}^{(\alpha_i)})(y)}_{\eqqcolon \mathrm{Term}_{2,1}} \\
&\ph + \underbrace{2\sum_{1\leq i\neq j\leq N}\int_{\R^2}v_{\ep_2}(x_j)\cdot(\nabla\g_{\alpha_j})(y-x_j)d(\d_{x_i}^{(\eta_i)}-\d_{x_i}^{(\alpha_i)})(y)}_{\eqqcolon \mathrm{Term}_{2,2}}
\end{split}
\end{equation}
and
\begin{equation}
\begin{split}
&2\sum_{1\leq i\neq j\leq N}\int_{\R^2}v_{\ep_2}(y)\cdot(\nabla\g\ast\mu)(y)d(\d_{x_i}^{(\eta_i)}-\d_{x_i}^{(\alpha_i)})(y)\\
&= \underbrace{2\sum_{1\leq i\neq j\leq N}\int_{\R^2}\paren*{v_{\ep_2}(y)-v_{\ep_2}(x_i)}\cdot(\nabla\g\ast\mu)(y)d(\d_{x_i}^{(\eta_i)}-\d_{x_i}^{(\alpha_i)})(y)}_{\eqqcolon\mathrm{Term}_{2,3}}\\
&\ph +\underbrace{2\sum_{1\leq i\neq j\leq N}\int_{\R^2}v_{\ep_2}(x_i)\cdot(\nabla\g\ast\mu)(y)d(\d_{x_i}^{(\eta_i)}-\d_{x_i}^{(\alpha_i)})(y)}_{\eqqcolon\mathrm{Term}_{2,4}},
\end{split}
\end{equation}
so that
\begin{equation}
\mathrm{Term}_2 = \mathrm{Term}_{2,1} + \mathrm{Term}_{2,2} - \mathrm{Term}_{2,3} - \mathrm{Term}_{2,4}.
\end{equation}
We proceed to estimate $\mathrm{Term}_{2,1},\ldots,\mathrm{Term}_{2,4}$ separately.
\begin{description}[leftmargin=*]
\item[Estimate for $\mathrm{Term}_{2,1}$]
As in the proof of \cref{lem:kprop_error}, we decompose the sum $\sum_{1\leq i\neq j\leq N}$ in $\mathrm{Term}_{2,1}$ in terms of ``close'' and ``far'' vortex pairs $(i,j)$ with distance threshold $\ep_3$:
\begin{align}
\mathrm{Term}_{2,1} &= 2\sum_{{1\leq i\neq j\leq N}\atop {|x_i-x_j|\leq\ep_3}}\int_{\R^2} \paren*{v_{\ep_2}(y)-v_{\ep_2}(x_j)}\cdot(\nabla\g_{\alpha_j})(y-x_j)d(\d_{x_i}^{(\eta_i)}-\d_{x_i}^{(\alpha_i)})(y)\nn\\
&\ph + 2\sum_{{1\leq i\neq j\leq N}\atop{|x_i-x_j|>\ep_3}}\int_{\R^2} \paren*{v_{\ep_2}(y)-v_{\ep_2}(x_j)}\cdot(\nabla\g_{\alpha_j})(y-x_j)d(\d_{x_i}^{(\eta_i)}-\d_{x_i}^{(\alpha_i)})(y) \nn\\
&\eqqcolon \mathrm{Term}_{2,1,1} + \mathrm{Term}_{2,1,2}.
\end{align}

For $\mathrm{Term}_{2,1,1}$, we use that $\|v_{\ep_2}\|_{LL(\R^2)}\leq \|v\|_{LL(\R^2)}$ by \cref{lem:conv_bnds} in order to obtain that
\begin{align}
|\mathrm{Term}_{2,1,1}| &\leq \|v\|_{LL(\R^2)}\sum_{{1\leq i\neq j\leq N}\atop{|x_i-x_j|\leq\ep_3}}\int_{\R^2} |y-x_j||\ln|y-x_j||(\nabla\g_{\alpha_j})(y-x_j)|d(\d_{x_i}^{(\et_i)}+\d_{x_i}^{(\al_i)})(y) \nn\\
&\lesssim \|v\|_{LL(\R^2)}\sum_{{1\leq i\neq j\leq N}\atop {|x_i-x_j|\leq\ep_3}}\int_{\R^2}\g(y-x_j)d(\d_{x_i}^{(\eta_i)}+\d_{x_i}^{(\al_i)})(y),
\end{align}
where the ultimate inequality follows from using the bound
\begin{equation}
|(\nabla\g_{\alpha_j})(y-x_j)| \leq \frac{1}{2\pi|y-x_j|}
\end{equation}
and $\ln|y-x_j|<0$ for $y\in\supp(\d_{x_j}^{(\eta_i)}+\d_{x_j}^{(\al_i)})$. Since for $1\leq i\neq j\leq N$, we have that $|x_i-x_j|\geq 4r_{i,\ep_1}$, by definition of $r_{i,\ep_1}$, and $r_{i,\ep_1}\geq \alpha_i>\eta_i$, by assumption, the reverse triangle inequality and the decreasing property of $\g$ imply that
\begin{equation}
\g(y-x_j) \leq \tl{\g}(|x_i-x_j| - \frac{1}{4}|x_i-x_j|) \leq \tl{\g}(\frac{3}{4}) + \g(x_i-x_j), \qquad \forall y\in \p B(x_i,\alpha_i)\cup\p B(x_i,\eta_i).
\end{equation}
Hence,
\begin{equation}
\|v\|_{LL(\R^2)}\sum_{{1\leq i\neq j\leq N}\atop {|x_i-x_j|\leq\ep_3}}\int_{\R^2}\g(y-x_j)d(\d_{x_i}^{(\eta_i)}+\d_{x_i}^{(\al_i)})(y) \lesssim \|v\|_{LL(\R^2)}\sum_{{1\leq i\neq j\leq N}\atop{|x_i-x_j|\leq\ep_3}} \paren*{\tl{\g}(\frac{3}{4}) + \g(x_i-x_j)}.
\end{equation}
Thus by \cref{cor:count},
\begin{equation}
\label{eq:T_211_fin}
|\mathrm{Term}_{2,1,1}|\lesssim \|v\|_{LL(\R^2)} |\ln\ep_3|\paren*{\Fr_N(\ux_N,\mu) + N|\ln\ep_3| + C_p N^2 \|\mu\|_{L^p(\R^2)} {\ep_3}^{\frac{2(p-1)}{p}}}.
\end{equation}

For $\mathrm{Term}_{2,1,2}$, we make a change of variable to write
\begin{align}
&\int_{\R^2}(v_{\ep_2}(y)-v_{\ep_2}(x_j)) \cdot(\nabla\g_{\alpha_j})(y-x_j)d(\d_{x_i}^{(\eta_i)}-\d_{x_i}^{(\al_i)})(y) \nn\\
&=\int_{\R^2} \paren*{v_{\ep_2}(\eta_i y+x_i-x_j)\cdot(\nabla\g_{\alpha_j})(\eta_i y +x_i-x_j) - v_{\ep_2}(\al_i y+x_i-x_j) \cdot(\nabla\g_{\alpha_j})(\al_i y +x_i-x_j)} d\d_0^{(1)}(y). \label{eq:LL_MVT_bnd}
\end{align}
Since $\ep_3>\ep_2\geq 2\ep_1$ by assumption and since $\eta_i\leq \alpha_i\leq \ep_1$, for every $i\in\{1,\ldots,N\}$, also by assumption, it follows from the log-Lipschitz property of $v_{\ep_2}$ and the (reverse) triangle inequality that
\begin{align}
|\eqref{eq:LL_MVT_bnd}| &\lesssim \int_{\R^2} \frac{\|v_{\ep_2}\|_{{LL}(\R^2)}|\et_i-\al_i| |\ln|\et_i-\al_i||}{|\eta_i y+ x_i-x_j|}d\d_0^{(1)}(y) + \int_{\R^2}\frac{|v_{\ep_2}(\al_i y + x_i-x_j)| \cdot |\et_i-\al_i|}{(|x_i-x_j|-\al_i)(|x_i-x_j|-\et_i)}d\d_0^{(1)}(y) \nn\\
&\lesssim \frac{\|v\|_{{LL}(\R^2)}|\eta_i-\al_i||\ln|\eta_i-\al_i||}{\ep_3} + \frac{\|v\|_{L^\infty(\R^2)}|\et_i-\al_i|}{{\ep_3}^2}.
\end{align}
Hence,
\begin{align}
|\mathrm{Term}_{2,1,2}| &\lesssim \sum_{{1\leq i\neq j\leq N}\atop{|x_i-x_j|>\ep_3}}\Bigg(\frac{\|v\|_{LL(\R^2)}|\eta_i-\al_i| |\ln|\eta_i-\al_i||}{\ep_3} + \frac{\|v\|_{L^\infty(\R^2)}|\et_i-\al_i|}{{\ep_3}^2}\Bigg) \nn\\
&\lesssim N \sum_{i=1}^N|\et_i-\al_i|\paren*{\frac{\|v\|_{L^\infty(\R^2)}}{\ep_3^2} + \frac{\|v\|_{LL(\R^2)}|\ln|\et_i-\al_i||}{\ep_3}}.
\end{align}
Since $x\mapsto |x|\ln|x|$ is increasing in the region $0<|x|<e^{-1}$, we conclude that
\begin{equation}
\label{eq:T_212_fin}
|\mathrm{Term}_{2,1,2}| \lesssim N\sum_{i=1}^N\al_i \paren*{\frac{\|v\|_{L^\infty(\R^2)}}{\ep_3^2} + \frac{\|v\|_{LL(\R^2)} |\ln \al_i|}{\ep_3}}.
\end{equation}

Combining estimates \eqref{eq:T_211_fin} and \eqref{eq:T_212_fin} for $\mathrm{Term}_{2,1,1}$ and $\mathrm{Term}_{2,1,2}$, respectively, we conclude that
\begin{equation}
\label{eq:T_21_fin}
\begin{split}
|\mathrm{Term}_{2,1}| &\lesssim \|v\|_{LL(\R^2)} |\ln\ep_3|\paren*{\Fr_N(\ux_N,\mu) + N|\ln\ep_3| + C_p N^2 \|\mu\|_{L^p(\R^2)} {\ep_3}^{\frac{2(p-1)}{p}}} \\
&\ph+N\sum_{i=1}^N \al_i\paren*{\frac{\|v\|_{L^\infty(\R^2)}}{\ep_3^2} + \frac{\|v\|_{LL(\R^2)} |\ln \al_i|}{\ep_3}}.
\end{split}
\end{equation}

\item[Estimate for $\mathrm{Term}_{2,2}$]

We claim that $\mathrm{Term}_{2,2}=0$. Indeed, since by identity \eqref{eq:f_et_al_id}, we have that
\begin{equation}
\d_{x_i}^{(\eta_i)}-\d_{x_i}^{(\alpha_i)} = (-\D\f_{\alpha_i,\eta_i})(\cdot-x_i),
\end{equation}
we can integrate by parts once, observing that
\begin{equation}
(-\D\g_{\alpha_j})(y-x_j) = \d_{x_j}^{(\alpha_j)}
\end{equation}
by identity \eqref{eq:delta_smear} and a change of variable, in order to obtain
\begin{equation}
\mathrm{Term}_{2,2} = 2\sum_{1\leq i\neq j\leq N}\int_{\R^2}v_{\ep_2}(x_i)\cdot (\nabla\f_{\alpha_i,\eta_i})(y-x_i)d\d_{x_j}^{(\alpha_j)}(y).
\end{equation}
Note that $\supp(\nabla\f_{\al_i,\et_i}(\cdot-x_i))\subset \ol{B}(x_i,\al_i)$ by identity \eqref{eq:f_grad_et_al_id} and $\supp(\d_{x_j}^{(\al_j)})\subset \ol{B}(x_j,\al_j)$, while the balls $B(x_i,\al_i), B(x_j,\al_j)$ are disjoint, for $i\neq j$. Thus, we conclude that each of the summands in $\mathrm{Term}_{2,2}$ vanishes, implying
\begin{equation}
\label{eq:T_22_fin}
\mathrm{Term}_{2,2} = 0.
\end{equation}

\item[Estimate for $\mathrm{Term}_{2,3}$]
Using that $\|v_{\ep_2}\|_{{LL}(\R^2)}\leq \|v\|_{{LL}(\R^2)}$ by \cref{lem:conv_bnds}, we find that
\begin{equation}
|\mathrm{Term}_{2,3}| \lesssim \|v\|_{{LL}(\R^2)}\sum_{1\leq i\neq j\leq N} \int_{\R^2}|x-x_i| |\ln|x-x_i|| |(\nabla\g\ast\mu)(x)| d(\d_{x_i}^{(\eta_i)}+\d_{x_i}^{(\al_i)})(x).
\end{equation}
By \cref{lem:Linf_RP} (note that $\|\mu\|_{L^1(\R^2)}=1$),
\begin{equation}
\|\nabla\g\ast\mu\|_{L^\infty(\R^2)} \lesssim_p \|\mu\|_{L^p(\R^2)}^{\frac{p}{2(p-1)}}.
\end{equation}
Since $x\mapsto |x||\ln|x||$ is increasing for $0<|x|<e^{-1}$ and $\eta_i<\al_i$, it follows from the supports of $\d_{x_i}^{(\al_i)}, \d_{x_i}^{(\eta_i)}$ that
\begin{align}
|\mathrm{Term}_{2,3}| &\lesssim_p \|v\|_{{LL}(\R^2)}\|\mu\|_{L^p(\R^2)}^{\frac{p}{2(p-1)}}\sum_{1\leq i\neq j\leq N}\al_i|\ln\al_i| \nn\\
&\leq N\|v\|_{{LL}(\R^2)} \|\mu\|_{L^p(\R^2)}^{\frac{p}{2(p-1)}}\sum_{i=1}^N \al_i |\ln\al_i|, \label{eq:T_23_fin}
\end{align}
which completes the analysis for $\mathrm{Term}_{2,3}$.

\item[Estimate for $\mathrm{Term}_{2,4}$]
We first observe that by a change of variable,
\begin{equation}
\label{eq:mod_cont_app}
\begin{split}
&\int_{\R^2}v_{\ep_2}(x_i)\cdot(\nabla\g\ast\mu)(x) d(\d_{x_i}^{(\et_i)}-\d_{x_i}^{(\al_i)})(x) \\
&=\int_{\R^2}v_{\ep_2}(x_i)\cdot\paren*{(\nabla\g\ast\mu)(x_i+\et_iy)-(\nabla\g\ast\mu)(x_i+\al_iy)}d\d_0^{(1)}(y),
\end{split}
\end{equation}
for every $i\in\{1,\ldots,N\}$. By \cref{lem:pot_bnds}, we have the modulus of continuity estimates
\begin{align}
|(\nabla\g\ast\mu)(x_i+\et_iy)-(\nabla\g\ast\mu)(x_i+\al_iy)| &\lesssim \|(\nabla\g\ast\mu)\|_{\dot{C}^{1-\frac{2}{p}}(\R^2) }|\et_i y - \al_i y|^{1-\frac{2}{p}} \nn\\
&\lesssim_p \|\mu\|_{L^p(\R^2)} |\et_i-\al_i|^{1-\frac{2}{p}}, \qquad \forall |y|=1, \ 2<p<\infty
\end{align}
and
\begin{align}
|(\nabla\g\ast\mu)(x_i+\et_iy)-(\nabla\g\ast\mu)(x_i+\al_iy)| &\lesssim \|(\nabla\g\ast\mu)\|_{{LL}(\R^2)} |\et_i y-\al_i y||\ln|\et_iy-\al_iy|| \nn\\
&\lesssim \|\mu\|_{L^\infty(\R^2)} |\et_i-\al_i||\ln|\et_i-\al_i||, \qquad \forall |y|=1, \ p=\infty.
\end{align}
Applying these estimates and using that $x\mapsto |x|^{1-\frac{2}{p}}$ and $x\mapsto |x||\ln|x||$ are increasing for $0<|x|<e^{-1}$, we find that
\begin{align}
|\eqref{eq:mod_cont_app}| &\leq \begin{cases} C_p\|v\|_{L^\infty(\R^2)}\|\mu\|_{L^p(\R^2)}|\al_i|^{1-\frac{2}{p}}, & {2<p<\infty} \\
C_\infty\|v\|_{L^\infty(\R^2)}\|\mu\|_{L^\infty(\R^2)}|\al_i||\ln|\al_i||, & {p=\infty}
\end{cases},
\end{align}
where $C_p, C_\infty>0$ are absolute constants. Hence,
\begin{equation}
\label{eq:T_24_fin}
|\mathrm{Term}_{2,4}| \leq \begin{cases} NC_p\|v\|_{L^\infty(\R^2)}\|\mu\|_{L^p(\R^2)}\sum_{i=1}^N|\al_i|^{1-\frac{2}{p}}, & {2<p<\infty} \\
N C_\infty\|v\|_{L^\infty(\R^2)}\|\mu\|_{L^\infty(\R^2)}\sum_{i=1}^N|\al_i||\ln|\al_i||, & {p=\infty}
\end{cases},
\end{equation}
which completes the analysis for $\mathrm{Term}_{2,4}$.
\end{description}
Combining estimates \eqref{eq:T_21_fin}, \eqref{eq:T_22_fin}, \eqref{eq:T_23_fin}, and \eqref{eq:T_24_fin} for $\mathrm{Term}_{2,1}$, $\mathrm{Term}_{2,2}$, $\mathrm{Term}_{2,3}$, and $\mathrm{Term}_{2,4}$, respectively, we see that
\begin{equation}
\label{eq:recomb_T2_fin}
\begin{split}
|\mathrm{Term}_2| &\lesssim  \|v\|_{LL(\R^2)} |\ln\ep_3|\paren*{\Fr_N(\ux_N,\mu) + N|\ln\ep_3| + C_p N^2 \|\mu\|_{L^p(\R^2)} {\ep_3}^{\frac{2(p-1)}{p}}}\\
&\ph +N\sum_{i=1}^N \al_i \paren*{\frac{\|v\|_{L^\infty(\R^2)}}{{\ep_3}^2} +\frac{\|v\|_{{LL}(\R^2)} |\ln \al_i|}{\ep_3}}+N\|v\|_{{LL}(\R^2)} \|\mu\|_{L^p(\R^2)}^{\frac{p}{2(p-1)}}\sum_{i=1}^N \al_i|\ln\al_i| \\
&\ph + N\|v\|_{L^\infty(\R^2)}\sum_{i=1}^N \paren*{C_p\|\mu\|_{L^p(\R^2)}|\al_i|^{1-\frac{2}{p}}1_{2<\cdot<\infty}(p) + C_\infty\|\mu\|_{L^\infty(\R^2)}\al_i||\ln|\al_i||1_{\geq \infty}(p)}.
\end{split}
\end{equation}
This completes the analysis for $\mathrm{Term}_2$.
\end{description}

With estimate \eqref{eq:recomb_T1_fin} for $\mathrm{Term}_1$, the definition of which we recall from \eqref{eq:recomb_T1_def}, and estimate \eqref{eq:recomb_T2_fin} for $\mathrm{Term}_2$, the definition of which we recall from \eqref{eq:recomb_T2_def}, respectively, together with our starting identity \eqref{eq:recomb_nte}, we conclude that
\begin{equation}
\begin{split}
&\left|\int_{\R^2} (\nabla v_{\ep_2})(x) : \paren*{\comm{H_{N,\ua_N}^{\mu,\ux_N}}{H_{N,\ua_N}^{\mu,\ux_N}}_{SE} - \comm{H_{N,\ue_N}^{\mu,\ux_N}}{H_{N,\ue_N}^{\mu,\ux_N}}_{SE} - \sum_{i=1}^N\comm{\f_{\al_i,\et_i}}{\f_{\al_i,\et_i}}_{SE}(\cdot-x_i)}(x)dx\right| \\
&\lesssim   \|v\|_{LL(\R^2)}|\ln\ep_3| \paren*{\Fr_N(\ux_N,\mu) + N|\ln\ep_3| + C_p N^2 \|\mu\|_{L^p(\R^2)} {\ep_3}^{\frac{2(p-1)}{p}}} \\
&\ph + N\sum_{i=1}^N \al_i\paren*{\frac{\|v\|_{L^\infty(\R^2)}}{{\ep_3}^2} +\frac{\|v\|_{{LL}(\R^2)} |\ln \al_i|}{\ep_3}} +N\|v\|_{{LL}(\R^2)} \|\mu\|_{L^p(\R^2)}^{\frac{p}{2(p-1)}}\sum_{i=1}^N \al_i|\ln\al_i| \\
&\ph + N\|v\|_{L^\infty(\R^2)}\sum_{i=1}^N \paren*{C_p\|\mu\|_{L^p(\R^2)}|\al_i|^{1-\frac{2}{p}} + \|\mu\|_{L^\infty(\R^2)}|\al_i||\ln \al_i|1_{\geq \infty}(p)}.
\end{split}
\end{equation}
Comparing the preceding inequality with the statement of \cref{lem:kprop_recomb}, we see that the proof of the lemma is complete.
\end{proof}

\subsection{Step 5: Conclusion}
\label{ssec:kprop_conc}
Returning to the identity \eqref{eq:S4_pref} and applying \cref{lem:kprop_recomb}, we have shown that for any $\ua_N\in (\R_+)^N$ with $\alpha_i\leq r_{i,\ep_1}$ for each $i\in\{1,\ldots,N\}$, it holds that
\begin{equation}
\label{eq:conc_RHS}
\begin{split}
&\int_{(\R^2)^2\setminus\D_2} \paren*{v_{\ep_2}(x)-v_{\ep_2}(y)}\cdot(\nabla\g)(x-y)d(\sum_{i=1}^N\d_{x_i}-N\mu)(x)d(\sum_{i=1}^N\d_{x_i}-N\mu)(y)\\
&=\mathrm{Error}_{\ep_2,\ua_N}+\int_{\R^2} (\nabla v_{\ep_2})(x) : \comm{H_{N,\ua_N}^{\mu,\ux_N}}{H_{N,\ua_N}^{\mu,\ux_N}}_{SE}(x)dx \\
&\ph -\sum_{i=1}^N\int_{(\R^2)^2} \paren*{v_{\ep_2}(x)-v_{\ep_2}(y)}\cdot(\nabla\g)(x-y)d\d_{x_i}^{(\alpha_i)}(x)d\d_{x_i}^{(\alpha_i)}(y),
\end{split}
\end{equation}
where $\mathrm{Error}_{\ep_2,\ua_N}$ satisfies the bound \eqref{eq:recomb_err_bnd}. We choose $\ua_N=\ur_{N,\ep_1}$, where the reader will recall the definition of $\ur_{N,\ep_1}$ from \eqref{eq:r_def}.

Using the error bound \eqref{eq:recomb_err_bnd} with $\ua_N=\ur_{N,\ep_1}$, that $r_{i,\ep_1}\leq \ep_1$ by definition, the increasing property of $r\mapsto r|\ln r|$ for $r\in (0,e^{-1})$, and some algebra, we find that
\begin{equation}
\label{eq:conc_err_fin}
\begin{split}
\left|\mathrm{Error}_{\ep_2,\ur_{N,\ep_1}}\right| &\lesssim \|v\|_{LL(\R^2)}|\ln\ep_3|\paren*{\Fr_N(\ux_N,\mu)+N|\ln\ep_3| + C_pN^2\|\mu\|_{L^p(\R^2)}{\ep_3}^{\frac{2(p-1)}{p}}} \\
&\ph + \frac{\ep_1 N^2}{\ep_3}\paren*{\frac{\|v\|_{L^\infty(\R^2)}}{\ep_3} + \|v\|_{LL(\R^2)}|\ln\ep_1|} \\
&\ph + N^{2}\ep_1|\ln\ep_1| C_p\|v\|_{LL(\R^2)}\|\mu\|_{L^p(\R^2)}^{\frac{p}{2(p-1)}} \\
&\ph + N^{2}\|v\|_{L^\infty(\R^2)}\paren*{C_p\|\mu\|_{L^p(\R^2)}\ep_1^{\frac{p-2}{p}} + C_\infty \|\mu\|_{L^\infty(\R^2)}\ep_1 |\ln\ep_1| 1_{\geq \infty}(p)}.
\end{split}
\end{equation}

For the second term in the right-hand side of \eqref{eq:conc_RHS}, we use that $\|v_{\ep_2}\|_{LL(\R^2)}\leq \|v\|_{LL(\R^2)}$ by \cref{lem:conv_bnds} and that
\begin{equation}
(\nabla\g)(x-y) = -\frac{1}{2\pi|x-y|} \qquad \forall x\neq y\in\R^2
\end{equation}
in order to obtain the estimate
\begin{align}
&\left|\sum_{i=1}^N\int_{(\R^2)^2} \paren*{v_{\ep_2}(x)-v_{\ep_2}(y)}\cdot(\nabla\g)(x-y)d\d_{x_i}^{(r_{i,\ep_1})}(x)d\d_{x_i}^{(r_{i,\ep_1})}(y)\right| \nn\\
&\lesssim \|v_{\ep_2}\|_{LL(\R^2)}\sum_{i=1}^N \int_{(\R^2)^2}|\ln|x-y||d\d_{x_i}^{(r_{i,\ep_1})}(x)d\d_{x_i}^{(r_{i,\ep_1})}(y) \nn\\
&\lesssim \|v\|_{LL(\R^2)}\sum_{i=1}^N\tl{\g}(r_{i,\ep_1}),
\end{align}
where the ultimate inequality follows from \cref{lem:g_smear_si} and that $|\ln|x-y||=-\ln|x-y|$ for $x,y\in\supp(\d_{x_i}^{(r_{i,\ep_1})})$. Using estimate \eqref{eq:g_r_sum_bnd} of \cref{cor:grad_H}, we conclude that
\begin{equation}
\begin{split}
&\left|\sum_{i=1}^N\int_{(\R^2)^2} \paren*{v_{\ep_2}(x)-v_{\ep_2}(y)}\cdot(\nabla\g)(x-y)d\d_{x_i}^{(r_{i,\ep_1})}(x)d\d_{x_i}^{(r_{i,\ep_1})}(y)\right|\\
&\lesssim \|v\|_{LL(\R^2)}\paren*{\Fr_N(\ux_N,\mu) + N\paren*{|\ln\ep_1|+ C_p\|\mu\|_{L^p(\R^2)}N\ep_1^{\frac{2(p-1)}{p}}}}. \label{eq:conc_T2_fin}
\end{split}
\end{equation}

For the first term in the right-hand side of \eqref{eq:conc_RHS}, we use H\"older's inequality, \cref{lem:conv_bnds}, and the point-wise bound
\begin{equation}
|\comm{H_{N,\ur_{N,\ep_1}}^{\mu,\ux_N}}{H_{N,\ur_{N,\ep_1}}^{\mu,\ux_N}}_{SE}(x)| \lesssim |(\nabla H_{N,\ur_{N,\ep_1}}^{\mu,\ux_N})(x)|^2 \qquad \forall x\in\R^2,
\end{equation}
which is immediate from the definition \eqref{eq:se_ten_def} of the stress-energy tensor and Young's inequality for products, in order to obtain
\begin{align}
\left|\int_{\R^2} (\nabla v_{\ep_2})(x) :\comm{H_{N,\ur_{N,\ep_1}}^{\mu,\ux_N}}{H_{N,\ur_{N,\ep_1}}^{\mu,\ux_N}}_{SE}(x)dx\right| &\lesssim \|\nabla v_{\ep_2}\|_{L^\infty(\R^2)}\int_{\R^2} |(\nabla H_{N,\ur_{N,\ep_1}}^{\mu,\ux_N})(x)|^2dx \nn\\
&\lesssim |\ln\ep_2|\|v\|_{LL(\R^2)}\int_{\R^2} |(\nabla H_{N,\ur_{N,\ep_1}}^{\mu,\ux_N})(x)|^2dx.\footnotemark \label{eq:v_ep_Lip}
\end{align}
\footnotetext{This inequality is precisely where we use that $v_{\ep_2}$ is Lipschitz.}
Now using estimate \eqref{eq:grad_H_r_bnd} from \cref{cor:grad_H}, we conclude that
\begin{equation}
\label{eq:conc_T1_fin}
\begin{split}
&\left|\int_{\R^2}(\nabla v_{\ep_2})(x) :\comm{H_{N,\ur_{N,\ep_1}}^{\mu,\ux_N}}{H_{N,\ur_{N,\ep_1}}^{\mu,\ux_N}}_{SE}(x)dx\right|\\
&\lesssim |\ln\ep_2|\|v\|_{LL(\R^2)}\paren*{\Fr_N(\ux_N,\mu) + N|\ln\ep_1| + C_pN^2\|\mu\|_{L^p(\R^2)}\ep_1^{\frac{2(p-1)}{p}}}. 
\end{split}
\end{equation}

Combining estimates \eqref{eq:conc_err_fin}, \eqref{eq:conc_T2_fin}, and \eqref{eq:conc_T1_fin}, we obtain that
\begin{equation}
\label{eq:conc_v_ep_fin}
\begin{split}
&\left|\int_{(\R^2)^2\setminus\D_2} \paren*{v_{\ep_2}(x)-v_{\ep_2}(y)}\cdot(\nabla\g)(x-y)d(\sum_{i=1}^N\d_{x_i}-N\mu)(x)d(\sum_{i=1}^N\d_{x_i}-N\mu)(y)\right| \\
&\lesssim \paren*{1+|\ln\ep_2|}\|v\|_{LL(\R^2)}\paren*{\Fr_N(\ux_N,\mu) + N\paren*{|\ln\ep_1|+ C_p\|\mu\|_{L^p(\R^2)}N\ep_1^{\frac{2(p-1)}{p}}}}\\
&\ph +|\ln\ep_3|\|v\|_{LL(\R^2)}\paren*{\Fr_N(\ux_N,\mu)+N\paren*{|\ln\ep_3| + C_pN\|\mu\|_{L^p(\R^2)}{\ep_3}^{\frac{2(p-1)}{p}}}} \\
&\ph + \frac{\ep_1 N^2}{\ep_3}\paren*{\frac{\|v\|_{L^\infty(\R^2)}}{\ep_3} + \|v\|_{LL(\R^2)}|\ln\ep_1|} + N^{2}\ep_1|\ln\ep_1| C_p\|v\|_{LL(\R^2)}\|\mu\|_{L^p(\R^2)}^{\frac{p}{2(p-1)}} \\
&\ph + N^{2}\|v\|_{L^\infty(\R^2)}\paren*{C_p\|\mu\|_{L^p(\R^2)}\ep_1^{\frac{p-2}{p}} + C_\infty \|\mu\|_{L^\infty(\R^2)}\ep_1 |\ln\ep_1| 1_{\geq \infty}(p)}.
\end{split}
\end{equation}
Combining estimate \eqref{eq:conc_v_ep_fin} with \cref{lem:kprop_error} and using that $\ep_3>\ep_2>2\ep_1$ in order to simplify, we finally conclude that
\begin{equation}
\begin{split}
&\left|\int_{(\R^2)^2\setminus\D_2} \paren*{v(x)-v(y)}\cdot(\nabla\g)(x-y)d(\sum_{i=1}^N\d_{x_i}-N\mu)(x)d(\sum_{i=1}^N\d_{x_i}-N\mu)(y)\right| \\
&\leq \|v\|_{LL(\R^2)}|\ln\ep_2| |\Fr_N(\ux_N,\mu)| + \|v\|_{LL(\R^2)}N|\ln\ep_1| |\ln\ep_2| + C_pN^2\|v\|_{LL(\R^2)}\ep_3^{\frac{2(p-1)}{p}}|\ln\ep_3| \\
&\ph + N^2\paren*{\frac{\ep_1\|v\|_{L^\infty(\R^2)}}{\ep_3^2}+ \frac{\ep_2|\ln\ep_2|\|v\|_{LL(\R^2)}}{\ep_3} + \ep_2|\ln\ep_2| C_p\|v\|_{LL(\R^2)}\|\mu\|_{L^p(\R^2)}^{\frac{p}{2(p-1)}}} \\
&\ph + N^{2}\|v\|_{L^\infty(\R^2)}\paren*{C_p\|\mu\|_{L^p(\R^2)}\ep_1^{\frac{p-2}{p}} + C_\infty \|\mu\|_{L^\infty(\R^2)}\ep_1 |\ln\ep_1| 1_{\geq \infty}(p)},
\end{split}
\end{equation}
which completes the proof of \cref{prop:key}.

\section{Proof of Main Results}
\label{sec:MR}
In this section, we give the proofs of \cref{thm:main} and \cref{cor:main} using an energy estimate, \cref{prop:key}, and \cref{prop:bes_conv}, as sketched in the introduction.

\subsection{Modulated Energy Derivative}\label{ssec:MR_EI}
We begin by showing that the function $t\mapsto \Fr_N^{avg}(\ux_N(t),\omega(t))$, where we recall that $\ux_N(t)$ is the solution to the point vortex model \eqref{eq:PVM} and $\omega(t)$ is the solution to the Euler equation \eqref{eq:Eul}, respectively, at time $t$, is locally Lipschitz continuous and compute its time derivative. The identity (see \eqref{eq:en_ineq} below) is well-known, but a rigorous proof taking into consideration that we are only working with weak solutions to the Euler equation \eqref{eq:Eul} does not seem to exist in the literature. Thus, we provide a proof in the interests of completeness.

\begin{lemma}[Modulated energy derivative]
\label{lem:en_ineq}
Fix $N\in\N$. Let $\ux_N\in C^\infty([0,\infty);\R^2\setminus\D_N)$ be a solution to the system \eqref{eq:PVM}, and let $\omega \in L^\infty([0,T]; L^1(\R^2)\cap L^p(\R^2))$, for some $2<p\leq\infty$ be a weak solution to the equation \eqref{eq:Eul}, satisfying the logarithmic bound \eqref{eq:log_grow} uniformly in time. Then $\Fr_N^{avg}(\ux_N,\mu): [0,\infty)\rightarrow [0,\infty)$ is locally Lipschitz continuous, and for a.e. $t\in [0,T]$, we have the identity
\begin{equation}
\label{eq:en_ineq}
\frac{d}{dt}\Fr_N^{avg}(\ux_N(t),\mu(t)) = \int_{(\R^2)^2\setminus\D_2} (\nabla\g)(x-y)\cdot\paren*{u(t,x)-u(t,y)}d(\omega_N-\omega)(t,x)d(\omega_N-\omega)(t,y),
\end{equation}
where $u$ is the velocity field associated to $\omega$ through the Biot-Savart law and $\omega_N\coloneqq \frac{1}{N}\sum_{i=1}^N \d_{x_i}$ is the empirical measure associated to $\ux_N$.
\end{lemma}
\begin{proof}
Fix $N\in\N$. By the Fubini-Tonelli theorem, we see that
\begin{align}
\Fr_N(\ux_N(t),\omega(t)) &= N^2\int_{\R^2}(\g\ast\omega(t))(x)\omega(t,x)dx + \sum_{1\leq i\neq j\leq N} \g(x_i(t)-x_j(t)) \nn\\
&\ph - 2N\sum_{i=1}^N (\g\ast\omega(t))(x_i(t)) \nn\\
&\eqqcolon \mathrm{Term}_1(t) + \mathrm{Term}_2(t) - \mathrm{Term}_{3}(t).
\end{align}
By \cref{lem:e_con} and \cref{rem:ham_con}, respectively,
\begin{equation}
\frac{d\mathrm{Term}_1}{dt} \equiv 0 \quad \text{and} \quad \frac{d\mathrm{Term}_2}{dt} \equiv 0.
\end{equation}

For $\mathrm{Term}_3(t)$, the reader can check from the Duhamel formula \eqref{eq:Eul_Duh}, the distributional calculus, and the Fubini-Tonelli theorem that
\begin{equation}
\label{eq:pt_id}
(\p_t(\g\ast\omega))(t) = -\nabla\g\ast\cdot (u(t)\omega(t)), \quad \text{a.e. } t\in [0,T].
\end{equation}
Since $u\omega\in L^\infty([0,T]; L^1(\R^2;\R^2)\cap L^p(\R^2;\R^2))$, it follows from \cref{lem:pot_bnds} that
\begin{equation}
\label{eq:pt_fs}
\p_t(\g\ast\omega)\in L^\infty([0,T]; C^{\al(p)}(\R^2)),
\end{equation}
for some $\al(p)>0$. Let $\rho\in C_c^\infty(0,T)$ be a test function. Then by definition of distributional derivative,
\begin{equation}
\int_{0}^T \rho(s) \p_t( (\g\ast\omega)(x_i(\cdot)))(s)ds = -\int_{0}^T (\p_t\rho)(s) (\g\ast\omega(s))(x_i(s))ds.
\end{equation}
Let $\gamma\in C_c^\infty(\R)$ be a standard bump function, and for $\ep>0$, set $\gamma_\ep(s)\coloneqq \ep^{-1}\gamma(s/\ep)$. We mollify $\g\ast\omega$ in time by defining
\begin{equation}
(\gamma_\ep \ast_t (\g\ast\omega))(s,x) \coloneqq \int_{\R} \gamma_\ep(s-s')(\g\ast\omega(s'))(x)ds'.
\end{equation}
By \cref{lem:pot_bnds}, we have that this temporally mollified function belongs to $C_{loc}^1(\R\times\R^2)$. Also, note from identity \eqref{eq:pt_id} that
\begin{equation}
\label{eq:pt_moll_id}
\p_t(\gamma_\ep \ast_t (\g\ast\omega)) = \gamma_\ep\ast_t(\nabla\g\ast\cdot (u\omega)).
\end{equation}
It now follows from \eqref{eq:pt_fs} and dominated convergence that
\begin{equation}
-\int_{0}^T (\p_t\rho)(s) (\g\ast\omega(s))(x_i(s))ds = -\lim_{\ep\rightarrow 0^+} \int_0^T (\p_t\rho)(s)(\gamma_\ep \ast_t (\g\ast\omega))(s,x_i(s)).
\end{equation}
Since $x_i([0,T])$ is a compact subset of $\R^2$, it follows from integration by parts, the chain rule, and the identity \eqref{eq:pt_moll_id} that
\begin{equation}
\begin{split}
-\int_0^T (\p_t\rho)(s)(\gamma_\ep \ast_t (\g\ast\omega))(s,x_i(s)) &= -\int_0^T\rho(s)(\gamma_\ep\ast_t(\nabla\g\ast\cdot (u\omega)))(s,x_i(s))ds \\
&\ph+ \int_0^T\rho(s)(\gamma_\ep\ast_t (\nabla\g\ast\omega))(s,x_i(s))\cdot \frac{1}{N}\sum_{{1\leq j\leq N}\atop {i\neq j}} (\nabla^\perp\g)(x_i(s)-x_j(s))ds.
\end{split}
\end{equation}
Taking the limit as $\ep\rightarrow 0^+$ of both sides and applying dominated convergence again, we conclude that
\begin{equation}
\begin{split}
\int_{0}^T \rho(s) \p_t( (\g\ast\omega)(x_i(\cdot)))(s)ds &= -\int_0^T\rho(s)(\nabla\g\ast\cdot (u(s)\omega(s)))(x_i(s))ds \\
&\ph + \int_0^T \rho(s)(\nabla\g\ast\omega(s))(x_i(s))\cdot \frac{1}{N}\sum_{{1\leq j\leq N}\atop {i\neq j}} (\nabla^\perp\g)(x_i(s)-x_j(s))ds.
\end{split}
\end{equation}
Since $\rho\in C_c^\infty(0,T)$ was an arbitrary test function, we conclude the distributional identity
\begin{equation}
\p_t((\g\ast\omega)(x_i(\cdot)))(s) = -(\nabla\g\ast\cdot (u(s)\omega(s)))(x_i(s)) + (\nabla\g\ast\omega(s))(x_i(s))\cdot \frac{1}{N}\sum_{{1\leq j\leq N}\atop {i\neq j}} (\nabla^\perp\g)(x_i(s)-x_j(s)).
\end{equation}
After a little bookkeeping, we conclude that
\begin{equation}
\begin{split}
\frac{d\mathrm{Term}_3}{dt}(t) &= 2N\sum_{i=1}^N \Bigg(-\int_{\R^2}(\nabla\g)(x_i(t)-y)\cdot u(t,y)\omega(t,y)dy \\
&\ph \hspace{25mm} + \frac{1}{N}\int_{\R^2} \sum_{{1\leq j\leq N}\atop{j\neq i}} (\nabla\g)(x_i(t)-y)\cdot (\nabla^\perp\g)(x_i(t)-x_j(t))\omega(t,y)dy\Bigg),
\end{split}
\end{equation}
which implies that
\begin{equation}
\begin{split}
\frac{d}{dt}\Fr_N(\ux_N(t),\omega(t)) &= -2N\sum_{i=1}^N\Bigg(\int_{\R^2}(\nabla\g)(x_i(t)-y)\cdot u(t,y)\omega(t,y)dy \\
&\ph \hspace{25mm}- \frac{1}{N}\int_{\R^2} \sum_{{1\leq j\leq N}\atop{j\neq i}} (\nabla\g)(x_i(t)-y)\cdot (\nabla^\perp\g)(x_i(t)-x_j(t))\omega(t,y)dy\Bigg). \label{eq:F_td}
\end{split}
\end{equation}

Next, using the Fubini-Tonelli theorem, we can write
\begin{align}
2N\sum_{i=1}^N  \int_{\R^2}(\nabla\g)(x_i(t)-y)\cdot u(t,y)\omega(t,y)dy &= 2N^2\int_{(\R^2)^2} (\nabla\g)(x-y)\cdot u(t,y)d\omega(t,y)d\omega_N(t,x). \label{eq:sing_sym_app}
\end{align}
Similarly, writing $\nabla^\perp\g = \J\nabla\g$ (recall the definition of the $2\times2$ matrix $\J$ from \eqref{eq:J_def}), so that by the anti-symmetry of $\J$,
\begin{equation}
(\nabla\g)(x_i(t)-y) \cdot (\nabla^\perp\g)(x_i(t)-x_j(t)) = (\nabla^\perp\g)(x_i(t)-y)\cdot(\nabla\g)(x_i(t)-x_j(t)),
\end{equation}
we observe that
\begin{align}
&-2\sum_{1\leq i\neq j\leq N}\int_{\R^2} (\nabla\g)(x_i(t)-y)\cdot (\nabla^\perp\g)(x_i(t)-x_j(t))\omega(t,y)dy\nn\\
&= 2\sum_{1\leq i\neq j\leq N} u(x_i(t))\cdot (\nabla\g)(x_i(t)-x_j(t)) \nn\\
&=2N^2\int_{(\R^2)^2\setminus\D_2} (\nabla\g)(x-y)\cdot u(t,x)d\omega_N(t,x)d\omega_{N}(t,y). \label{eq:dub_sym_app}
\end{align}

Swapping $x$ and $y$ and using that $(\nabla\g)(x-y)=-(\nabla\g)(y-x)$, we find that
\begin{equation}
\eqref{eq:dub_sym_app} = N^2\int_{(\R^2)^2\setminus\D_2} (\nabla\g)(x-y)\cdot\paren*{u(t,x)-u(t,y)}d\omega_N(t,x)d\omega_N(t,y). \label{eq:dub_sym}
\end{equation}
Similarly,
\begin{equation}
\eqref{eq:sing_sym_app} = -2N^2\int_{(\R^2)^2\setminus\D_2} (\nabla\g)(x-y)\cdot u(t,x)d\omega(t,x)d\omega_N(t,y). \label{eq:sing_sym}
\end{equation}
Now observing the cancellations
\begin{equation}
\begin{split}
&N^2\int_{(\R^2)^2\setminus\D_2} (\nabla\g)(x-y)\cdot u(t,x) d\omega(t,y)d\omega_N(t,x)\\
&= -N^2\int_{(\R^2)^2\setminus\D_2}(\nabla\g)(x-y)\cdot u(t,y)d\omega(t,x)d\omega_N(t,y) \\
&=N^2\int_{(\R^2)^2\setminus\D_2} (\nabla\g)(x-y)\cdot\paren*{u(t,x)-u(t,y)}d\omega(t,x)d\omega_N(t,y)\\
&=0,
\end{split}
\end{equation}
which follows from the Fubini-Tonelli theorem and the orthogonality identity
\begin{equation}
(\nabla\g\ast\omega)(t,x) \perp u(t,x) \qquad \forall (t,x)\in [0,T]\times\R^2,
\end{equation}
and also the cancellation
\begin{equation}
0=N^2\int_{(\R^2)^2}(\nabla\g)(x-y)\cdot\paren*{u(t,x)-u(t,y)}d\omega(t,x)d\omega(t,y),
\end{equation}
which follows from considerations of anti-symmetry, it follows from \eqref{eq:F_td}, \eqref{eq:dub_sym}, and \eqref{eq:sing_sym} that
\begin{equation}
\frac{d}{dt}\Fr_N(\ux_N(t),\omega(t)) = N^2\int_{(\R^2)^2\setminus\D_2} (\nabla\g)(x-y)\cdot\paren*{u(t,x)-u(t,y)}d(\omega-\omega_N)(t,x)d(\omega-\omega_N)(t,y)s.
\end{equation}
Dividing both sides by $N^2$ completes the proof of the lemma.
\end{proof}

\subsection{Proof of \cref{thm:main}} \label{ssec:MR_thm}
We now are prepared to prove \cref{thm:main}, which we do in this subsection. By the fundamental theorem of calculus for Lebesgue integrals, \cref{lem:en_ineq}, and the triangle inequality,
\begin{equation}
\begin{split}
|\Fr_N^{avg}(\ux_N(t),\omega(t))| &\leq |\Fr_N^{avg}(\ux_N^0,\omega^0)| \\
&\ph + \int_0^t\left|\int_{(\R^2)^2\setminus\D_2} (\nabla\g)(x-y)\cdot\paren*{u(s,x)-u(s,y)}d(\omega-\omega_N)(s,x)d(\omega-\omega_N)(s,y)\right|ds.
\end{split}
\end{equation}
Applying \cref{prop:key} with $p=\infty$ to the integrand in the right-hand side of the preceding inequality point-wise in $s$ and using that
\begin{equation}
\|u\|_{L^\infty([0,\infty); L^\infty(\R^2))} \lesssim \|\omega^0\|_{L^\infty(\R^2)}^{1/2} \quad \text{and} \quad \|u\|_{L^\infty([0,\infty); LL(\R^2))} \lesssim \|\omega^0\|_{L^\infty(\R^2)},
\end{equation}
by \Cref{lem:Linf_RP,lem:pot_bnds}, respectively, and conservation of the $L^\infty$ norm, we find that there exists a constant $C_1>0$ such that
\begin{equation}
\label{eq:ep1_ep2_pre}
\begin{split}
&|\Fr_N^{avg}(\ux_N(t),\omega(t))|\\
&\leq |\Fr_N^{avg}(\ux_N^0,\omega^0)| + C_1\|\omega^0\|_{L^\infty(\R^2)}^{3/2}\int_0^t \ep_2(s)|\ln\ep_2(s)| ds\\
&\ph + C_1\|\omega^0\|_{L^\infty(\R^2)} \int_0^t \paren*{|\ln\ep_2(s)| |\Fr_N^{avg}(\ux_N(s),\omega(s))| + \frac{|\ln\ep_1(s)| |\ln\ep_2(s)|}{N} +\ep_3(s)^2|\ln\ep_3(s)|}ds \\
&\ph +C_1\int_0^t\paren*{\frac{\ep_1(s)\|\omega^0\|_{L^\infty(\R^2)}^{1/2}}{\ep_3(s)^2}+ \frac{\ep_2(s)|\ln\ep_2(s)|\|\omega^0\|_{L^\infty(\R^2)}}{\ep_3(s)}}ds,
\end{split}
\end{equation}
where $\ep_1,\ep_2,\ep_3: [0,\infty)$ are measurable functions such that $1\gg\ep_3(s)>\ep_2(s)\geq 2\ep_1(s)>0$ for every $s\in [0,\infty)$. We choose $\ep_1,\ep_2$ as follows:
\begin{align}
\ep_1(s) &\coloneqq \ep_3(s)^3, \label{eq:ep1_choice}\\
\ep_2(s) &\coloneqq \ep_3(s)^2. \label{eq:ep2_choice}
\end{align}
Substituting these choices for $\ep_1(s),\ep_2(s)$ into the right-hand side of inequality \eqref{eq:ep1_ep2_pre} and simplifying, we find that
\begin{equation}
\label{eq:ep3_pre_max}
\begin{split}
|\Fr_N^{avg}(\ux_N(t),\omega(t))| &\leq |\Fr_N^{avg}(\ux_N^0,\omega^0)| +  C_2\|\omega^0\|_{L^\infty(\R^2)}\int_0^t \paren*{|\ln\ep_3(s)| |\Fr_N^{avg}(\ux_N(s),\omega(s))| + \frac{|\ln\ep_3(s)|^2}{N}}ds\\
&\ph + C_2\|\omega^0\|_{L^\infty(\R^2)}\int_0^t \ep_3(s)|\ln\ep_3(s)| ds\\
&\ph + C_2\int_0^t \paren*{\ep_3(s)\|\omega^0\|_{L^\infty(\R^2)}^{1/2} + \ep_3(s)^2|\ln\ep_3(s)|\paren*{\|\omega^0\|_{L^\infty(\R^2)} +\|\omega^0\|_{L^\infty(\R^2)}^{3/2}}}ds,
\end{split}
\end{equation}
where $C_2\geq C_1$ is an absolute constant. Note that if we define $\Fbr_N^{avg}:[0,\infty)\rightarrow[0,\infty)$ by
\begin{equation}
\Fbr_{N}^{avg}(\ux_N(s),\omega(s)) \coloneqq \sup_{0\leq s'\leq s}|\Fr_{N}^{avg}(\ux_N(s'),\omega(s'))|,
\end{equation}
which is tautologically a nondecreasing function, then it follows from \eqref{eq:ep3_pre_max} that
\begin{equation}
\label{eq:ep3_pre}
\begin{split}
\Fbr_{N}^{avg}(\ux_N(t),\omega(t)) &\leq |\Fr_N^{avg}(\ux_N^0,\omega^0)| +  C_2\|\omega^0\|_{L^\infty(\R^2)}\int_0^t \paren*{|\ln\ep_3(s)| \Fbr_{N}^{avg}(\ux_N(s),\omega(s)) + \frac{|\ln\ep_3(s)|^2}{N}}ds\\
&\ph + C_2\|\omega^0\|_{L^\infty(\R^2)}\int_0^t \ep_3(s)|\ln\ep_3(s)| ds\\
&\ph + C_2\int_0^t \paren*{\ep_3(s)\|\omega^0\|_{L^\infty(\R^2)}^{1/2} + \ep_3(s)^2|\ln\ep_3(s)|\paren*{\|\omega^0\|_{L^\infty(\R^2)} +\|\omega^0\|_{L^\infty(\R^2)}^{3/2}}}ds.
\end{split}
\end{equation}

It remains to choose $\ep_3$, which we do now in a piece-wise fashion:
\begin{equation}
\label{eq:ep3_choice}
\forall s\in [0,\infty), \qquad \ep_3(s) \coloneqq \begin{cases} \frac{\ln N}{N}, & {\Fbr_N^{avg}(\ux_N(s),\omega(s)) \leq \frac{\ln N}{N}} \\ \Fbr_N^{avg}(\ux_N(s),\omega(s)) , & {\frac{\ln N}{N} < \Fbr_N^{avg}(\ux_N(s),\omega(s)) < e^{-1}} \\ e^{-1} , & {\Fbr_N^{avg}(\ux_N(s),\omega(s))\geq e^{-1}} \end{cases}.
\end{equation}
Substituting this choice for $\ep_3(s)$ into the right-hand side of inequality \eqref{eq:ep3_pre} and partitioning the interval $[0,t]$ according to the values of $\Fbr_N^{avg}(\ux_N(s),\omega(s))$, as in the definition \eqref{eq:ep3_choice}, we find that
\begin{equation}
\label{eq:ep3_post}
\begin{split}
\Fbr_N^{avg}(\ux_N(t),\omega(t)) &\leq |\Fr_N^{avg}(\ux_N^0,\omega^0)|  + \frac{C_3\|\omega^0\|_{L^\infty(\R^2)}^{1/2}(1+\|\omega^0\|_{L^\infty(\R^2)}) |\ln N|^2 t}{N} \\
&\ph + C_3\|\omega^0\|_{L^\infty(\R^2)}^{1/2}\paren*{1+\|\omega^0\|_{L^\infty(\R^2)}}\int_0^t \Fbr_N^{avg}(\ux_N(s),\omega(s)) |\ln \Fbr_N^{avg}(\ux_N(s),\omega(s))|ds \\
&\ph + C_3\|\omega^0\|_{L^\infty(\R^2)}^{1/2}\paren*{1 + \|\omega^0\|_{L^\infty(\R^2)}} \int_0^t 1_{>e^{-1}}(\Fbr_N^{avg}(\ux_N(s),\omega(s)))ds,
\end{split}
\end{equation}
where $C_3\geq C_2$ is an absolute constant.

We now want to close the estimate \eqref{eq:ep3_post} by using the Osgood lemma (recall \cref{lem:Os}).  It is not hard to check from the continuity of the map $s\mapsto \Fr_N^{avg}(\ux_N(s),\omega(s))$ that the map $s\mapsto \Fbr_N^{avg}(\ux_N(s),\omega(s))$ is also continuous (see the proof of \cite[Lemma 5.4]{Rosenzweig2019_PV}). Now fix a time $t\in (0,\infty)$. Let $N\in\N$ be sufficiently large so that
\begin{equation}
\begin{split}
\label{eq:t_cond}
C_3\|\omega^0\|_{L^\infty(\R^2)}^{1/2}\paren*{1+\|\omega^0\|_{L^\infty(\R^2)}}t &<\ln\ln\paren*{|\Fr_N^{avg}(\ux_N^0,\omega^0)|  + \frac{C_3t\|\omega^0\|_{L^\infty(\R^2)}^{1/2}(1+\|\omega^0\|_{L^\infty(\R^2)})|\ln N|^2}{N}}^{-1}.
\end{split}
\end{equation}
By continuity, there exists a minimal time $0<t_N^*\leq t$ (we adopt the convention that $t_N^*=t$ if no such time exists) such that
\begin{equation}
\label{eq:tN*_def}
\Fbr_N^{avg}(\ux_N(s),\omega(s)) < e^{-1}, \enspace \forall 0\leq s< t_N^* \qquad \text{and} \qquad \Fbr_N^{avg}(\ux_N(t_N^*),\omega(t_N^*)) = e^{-1}.
\end{equation}
Applying \cref{lem:Os} with modulus of continuity $r \mapsto r\ln(1/r)$ on the interval $[0,e^{-1}]$ and using \cref{rem:Os_ex_log}, we find from condition \eqref{eq:t_cond} that for every $0\leq s\leq t_N^*$,
\begin{align}
&\Fbr_N^{avg}(\ux_N(s),\omega(s))\nn\\
&\leq \exp\Bigg(-\exp\Bigg(\ln\ln\paren*{|\Fr_N^{avg}(\ux_N^0,\omega^0)|  + \frac{C_3s\|\omega^0\|_{L^\infty(\R^2)}^{1/2}(1+\|\omega^0\|_{L^\infty(\R^2)})|\ln N|^2}{N}}^{-1} \nn\\
&\ph \hspace{32mm} - C_3s\|\omega^0\|_{L^\infty(\R^2)}^{1/2}\paren*{1+\|\omega^0\|_{L^\infty(\R^2)}}\Bigg)\Bigg) \nn\\
&=\paren*{|\Fr_N^{avg}(\ux_N^0,\omega^0)|  + \frac{C_3s\|\omega^0\|_{L^\infty(\R^2)}^{1/2}(1+\|\omega^0\|_{L^\infty(\R^2)})|\ln N|^2}{N}}^{e^{-C_3s\|\omega^0\|_{L^\infty(\R^2)}^{1/2}(1+\|\omega^0\|_{L^\infty(\R^2)})}}.
\end{align}
It then follows from \eqref{eq:t_cond} that the expression in the ultimate line is $<e^{-1}$, which implies that $t_N^*=t$. Thus, the proof of \cref{thm:main} is complete.

\subsection{Proof of \cref{cor:main}}\label{ssec:MR_cor}
We now show how to obtain \cref{cor:main} from \cref{thm:main}. Fix $s<-1$ and $T>0$, and let $N\in\N$ be sufficiently large so as to satisfy the condition \eqref{eq:N_cond}. Then applying estimate \eqref{eq:Hs_conv} of \cref{prop:bes_conv} with $p=\infty$ for each $t\in [0,T]$ fixed and also using conservation of the $L^\infty$ norm, followed by applying \cref{thm:main}, we find that there exists constants $C>0$, such that
\begin{equation}
\begin{split}
&\|\omega(t)-\omega_N(t)\|_{H^{s}(\R^2)} \\
&\lesssim_{s} \paren*{|\Fr_N^{avg}(\ux_N^0,\omega^0)| + \frac{Ct\|\omega^0\|_{L^\infty(\R^2)}^{1/2}(1+\|\omega^0\|_{L^\infty(\R^2)}) |\ln N|^2}{N}}^{e^{-Ct\|\omega^0\|_{L^\infty(\R^2)}^{1/2}(1+\|\omega^0\|_{L^\infty(\R^2)})}} \\
&\ph + \frac{|\ln N|^{1/2} + \|\omega^0\|_{L^\infty(\R^2)}}{N^{1/2}}.
\end{split}
\end{equation}
Taking the supremum over $t\in [0,T]$ of both sides yields \eqref{eq:cor_main_est}. The assertion regarding weak-* convergence in $\M(\R^2)$ follows by using \eqref{eq:M_conv}. Thus, the proof of \cref{cor:main} is complete.

\bibliographystyle{siam}
\bibliography{PointVortex}
\end{document}